\documentclass[12pt,a4paper]{article}
\usepackage[utf8]{inputenc}
\usepackage[T1]{fontenc}
\usepackage{lmodern}
\usepackage{braket}
\usepackage{amsmath,amssymb,amsthm}
\usepackage{mathrsfs}
\usepackage{enumitem}
\usepackage{graphicx}
\usepackage{caption}
\usepackage{subcaption}
\usepackage{booktabs}
\usepackage{tabularx}
\usepackage{mathtools}
\usepackage{geometry}
\usepackage{cite}
\usepackage{hyperref}
\hypersetup{
	colorlinks=true,
	linkcolor=blue,
	citecolor=blue,
	urlcolor=blue,
}
\usepackage{tikz}
\usetikzlibrary{decorations.pathreplacing}

\usepackage{booktabs}   
\usepackage{array}      

\newcolumntype{P}[1]{>{\centering\arraybackslash}p{#1}}
\newcolumntype{L}[1]{>{\raggedright\arraybackslash}p{#1}}
\newcolumntype{C}[1]{>{\centering\arraybackslash}p{#1}}

\newtheorem{theorem}{Theorem}[section] 
\newtheorem{definition}{Definition}[section]
\newtheorem{lemma}{Lemma}[section]
\newtheorem{property}{Property}[section]
\newtheorem{remark}{Remark}[section]
\newtheorem{example}{Example}[section]
\newtheorem{proposition}{Proposition}[section]
\newtheorem{corollary}{Corollary}[section]

\newcommand{\C}{\mathbb C}

\title{
Symmetry Packaging I:
Irreducible Representation Blocks, Superselection, and Packaged Entanglement in Quantum Field Theory
}

\author{
	Rongchao Ma \\
	\textit{Department of Physics, University of Alberta, Edmonton, Canada}
}

\date{\today}

\begin{document}

	\maketitle

	\begin{abstract}		
		We introduce the concept of symmetry packaging for quantum field excitations:
		in a quantum field theory with a gauge group $G$, every local creation operator carries its full set of internal quantum numbers (IQNs) as a single irreducible $G$-block and forbids any partial factorization.
		We elevate this observation to a symmetry packaging principle, which asserts that packets of IQNs remain intact throughout all physical processes.
		We analyze a quantum-field excitation in six successive stages:
		(1) particle creation/annihilation,
		(2) hybridization with gauge-blind external degrees of freedom (DOFs),
		(3) tensor-product assembly,
		(4) isotypic decomposition,
		(5) packaged superposition/entanglement, and
		(6) local Gauge-Invariance constraint.		
		We show that packaging survives every stage and culminates in a gauge-invariant physical Hilbert space.		
		These stages unfold within a three packaging layer hierarchy (raw-Fock $\to$ isotypic $\to$ physical) with distinct packaging characters.		
		Packaging alone reproduces familiar charge-superselection rules and, within any fixed-charge sector, admits a new class of packaged entangled states where internal and external DOFs are inseparably locked.
		We derive necessary and sufficient conditions for such superpositions and show that packaged irreps behave as noise-protected logical qudits.		
		This framework unifies representation theory, superselection, and entanglement under a single mathematical roof and provides a roadmap for constructing and manipulating packaged states in any gauge theory.					
	\end{abstract}

	\tableofcontents

	\section{Introduction}

	Symmetry \cite{Weyl1952}, representation theory \cite{Wigner1959}, and entanglement \cite{Agullo2023} are important concepts in modern \textbf{quantum field theory (QFT)}.
	Whenever a field carries a nontrivial internal symmetry (finite \cite{Serre1977} or compact \cite{Nakahara2003,Georgi2018}) group $G$ (for example, the color group $\mathrm{SU}(3)$ of QCD \cite{Gross1973,Politzer1973} or the electroweak $\mathrm{SU}(2)\!\times\!\mathrm{U}(1)$) \cite{Glashow1959,SalamWard1959,Weinberg1967}, each local excitation simultaneously possesses many \textbf{internal quantum numbers (IQNs)}: charge, isospin, color, flavor, and so on.
	Although we routinely speak of these IQNs individually, a closer inspection shows that no local process can create, measure, or annihilate them individually:	
	the local gauge invariance \cite{Feynman1949,Yang1954,Utiyama1956,Weinberg1967} forces all internal \textbf{degrees of freedom (DOFs)} to appear and to propagate as a single, inseparable \textbf{irreducible representation (irrep)} \cite{Weyl1925,Wigner1939} block of $G$.	
	We call this observation the Symmetry Packaging Principle:
	\textit{Every local operator that can create (or destroy) an excitation carries exactly one irrep $V_{\lambda}$ of $G$ locked by local gauge invariance and no physical process can split $V_{\lambda}$ into smaller pieces}.

	This work is motivated by our early study on packaged entangled states \cite{Ma2017} in which all relevant IQNs are inseparably entangled under the charge conjugation operator $\hat{C}$.
	This means that symmetry packaging applies not only to single particles but also to multi-particle states and their entanglement.	
	Quantum information science seeks to exploit every available DOF as a robust information carrier \cite{NielsenChuang}.	
	However, conventional quantum communication protocols use external DOFs and ignore IQNs because superselection rules have so far limited the use of IQNs in communication protocols.	
	The ultimate goal is to teleport a particle's entire quantum identity, which requires the full set of quantum DOFs that can completely identify the particle.
	However, the external DOFs are dynamical variables that can be easily altered by the environment and cannot by themselves distinguish one particle from another.
	Instead, one needs intrinsic labels that remain invariant.
	IQNs are exactly such static identifiers.
	Since the Hilbert space of a gauge theory is not a simple tensor product, standard notions of entanglement break down.
	This calls for a unified framework both to describe the IQN entanglement and resolve the non-factorizable Hilbert space structure.
	Bridging that gap is one of the goals of this work.

	In Ref. \cite{Ma2017}, we introduced the packaged entangled state in a heuristic way and not based on field theory.
	Here we start from the first principles of quantum field theory \cite{PeskinSchroeder,WeinbergQFTI} to explore the effect of packaging on quantum field excitations.
	We describe the quantum field excitation in six successive stages and three packaging layers.
	These irreducible blocks survive six successive operations and finally lead to a gauge-invariant \cite{Cornwall1974,LavelleMcMullan1997} packaged subspace.
	Packaging sharpens earlier statements about complete sets of commuting observables \cite{StreaterWightman} and dressing transformations \cite{Dirac1950} by proving that no operator, local or non-local, can create a partial irrep.
	We prove:	
	(1) Superselection and confinement \cite{WWW1952,DHR1971,DHR1974,StreaterWightman}.
	Packaging alone reproduces familiar charge-superselection rules: superselection emerges because coherent superpositions between distinct packaged irreps are forbidden, while confinement is the statement that only packaged singlets of $G$ survive the gauge projector.	
	(2) Entanglement structure \cite{NielsenChuang}.
	A new class of packaged entangled states forms within a fixed net-charge sector in which the Schmidt decomposition is taken irrep-by-irrep.
	In a fully hybrid-packaged entangled state, measuring even an external observable collapses both the external and internal correlations simultaneously.
	(3) Packaged superposition conditions.
	We give necessary and sufficient conditions for packaged superpositions and identify them as symmetry-protected logical qudits \cite{Kitaev2003,HastingsWen2005}.

	The packaging framework clarifies how excitations arise in QFT and reproduces familiar high energy physics phenomena.
	For example, 
	correlations in electric charge \cite{WWW1952,DHR1971,DHR1974},
	flavor entanglement in neutral mesons \cite{Abe2001,Aubert2002,Go2007,Blasone2009,GellMann2019}, 
	color correlations in quark-antiquark systems \cite{Brandelik1979,Dumitru2022}, Bell inequalities in Higgs boson decays \cite{Barr2022},
	entanglement signatures in top-quark pairs and in electroweak bosons \cite{TheATLASCollaboration2024,TheCMSCollaboration2024},
	and nontrivial quantum correlations in bosons, fermions \cite{Li2001,Zanardi2002,Eckert2002}
	and generic bipartite systems \cite{Li2014}.
	All these manifest that symmetry packaging governs how relativistic particles/antiparticles are created and annihilated.

	The present paper is the first part of the symmetry packaging programme.
	Our goal is to reveal how IQN packaging arises from first principles and how it governs quantum field excitations.	
	We organize the present paper as follows:	
	Section \ref{SEC:Overview} overviews the six stages and three layers of quantum field excitations.	
	Section \ref{SEC:Stage1SingleParticlePackaging} proves single-particle packaging: each field creation operator carries its full set of IQNs as an irrep block and no partial factorization is possible \cite{Weyl1925,Wigner1939}.	
	Section \ref{SEC:Stage2Hybridization} analyses the hybridization of internal and external DOFs.				
	Section \ref{SEC:Stage3TensorProduct} investigates how single-particle packaged states assemble into multi-particles product states.
	Section \ref{SEC:Stage4IsotypicDecomposition} introduces isotypic decomposition to split the raw-Fock space into multi-particle charge sectors (irrep blocks) and derives the familiar superselection rules from IQN packaging.
	Section \ref{SEC:Stage5PackagedSuperposition} studies the superposition of packaged states within a fixed sector.
	Packaging leads to covariant hybrid-packaged entangled states that span every IQN.
	Moreover, packaged entanglement not only occurs in particle-antiparticle pairs (under charge-conjugation), but also among identical particles, different species of particles, and multi-particle systems of any size.	
	Section \ref{SEC:Stage6LocalGaussLawConstraint} introduces a gauge projector to obtain the physical packaged states (singlets) \cite{WWW1952,DHR1971,DHR1974,StreaterWightman,Bartlett2007,BulgarelliPanero2024}.

	\section{Overview}
	\label{SEC:Overview}

	In this manuscript, symmetry packaging refers to the fact that, under a gauge group $G$, internal quantum numbers (IQNs) (such as electric charge, flavor, color, baryon number, etc.) are locked into an irreducible block.
	No partial factorization to the IQNs is allowed.

	Before we dive into technical proofs, we give a bird’s-eye view of what we will do (the stages) and where those operations live inside the full Hilbert space (or packaging layers).

	\subsection{Road‑map: Six Successive Stages of Quantum Field Excitation}
	\label{subsec:SixStages}

	According to logical flow and physical intuition on symmetry packaging, we divide the quantum field excitation process into six successive stages (minimal axiomatic form).
	At each stage, we impose an additional physical constraint on the packaging.
	Let $G$ be the gauge group with unitary action
	$
	U: G \longrightarrow \mathcal U(\mathcal H_{\rm Fock})
	$
	on the raw Fock space.

	\begin{enumerate}[leftmargin=2.2em,label=\textbf{S\arabic*}.]
		\item \textbf{Creation/annihilation (single-particle packaging)}:
		Every field creation operator transforms in an irreducible representation (irrep) block of the local gauge group, carrying its full set of IQNs as one inseparable packet.
		We implement this by applying the creation operators to vacuum.

		\item \textbf{Hybridization}:		
		The irreducible internal packet adjoins gauge-blind external DOFs (spin, momentum, orbital, $\dots$) while preserving its internal irreducibility.
		We implement this by combining the internal block with external spectator.
		
		\item \textbf{Tensor product}:
		Individual packaged blocks combine into multi-particle product states via (anti)symmetrized tensor products, without splitting its IQNs.
		We implement this by taking the tensor product of the hybridized blocks into $n$-body product states.
		
		\item \textbf{Isotypic decomposition (multi-particle packaging)}:	
		The reducible multi-particle space splits into coupled irreducible $G$-blocks labeled by irrep $\lambda$ and net charge $Q$.
		It is the multi-particle extension of packaging.
		We implement this by applying a Peter-Weyl projector to the raw-Fock space.
		
		\item \textbf{Packaged superposition/entanglement}:
		Within each fixed $G$-block (charge sector), coherent superpositions yield packaged entangled states in which all IQNs remain inseparably locked together (and can even hybrid-entangle with external DOFs) without violating gauge invariance.
		We implement this by taking coherent linear combinations within a single net-charge sector.
		
		\item \textbf{Local gauge-invariance constraint (gauge-invariant packaging)}:
		Local gauge-invariance constraint
		$
		U(g)\,\ket{\Psi} = \ket{\Psi}
		$
		annihilates all non-singlet sectors (blocks).
		We implement this by applying the gauge (group-averaging) projector
		$$
		\Pi_{\rm phys} = \int_G d \mu(g)\,U(g)
		$$
		to filter out the non-singlet blocks and retain only the trivial irrep sector as gauge-invariant physical states.		
	\end{enumerate}

	To help understand better, we draw the following flowchart to display the six stages vs. three layers:
	
	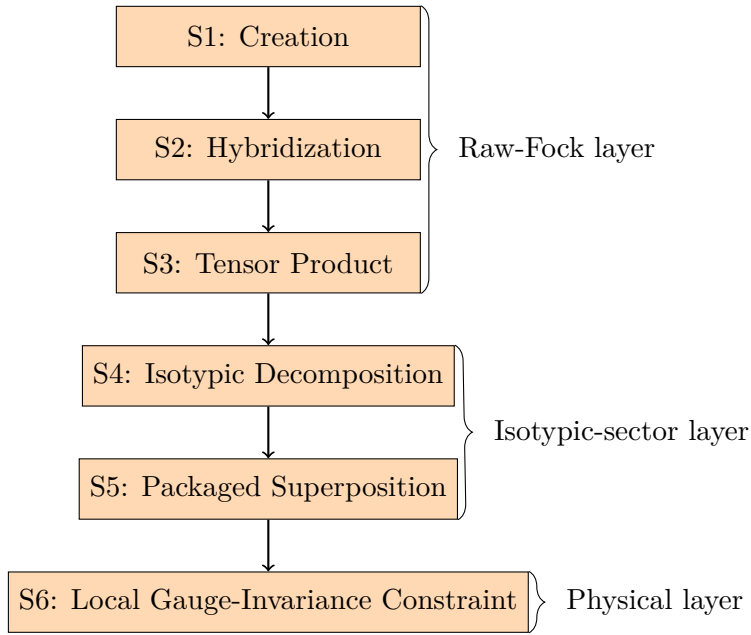
\begin{figure}[h]
		\centering
		\begin{tikzpicture}[font=\small]
			\foreach [count=\i from 1] \name in {
				S1: Creation,
				S2: Hybridization,
				S3: Tensor Product,
				S4: Isotypic Decomposition,
				S5: Packaged Superposition,
				S6: Local Gauge-Invariance Constraint
			} {
				\pgfmathsetmacro\y{7 - (\i-1)*1.5}
				\node[draw, fill=orange!30, minimum width=4cm, minimum height=8mm, anchor=center]
				(stage\i) at (2,\y) {\name};
				\ifnum\i>1
				\pgfmathtruncatemacro\prev{\i-1}
				\draw[->, thick] (stage\prev.south) -- (stage\i.north);
				\fi
			}
			
			\draw[decorate, decoration={brace, mirror, amplitude=6pt}, xshift=5cm]
			(stage3.south east) -- (stage1.north east)
			node[midway, right=10pt] {Raw‑Fock layer};
			
			\draw[decorate, decoration={brace, mirror, amplitude=6pt}, xshift=5cm]
			(stage5.south east) -- (stage4.north east)
			node[midway, right=10pt] {Isotypic‑sector layer};
			
			\draw[decorate, decoration={brace, mirror, amplitude=6pt}, xshift=5cm]
			(stage6.south east) -- (stage6.north east)
			node[midway, right=10pt] {Physical layer};
		\end{tikzpicture}
		
		\caption{Flowchart of the six stages and three layers of symmetry packaging.}
		\label{FIG:SixStagesAndThreeLayers}
	\end{figure}

	\subsection{Conceptual Scaffold: Three Layers of Packaging Hierarchy}

	In parallel, we split the symmetry packaging into three conceptual layers according to their physical and geometric pictures:
	a raw‑Fock layer, an isotypic‑sector layer, and a final gauge‑invariant (physical) layer.

	\subsubsection{Three Packaging Layers}
	
	\begin{enumerate}
		\item Raw-Fock layer (single-particle packaging layer) $\mathcal H_{\rm Fock}$:
		
		This layer includes Stages S1 - S3, where single-particle irrep blocks form and persist.
		Single-particle packaging is born in this layer.
		Multi-particle (anti)symmetrized products are taken but no coupling or projection has yet occurred.
		We implement this by applying arbitrary products of the creation/annihilation operators on the vacuum.
		A raw-Fock space $\mathcal H_{\rm Fock}$ is obtained.
			
		\item Isotypic‑sector layer (multi-particle packaging layer) $\mathcal H_{\rm iso}$: 
		
		This layer includes Stage S4 - S5, where multi-particle irrep blocks form and persist.
		Multi-particle packaging is born in this layer.
		Superselection rules form here.
		Packaged superpositions/entanglements form within each sector.
		A isotypic decomposition splits the product Fock space $\mathcal H_{\rm Fock}$ into irreducible $G$‐blocks labeled by irrep $\lambda$ and net charge $Q$.
		We implement this by applying a Peter-Weyl projector $P_\lambda$ (or a Clebsch-Gordan rotation $\mathcal U_{\rm CG}$) to the raw-Fock space $\mathcal H_{\rm Fock}$:
		$$
		\mathcal H_{\rm iso}
		\;\cong\;
		\bigoplus_{\lambda\in \widehat G}
		P_\lambda \mathcal H_{\rm Fock}
		\;=\;
		\bigoplus_{\lambda\in\widehat G}
		\mathcal H_\lambda.
		$$

		\item Physical layer (gauge-invariant packaging layer) $\mathcal H_{\rm phys}$:
		
		This layer includes Stage S6, where only physical packaged states persist.
		Gauge-invariant packaging is born in this layer.
		Local gauge-invariance constraint eliminates all nontrivial irreps and leaves only the trivial ($Q$-singlet) sectors that constitute the gauge-invariant physical Hilbert space.
		We implement this by applying a gauge projector $\Pi_{\rm phys}$ to the isotypic space $\mathcal H_{\rm iso}$ and pick out the gauge singlets (e.g., color singlets in QCD, net-charge-zero states in QED, etc):		
		$$
		\mathcal H_{\rm phys}
		\;=\; \Pi_{\rm phys}\,\mathcal H_{\rm iso}
		\;=\;\bigoplus_{Q} \mathcal H_{Q}^{\rm singlet},
		\quad
		\Pi_{\rm phys}\,\bigl|\Psi\bigr\rangle =\bigl|\Psi\bigr\rangle
		\quad
		\forall ~ \bigl|\Psi\bigr\rangle \in \mathcal H_{\rm phys}
		$$
		A reduced Hilbert space $\mathcal H_{\rm phys}$ is obtained after gauge projection.
	\end{enumerate}

	Finally, we draw the following flowchart to display the three layers:	
	$$
	\boxed{\mathcal H_{\rm Fock}}
	\xrightarrow{P_{\rm \lambda}}
	\boxed{\mathcal H_{\rm iso}}
	\xrightarrow{\Pi_{\rm phys}}
	\boxed{\mathcal H_{\rm phys}}
	$$

	\subsubsection{Necessity of the Intermediate Isotypic‑sector Layer}
	\label{SEC:LayerMotivation}

	While the need for a raw-Fock and a physical layer is obvious, why introduce an intermediate `isotypic layer'?
	We do so because:

	\begin{itemize}
		\item It isolates purely multi-particle representation-theoretic structure (multi-particle irreducible IQN blocks) before any dynamics or projection is imposed.
		In raw-Fock layer, single-particle irrep blocks forms.
		In isotypic-sector layer, multi-particle irrep blocks forms.
		In physical layer, only gauge-invariant singlets are kept.
		
		\item It is the natural home of superselection and packaged superposition/entanglement. 
		It is the space within which superselection starts and coherent packaged superpositions (packaged entanglements) occur in fixed irreps but not yet projected to a singlet.
		
		\item It emphasizes that ``superselection'' and ``packaged superposition'' occurs before the formation of physical states (gauge projection), which is contrary to our intuition that ``superselection and packaged superposition occurs in the physical packaged subspace''.
		
		\item It emphasizes that ``packaged $\ne$ physical''.
		Packaging applies even in theories where local gauge-invariance constraint is not imposed.
		Therefore, the so-called physical states are not something superior formed at the very creation, but are just the same type of quantum field excitations who are fortunate enough to form special combinations (singlets) that remained after stringent ``selections'' by the local gauge-invariance constraint.
	\end{itemize}

	In the remainder of the paper, we will repeatedly refer to the Six Stages.
	Each stage adds one additional physical requirement, progressively turning a raw creation operator into a fully gauge‑invariant physical excitation.

	\section{Stage 1: Creation/Annihilation (Single-particle Packaging)}
	\label{SEC:Stage1SingleParticlePackaging}

	In this section, we show that as soon as a particle or antiparticle is created, its entire set of IQNs is locked into a single irrep block of the gauge group.
	No subsequent operation can split that block.	
	Throughout we assume point-like matter fields.
	Gauge-invariant dressings are postponed to Stage 6 where the physical projector is applied

	\subsection{Raw Fock Construction}

	\paragraph{(1) Mode expansion of quantum field.}
	Consider a Dirac field $\psi(x)$ in a local gauge theory with gauge group $G$ (for example, $U(1)$ or $SU(N)$) \cite{PeskinSchroeder,WeinbergQFTI}.  
	Under canonical quantization, $\psi(x)$ expands in creation/annihilation operators \cite{Dixmier1981,Takesaki1979}:
	\begin{equation}\label{EQ:DiracFieldExpansion}	
		\psi(x)
		=
		\int\!\frac{d^{3}\!p}{(2\pi)^{3}}\frac{1}{\sqrt{2E_{\mathbf p}}}\sum_{s}
		\Bigl[
		\hat a_{s}(\mathbf p)\,u_{s}(\mathbf p)\,e^{-ip \cdot x} + \hat b^{\dagger}_{s}(\mathbf p)\,v_{s}(\mathbf p)\,e^{+ip \cdot x}
		\Bigr]
	\end{equation}	
	where
	$\hat a_{s}^{\dagger}(\mathbf p)$ creates a particle with momentum $\mathbf p$ and helicity or spin‐projection label $s$, and 
	$\hat b_{s}^{\dagger}(\mathbf p)$ creates an antiparticle with momentum $\mathbf p$ and the same spin label $s$.
	Here $s$ is purely an external DOF (spin or helicity) that is not part of the gauge representation, while each creation/annihilation operator also carries the appropriate (hidden) internal gauge‐charge indices (e.g., distinguishing an electron from a positron, or a quark from an antiquark).
	The field $\psi(x)$ therefore transforms in an irrep of $G \times \mathrm{Lorentz}$.
	For a Dirac spinor, the Lorentz factor is the $\left(\tfrac12,0\right) \oplus \left(0,\tfrac12\right)$ spinor representation.

	\paragraph{(2) Raw-Fock space.}
	The full raw-Fock space $\mathcal H_{\mathrm{Fock}}$ is generated by acting with arbitrary products of creation operators on the vacuum:
	\begin{equation}\label{EQ:FockSpace}
		\mathcal{H}_{\rm Fock}
		\;=\;
		\mathrm{span}\Bigl\{
		\,\hat a_{s_{1}}^{\dagger}(\mathbf p_{1}) 
		\cdots 
		\hat a_{s_{n}}^{\dagger}(\mathbf p_{n})\,\lvert0\rangle,\;
		\hat b_{s'_{1}}^{\dagger}(\mathbf{q}_{1}) 
		\cdots 
		\hat b_{s'_{m}}^{\dagger}(\mathbf{q}_{m})\,\lvert0\rangle,\;\ldots
		\Bigr\},
	\end{equation}
	where each $\hat a_{s_{i}}^{\dagger}(\mathbf p_{i})$ creates a particle of momentum $\mathbf p_{i}$ and spin/helicity $s_{i}$, and each $\hat b_{s'_{j}}^{\dagger}(\mathbf{q}_{j})$ creates an antiparticle of momentum $\mathbf{q}_{j}$ and spin/helicity $s'_{j}$.	
	Since both $\hat{a}^\dagger(\mathbf p_i)$ and $\hat{b}^\dagger(\mathbf{q}_j)$ serve as creation operators, for brevity we use $\hat a^{\dagger}$ for either particle or antiparticle creation.
	Electric or color charge is inferred from context or made explicit with a superscript.

	\subsection{Single-particle Packaging is Born at Particle Creation}
	\label{SEC:SingleParticlePackagingBorn}

	When a creation operator $\hat a_{s}^{\dagger}(\mathbf p)$ or $\hat b_{s}^{\dagger}(\mathbf p)$ acts on the vacuum, all of its IQNs are instantly locked into a single irreducible $G$‐block.
	This moment (when the field quanta first appear) marks the birth of packaging.
	That initial irreducible packet is exactly where packaging first emerges.
	Every subsequent stage simply manipulates that irreducible packet but never splits it.
	We give a formal definition for single-particle packaging:	
	
	\begin{definition}[Single-particle packaging]\label{AX:SinglePkg}		
		Under the action of a local gauge group $G$, every local operator $\hat a_{s}^{\dagger}(\mathbf p)$ or $\hat b_{s}^{\dagger}(\mathbf p)$ (creates or destroys an excitation) carries exactly one irrep $V_{\lambda}$ of $G$ locked by local gauge invariance and no physical process can split $V_{\lambda}$ into smaller pieces.
		We refer this irreducibility as \textbf{single-particle packaging}.
	\end{definition}

	We will extend this packaging to multi-particle packaging in Sec. \ref{SEC:MultiParticlePackaging} and project onto gauge-invariant packaging in Sec. \ref{SEC:GaugeInvariantPackaging}.

	\subsection{Single-particle Internal Packaged States}

	We now give a formal definition of a single-particle packaged state and then prove that no further factorization of its IQNs is possible.

	\subsubsection{Definition of a Single-particle Internal Packaged State}

	\begin{definition}[Single‐particle internal packaged state]\label{DEF:SinglePackaging}
		Let $G$ be a finite or compact gauge group and let 
		$$
		D:\;G \longrightarrow \mathrm{U}(V)
		$$
		be a finite‐dimensional unitary irrep on the internal (gauge) space $V$.
		Consider the single‐particle state
		\begin{equation}\label{SingleParticlePackagedState}
			\ket{P} \;=\; \hat{a}^\dagger(\mathbf p)\ket{0},
			\quad
			\ket{\bar P} \;=\; \hat{b}^\dagger(\mathbf p)\ket{0}\,.
		\end{equation}
		If, under every local gauge transformation $U(g)$, the creation operator $\hat{a}^\dagger(\mathbf p)$ (or $\hat{b}^\dagger(\mathbf p)$) transforms as
		$$
		U(g)\,\hat a^\dagger(\mathbf p)\,U(g)^{-1} \;=\; D(g)\,\hat a^\dagger(\mathbf p),
		\quad 
		U(g)\,\hat b^\dagger(\mathbf p)\,U(g)^{-1} =D(g)\,\hat b^\dagger(\mathbf p),
		$$
		then $\ket{P}$ and $\ket{\bar P}$ each carry all its IQNs (electric charge, flavor, color, baryon/lepton number, etc.) as a single irreducible $G$‐block.
		We say that state $\lvert P\rangle$ or $\lvert \bar{P}\rangle$ is a \textbf{single-particle internal packaged state}.
	\end{definition}

	\begin{example}[QED]
		Let $G = U(1)$.
		\begin{enumerate}
			\item Electron in QED:			
			The electron creation operator 
			$\hat a_{e^{-}}^\dagger(\mathbf p,s)$
			carries charge $Q=-e$, lepton number $+1$, spin projection $s$, and momentum $\mathbf p$.
			Under a local $U(1)$ transformation $U_\alpha=e^{-ie\alpha(x)}$,
			$$
			U_\alpha\,\hat a_{e^{-}}^\dagger\,U_\alpha^{-1}
			=e^{-ie\alpha}\,\hat a_{e^{-}}^\dagger\,.
			$$
			So $\ket{e^{-};\mathbf p,s}=\hat a_{e^{-}}^\dagger(\mathbf p,s)\ket{0}$ is a single-particle packaged state that carries all its IQNs as a 1D irrep block of $U(1)$.
			
			\item Photons in QED:			
			A photon creation operator transforms as a trivial 1-dimensional block because the fundamental representation of U(1) is one-dimensional.
			The field has no internal structure as that of color and automatically behaves like a irreducible block (gauge-invariant under gauge transformations).
		\end{enumerate}			
	\end{example}

	\begin{example}[Single-particle packaging under SU(3)]
		For $G=\mathrm{SU}(3)$, a single-particle field operator belongs to a single-particle irrep block:
		
		\begin{enumerate}
			\item \textbf{Young-diagram notation.}		
			A quark is in the $\mathbf{3}$ irrep block, an antiquark in the $\overline{\mathbf{3}}$ irrep block, and a gluon in the $\mathbf{8}$ irrep block.
			
			\item \textbf{Coordinate (wavefunction) representation.}		
			Let $U \in \mathrm{SU}(3)$.
			Let		
			$\psi_i(x)$, $i=1,2,3$ be a $\mathbf3$-valued quark wavefunction,
			$\bar\psi^i(x)$, $i=1,2,3$ be an $\overline{\mathbf3}$-valued antiquark wavefunction,
			and $A^a(x)$, $a=1,\dots,8$ be an $\mathbf8$-valued gluon field.		
			Their transformation laws under a gauge rotation $U$ are		
			$$
			\begin{aligned}
				\text{quark:}\quad &\psi_i(x)\;\longrightarrow\;(U)_i{}^j\;\psi_j(x),
				\\
				\text{antiquark:}\quad &\bar\psi^i(x)\;\longrightarrow\;(U^{-1})^i{}_j\;\bar\psi^j(x)
				\;=\;(U^*)^i{}_j\,\bar\psi^j(x),
				\\
				\text{gluon:}\quad &A^a(x)\;\longrightarrow\;D^{(8)}(U)^a{}_b\;A^b(x),
			\end{aligned}
			$$		
			where $D^{(8)}(U)^a{}_b\equiv \bigl(e^{i\,\alpha^c\,T^{(8)}_c}\bigr)^a{}_b$ is the adjoint action with the structure‐constant matrices $(T^{(8)}_c)^a{}_b=-\,f_{cab}$.
			
			\item \textbf{Second-quantized creation operators.}		
			Introduce a Fock vacuum $\lvert0\rangle$. Then
			
			Quark creation operators:
			$$
			a_i^\dagger\,,\quad i=1,2,3
			\,,\qquad
			\{\,a_i\,,\,a_j^\dagger\}\;=\;\delta_{ij}, 
			\quad
			\{a_i,a_j\}=0,
			\quad
			\{a_i^\dagger,a_j^\dagger\}=0.
			$$
			
			Antiquark creation operators:
			$$
			b^{\dagger\,i}\,,\quad i=1,2,3,
			\qquad
			\bigl\{\,b^i\,,\,b^{\dagger\,j}\bigr\}\;=\;\delta^{ij},
			\quad
			\{b^i,b^j\}=0,
			\quad
			\{b^{\dagger\,i},b^{\dagger\,j}\}=0.
			$$
			
			Gluon creation operators:
			$$
			g^{\dagger\,a}\,,\quad a=1,\dots,8,
			\qquad
			\bigl[\,g^a\,,\,g^{\dagger\,b}\bigr]\;=\;\delta^{ab},
			\quad
			[g^a,g^b]=[g^{\dagger\,a},g^{\dagger\,b}]=0.
			$$
			
			All of these act on $\lvert0\rangle$ to create single-particle packaged states:
			$\lvert q_i\rangle=a_i^\dagger\lvert0\rangle,\;\lvert\bar q^i\rangle=b^{\dagger\,i}\lvert0\rangle,\;\lvert g^a\rangle=g^{\dagger\,a}\lvert0\rangle$.

			\item \textbf{How to distinguish that $a^\dagger$ is a $\mathbf 3$ while $b^\dagger$ is a $\overline{\mathbf 3}$ ?}		
			One can see that $a^\dagger$ is a $\mathbf 3$ while $b^\dagger$ is a $\overline{\mathbf 3}$ by the ``index position'' and hence the transformation law.
			Specifically:
			\begin{itemize}
				\item A lower index $i$ on $a_i^\dagger$ means fundamental ($\mathbf 3$) and transforms with $U$:
				$$
				a_i^\dagger\;\xrightarrow{U}\;U_i{}^j\,a_j^\dagger
				\quad\Longrightarrow\quad\mathbf3.
				$$
				
				\item An upper index $i$ on $b^{\dagger\,i}$ means anti-fundamental ($\overline{\mathbf 3}$) and transforms with $U^\dagger=U^{-1}=U^{\!*}$:
				$$
				b^{\dagger\,i}\;\xrightarrow{U}\;(U^{-1})^i{}_j\,b^{\dagger\,j}
				\;=\;(U^*)^i{}_j\,b^{\dagger\,j}
				\quad\Longrightarrow\quad\overline{\mathbf3}.
				$$
			\end{itemize}
			Equivalently, under an infinitesimal generator $T^c$,	
			$$
			[T^c,\;a_i^\dagger] \;=\;(T^c)_i{}^j\;a_j^\dagger,
			\qquad
			[T^c,\;b^{\dagger\,i}]\;=\;-(T^c)_j{}^i\;b^{\dagger\,j},
			$$	
			so $a^\dagger$ carries the defining rep and $b^\dagger$ its complex conjugate.

			Since the fundamental representation of $SU(3)$ is complex, $\mathbf3$ and $\overline{\mathbf3}$ are inequivalent.
			They cannot be mapped into each other by a similarity transformation that is also an $SU(3)$ element.
		\end{enumerate}
	\end{example}

	\begin{example}[Majorana neutrino]
		If $\nu(x)$ is a Majorana field, then $\hat a_\nu^\dagger=\hat b_\nu^\dagger$ carries no gauge charge (only global lepton number).  The state $\ket{\nu;\mathbf p,s}$ is invariant under all local gauge transformations and still forms a single‐particle packaged state.
	\end{example}

	\begin{remark}[Package charge only or all IQNs?]
		\label{REMARK:PackageAllIQNs}
		In Definition \ref{DEF:SinglePackaging}, we used gauge group $G$ to verify package, which is related to the gauge charge, such as electric charge or color charge.		
		But \textbf{do not be misled to think that the package states only package the charge}.		
		In fact, the creation operation creates the entire particle/antiparticle as an irreducible block.
		This means that \textbf{all IQNs are fully packaged}.
		The same logically apply to the product state and packaged entangled state, which will be discussed in later sections.
	\end{remark}

	\subsubsection{Properties of Single-particle Internal Packaged States}

	We now prove a very important property of single-particle packaged states:
	the internal block in a single-particle packaged state is indivisible.
	This property is also effective in multiple field excitations.
	Throughout this work, we assume that $G$ is compact so that every unitary irrep is finite dimensional and Schur’s lemma applies.

	Why no partial factorization of IQNs?
	The reason is that local gauge-invariance forbids splitting these quantum numbers.
	We restate this standard fact in an group-theoretic language as follows:

	\begin{theorem}[No partial factorization of single‐particle IQNs]
		\label{THM:NoPartialFactorization}
		Let $\psi(x)$ be a local quantum field with a finite or compact gauge group $G$ (Abelian or non-Abelian).
		Assume that under 
		$G\times\mathrm{Lorentz}$, $\psi(x)$ transforms as  
		$$
		\psi(x)\;\longmapsto\; D(g)\,
		U(\Lambda)\,\psi\!\bigl(\Lambda^{-1}x\bigr),
		\qquad g\in G,\;\Lambda\in\mathrm{Lorentz},
		$$
		where
		$
		D: G \to \mathrm{GL}(V)
		$
		is a unitary and finite dimensional irrep on the internal (gauge) space $V$.
		In the mode expansion Eq.~(\ref{EQ:DiracFieldExpansion}),
		denote the column of particle creation operators by
		$
		\hat{\mathbf a}^{\dagger}(\mathbf p)
		:= \bigl(\hat a_{1}^{\dagger}(\mathbf p),\,
		\hat a_{2}^{\dagger}(\mathbf p),\,
		\dots\bigr)^{T}.
		$
		Then for each fixed momentum $\mathbf p$, we have
		
		\begin{enumerate}
			\item \emph{Gauge covariance of $\hat{\mathbf a}^{\dagger}(\mathbf p)$}.
			Under a pure gauge transformation $(g,\Lambda=\mathbf 1)$,
			the entire column $\hat{\mathbf a}^\dagger(\mathbf p)$ transforms in the same irrep $D$,
			$$
			\hat{\mathbf a}^\dagger(\mathbf p)
			\;\longmapsto\;
			D(g)\,\hat{\mathbf a}^\dagger(\mathbf p),
			\quad g\in G.
			$$
			
			\item \emph{No partial projectors:}
			By Schur’s lemma, any internal projector $P$ commuting with all $D(g)$ must be $0$ or $\mathbf 1$.
			Thus, the IQNs cannot be split into independent subspaces.
			
			\item \emph{Lorentz invariance:}
			Since each Lorentz transformation $U(\Lambda)$ acts on a different tensor factor (of the form $V_{\text{Lorentz}}\otimes V_{\text{int}}$), it commutes with $D(g)$, i.e., $[D(g), U(\Lambda)] = 0$.
			Thus, it cannot induce further splitting of the internal block (they act on a separate spinor space).
		\end{enumerate}
		Therefore, a single‐particle creation operator always carries its full set of IQNs as one inseparable irreducible block under $G$.
	\end{theorem}

	\begin{proof}
		We break the argument into three steps.
		
		\begin{enumerate}
			\item Gauge covariance of the creation operators.	
			
			The quantum field $\psi(x)$ transforms as
			$
			\psi(x)\;\longmapsto\; D(g)\,
			U(\Lambda)\,\psi\!\bigl(\Lambda^{-1}x\bigr)
			$
			When $\psi(x)$ is expanded into its mode (Fourier) decomposition Eq.~(\ref{EQ:DiracFieldExpansion}),
			each creation operator $ \hat{a}^{\dagger}_{\alpha}(\mathbf p) $ inherit the same transformation law:
			$$
			\hat{\mathbf{a}}^\dagger(\mathbf p) \equiv (\hat{a}^{\dagger}_{\alpha}(\mathbf p))_{\alpha}
			\quad\text{transforms as}\quad
			\hat{\mathbf{a}}^\dagger(\mathbf p) \longmapsto D(g)\,\hat{\mathbf{a}}^\dagger(\mathbf p).
			$$
			This is because the spinor wavefunctions $u_s(\mathbf p), v_s(\mathbf p)$ are just ordinary c-number vectors which are not affected by $ D(g) $.
			
			\item Schur’s lemma and absence of partial projectors.
			
			Schur’s lemma \cite{Jones1998} states that the commutant of an irreducible $D$ is $c \mathbf 1$.
			We now prove the ``absence of partial projectors'' by contradiction:
			
			Assume that there exists a nontrivial projector $P:V\to V$ ($P^2=P,\;P\neq0,I$) that commutes with every group element: $P\,D(g)=D(g)\,P, \;\forall g\in G$.
			According to Schur’s lemma, therefore, we have $P=c\,\mathbf 1$.
			Since $P$ is idempotent $c^2=c$, we must have $c=0$ or $1$.
			In other words, any projector $P$ that commutes with all $D(g)$ must be either $P=0$ or $P=\mathbf 1$.
			This means that no nontrivial projector exists and the internal
			indices cannot be decomposed into independent subspaces.
			In other words, no partial factorization of the single-particle IQNs is possible.
			
			\item Lorentz transformations do not alter irreducibility.
			
			Although $\psi$ also carries a Lorentz index (acted on by $U(\Lambda)$), that action lives on a different tensor factor from $D(g)$ and commutes with $D(g)$, that is, $[U(\Lambda), D(g)]=0$. Therefore, Lorentz reducibility cannot induce any further
			splitting of the internal gauge block.			
		\end{enumerate}	
		All these facts together show that the full set of IQNs carried by a single-particle creation operator is locked together by irreducibility and cannot be partially factorized.
		It should be emphasized that, although the proof used gauge indices, the irreducible block includes all conserved internal labels of the field (similar to Definition \ref{DEF:SinglePackaging}, see Remark \ref{REMARK:PackageAllIQNs}).
	\end{proof}

	\begin{remark}
		Schur's Lemma applies only if the gauge representation is truly irreducible.
		In the standard model, several fields sit in reducible internal spaces (e.g. SU(2) $\times$ U(1) hypermultiplets). 
		Thus, irreducibility is taken after choosing a definite multiplet (say, left-handed doublet) and cannot be applied to a whole generation without first projecting.
	\end{remark}

	\begin{example}[A counter example: ``half quark'' cannot exist]		
		Suppose there exists an operator $X_i^\dagger$ with a single color index $i=1,2,3$, but which transforms under half of the fundamental representation, schematically		
		$$
		U(g)\,X_i^\dagger\,U(g)^{-1}
		=\bigl[D_{\mathbf 3}(g)\bigr]_{ij}^{1/2}\,X_j^\dagger,
		$$		
		where the square-root is intended entry-wise.
		Since $D_{\mathbf 3}(g)$ is unitary, any matrix square root violates
		group multiplication:		
		$$
		D(g_1g_2)\neq\sqrt{D(g_1)}\,\sqrt{D(g_2)} \qquad
		\text{for generic }g_1,g_2\in SU(3).
		$$
		Thus, the map $g\mapsto\sqrt{D_{\mathbf 3}(g)}$ is not a representation of SU(3).
		The operator $X_i^\dagger$ cannot respect locality and canonical (anti)commutation relations.

		Equivalently, if one tries to assign fractional color charge
		$\tfrac12 \times \mathbf 3$, then the commutant of the gauge group is larger than $\mathbb C\,\mathbf 1$.
		This is contradict to Theorem \ref{THM:NoPartialFactorization}.		
		Thus, no local operator can create a fractional color charge without violating group representation laws.

		This explicit failure highlights the power of IQN packaging:
		only irreps of the local gauge group appear as creation operators in a consistent quantum field theory.
	\end{example}

	\subsection{Single-particle Internal Space}
	\label{SEC:SingleParticleInternalSpace}

	Having established single‐particle packaging, we now collect all allowed single‐particle packaged states into a Hilbert space.
	This space provides the basic building block for the multi-particle theory developed in the next section.

	\subsubsection{Single-particle Internal Space $\mathcal H_{\rm int}$}

	Due to the local gauge-invariance, each single-particle state given by Definition \ref{DEF:SinglePackaging} is born to be packaged.
	If we collect all packaged states, then we obtain a full single-particle internal space $\mathcal H_{\rm int}$, which is generally reducible.

	The single-particle internal space $\mathcal H_{\rm int}$ are the foundation for multi-particle packaged subspaces.

	\subsubsection{Isotypic Decomposition of $\mathcal H_{\rm int}$}
	\label{SEC:IsotypicDecompositionSingleParticleInternalSubspaces}

	According to representation theory, due to the gauge group $G$, the full single-particle internal space $\mathcal H_{\rm int}$ indeed experiences an isotypic decomposition.
	Specifically, for a finite or compact gauge group $G$ acting unitarily on the full one-particle internal Hilbert space $\mathcal H_{\rm int}$, we have the following proposition:

	\begin{proposition}[Isotypic decomposition of full single-particle internal space]
		\label{PROP:SingleParticleIntDecomp}
		Let $G$ be a finite or compact gauge group that acts unitarily on the full single-particle internal space $\mathcal H_{\rm int}$.
		Then we have:
		\begin{enumerate}
			\item The isotypic decomposition:
			\begin{equation}\label{EQ:SingleParticleIntDecomp}
				\mathcal H_{\rm int} \cong
				\bigoplus_{R \in \widehat G}
				V_{R} \otimes \mathcal M_{R},
			\end{equation}
			where $V_R$ is the carrier space of the irrep $R$ (the packaged block that contains the full set of IQNs for that irrep),
			and $\mathcal M_R\simeq\mathbb C^{\,m_R}$ is the multiplicity space recording how many times $R$ occurs.
			
			\item Each factor $V_R$ is a single-particle internal-packaged subspace.
			No symmetry-respecting operation can act non‑trivially on a proper subspace
			of $V_R$.
		\end{enumerate}
	\end{proposition}

	The natural home for discussing isotypic decomposition is the multi-particle packaged subspace.
	Thus, we leave the detailed proof to Sec. \ref{SEC:Stage4IsotypicDecomposition}

	In Eq. \eqref{EQ:SingleParticleIntDecomp}, each $V_R$ already carries all IQNs associated with that irrep.
	Nothing can split it without breaking $G$.	
	A local creation operator $\hat a_R^{\dagger}$ always lands in one
	carrier $V_R$.
	That block is ``born packaged''.
	Stage 1 selects one block $V_R$ and populates it via a creation operator
	$\hat a_R^{\dagger}$.
	All later stages manipulate that block (possibly tensored with external labels) but never split it.

	\begin{example}[Single-particle Hilbert spaces]
		\leavevmode
		\begin{enumerate}[label=(\alph*)]
			\item QED:  
			$\mathcal H_{\mathrm{1}}=\!\bigoplus_{Q\in\mathbb Z} V_Q
			\otimes \mathcal M_Q$ with $V_Q\cong\mathbb C$.
			
			\item SU(3) quark:  
			$V_{\mathbf 3}\cong\mathbb C^3$ for a quark,
			$V_{\mathbf 8}\cong\mathbb C^8$ for a gluon, etc.
		\end{enumerate}
	\end{example}

	\subsubsection{Single-particle Internal-packaged Subspace $\mathcal H_{\rm pkg}$}

	If we select a subspace of $\mathcal H_{\rm pkg} \subset \mathcal H_{\rm int}$ that satisfies certain conditions, then it can be a packaged subspace.
	We now give the formal definition as follows:

	\begin{definition}[Single-particle internal-packaged subspace]
		\label{DEF:SingleParticleInternalPackagedSubspace}
		Let $G$ be a gauge group that acts on the full internal space $\mathcal H_{\rm int}$ by local unitaries $U(g)(x)$ and $\hat Q$ be the global charge operator.
		If a subspace 
		$\mathcal H_{\rm pkg} \subset \mathcal{H}_\mathrm{int}$
		satisfies the following conditions:	
		\begin{enumerate}[label=(\alph*)]
			\item \textbf{Fixed charge.}
			There exists $Q_0$ such that $\hat Q\ket{\psi}=Q_0\ket{\psi}$ for all $\ket{\psi}\in\mathcal H_{\rm pkg}.$	
			
			\item \textbf{Gauge invariance (singlet property).}
			For every local gauge transformation $U(g)(x)$ (with $g \in G$), there is a one-dimensional character $\chi_\lambda(g)$ (fixed for the whole subspace, independent of the particular $\ket\psi$), such that
			$$
			U(g)(x)\,\ket{\psi}
			=\chi(g)\,\ket{\psi},
			\quad\forall\ket{\psi} \in \mathcal H_{\rm pkg}.
			$$
			Equivalently, $U(g)(x)\big\vert_{\mathcal H_{\rm pkg}}=\chi(g)\,\mathbf 1$ on~$\mathcal H_{\rm pkg}$).			
			\footnote{%
				If $G$ is non‑Abelian, then every one‑dimensional unitary
				representation $\chi$ factors through the Abelianisation $G/[G,G]$.
				The special case $\chi(g)\equiv1$ gives the usual notion of a gauge singlet.
			}
		\end{enumerate}
		then we say that $\mathcal H_{\rm pkg}$ is a \textbf{single-particle internal-packaged subspace} of $\mathcal{H}_\mathrm{int}$.
	\end{definition}

	\begin{example}[One‐ and two‐dimensional packaged subspaces]
		\leavevmode
		\begin{enumerate}
			\item \textbf{Single electron.}  
			In QED, the one‐electron state $\ket{e^-;\mathbf p}=\hat a_{e^-}^\dagger(\mathbf p)\ket0$
			has charge $Q=-e$ and transforms as $U_\alpha=e^{-ie\alpha}\mathbf 1$.
			Thus, the span $\mathcal H_{\rm int}^{(e^-)}=\mathrm{span}\{\ket{e^-;\mathbf p}\}$
			is a one‐dimensional internal‐packaged subspace.
			
			\item \textbf{Neutral kaon doublet.}  
			The flavor states
			$\ket{K^0}$ and $\ket{\overline K^{0}}$ both satisfy $Q=0$ and share the same one-dimensional trivial character  $\chi(g)=1$.
			Their span
			$\mathcal H_{\rm int}^{(K)}=\mathrm{span}\{\ket{K^0},\ket{\overline
				K^{0}}\}$ is a two-dimensional internal-packaged subspace.
		\end{enumerate}
	\end{example}

	\begin{remark}[Gauge‑blind labels]
		Every creation operator carries, in addition to its internal index $\alpha$, a collection of gauge‑blind external labels such as spin/helicity $s$, momentum $\mathbf p$, orbital quantum numbers, etc.
		They play no role in the irreducibility argument of Stage 1 and will be organized systematically only in Stage 2, where we adjoin an explicit spectator Hilbert space $\mathcal H_{\mathrm{ext}}$.
	\end{remark}

	A internal packaged subspace has the following basic properties:
	
	\begin{proposition}\label{PROP:SinglePackagedProperties}
		Let $\mathcal H_{\rm pkg}$ be an internal packaged subspace.
		Then:
		\begin{enumerate}[label=(\roman*)]
			\item \textbf{Stability of $\mathcal H_{\rm pkg}$ under charge‐preserving maps.}
			Any bounded operator $V$ satisfying $[V,\hat Q]=0$ maps $\mathcal H_{\rm pkg}$ into itself: $V\,\mathcal H_{\rm pkg} \subseteq \mathcal H_{\rm pkg}$.
			
			\item \textbf{Existence of $\mathcal H_{\rm pkg}$ by construction.}
			Fix a charge sector $\mathcal H_{Q_0}$ with $\dim\mathcal H_{Q_0}\ge1$.
			Choose any orthonormal set $\{\ket{\psi_j}\}_{j=1}^k\subset\mathcal H_{Q_0}$ that all transform under the same one‐dimensional character $\chi$.
			Then the span
			$\mathcal H_{\rm pkg}=\mathrm{span}\{\ket{\psi_j}\}$ is an internal‐packaged subspace.
		\end{enumerate}
	\end{proposition}

	\begin{proof}
		\leavevmode
		\begin{enumerate}[label=(\roman*)]
			\item According to Definition \ref{DEF:SingleParticleInternalPackagedSubspace} (a), $\forall ~ \ket\psi\in\mathcal H_{\rm pkg}$, we have $\hat Q\ket\psi=Q_0\ket\psi$.
			Now we further have $[V,\hat Q]=0$.
			These together gives $\hat Q(V\ket\psi)=V\hat Q\ket\psi=Q_0(V\ket\psi)$.			
			Moreover, $\;U(g)(x)V\ket\psi=V\,U(g)(x)\ket\psi=\chi(g)\,V\ket\psi$.
			Thus, $V\ket\psi\in\mathcal H_{\rm pkg}$.
			
			\item The chosen vectors share the eigenvalue $Q_0$ and the character
			$\chi$.
			Therefore, their span obeys Definition
			\ref{DEF:SingleParticleInternalPackagedSubspace} (a) and (b).
		\end{enumerate}
	\end{proof}

	\begin{remark}
		\leavevmode
		\begin{enumerate}[label=(\roman*)]
			\item Definition~\ref{DEF:SingleParticleInternalPackagedSubspace} applies to any compact gauge group, continuous or discrete, Abelian or non‐Abelian.
			In the non‐Abelian case, condition (b) simply singles out a one‐dimensional character.
			
			\item No particular dimension is imposed.
			Later we use two‐dimensional packaged subspaces to represent packaged qubits and use higher dimensions to represent packaged qudits.
			
			\item If an operator $P_{\rm pkg}$ projects onto $\mathcal H_{\rm int}$, then we have $[\hat Q,P_{\rm pkg}]=0$ and $[U(g)(x),P_{\rm pkg}]=0$.
		\end{enumerate}
	\end{remark}

	\subsection{Gauge Anomalies}

	A gauge anomaly is a quantum obstruction to the conservation of the gauge current.
	Equivalently, the Jacobian of the path integral measure under a local gauge transformation fails to be unity \cite{Fujikawa1979}.
	Everything we’ve said up to now relies on the assumption that our gauge symmetry really survives quantization, that is, there are no gauge anomalies.
	If the gauge group is anomalous, then the irreducibility arguments of Stage 1 can fail because the quantum theory is not invariant under the naïve local transformations.

	In a theory with a gauge anomaly, we may meet the following problems:
	
	\begin{enumerate}
		\item \emph{Non‐invariance of the vacuum.}  
		Under an infinitesimal gauge rotation generated by
		$$
		\hat G^a(x) \;=\; D_i E^{a,i}(x) - \rho^a(x),
		$$
		the measure picks up an anomalous term.
		Therefore, even the vacuum is not truly invariant \cite{Bertlmann1996}.
		
		\item \emph{Failure of Schur’s Lemma.}  
		In Stage 1, we used Schur’s lemma on the commutant of the local representation $D(g)$ to argue that only trivial projectors commute with all $D(g)$.
		But if local gauge invariance is broken \cite{Higgs1964,EnglertBrout1964}, then there is no consistent unitary representation $U(g)$ of the gauge group on the full Hilbert space.
		The entire irreducibility argument collapses \cite{WeinbergQFTII}.
		
		\item \emph{No consistent packaging.}  
		Since the gauge covariance of creation operators $\hat a^\dagger(\mathbf p)$ is spoiled, we can no longer assert that each $\hat a^\dagger$ carries an irreducible packet of IQNs.
		Even worse, physical states cease to lie in well defined gauge singlet sectors.
		Therefore, the local gauge-invariance constraint cannot be imposed consistently \cite{HenneauxTeitelboim}.
	\end{enumerate}

	Genuine gauge anomalies completely wreck the packaging framework:
	there is no consistent local gauge action, no local gauge-invariance, and no irreducible charge packets \cite{Bertlmann1996}.
	If the gauge symmetry is anomalous, we cannot even talk about charge superselection and define a proper local gauge-invariance constraint, let alone derive color singlets.
	In practice, any gauge theory we work with must be anomaly free (the Standard Model cancels all gauge anomalies \cite{Bouchiat1972}).
	Only then we can safely invoke the irreducibility and packaging arguments of Stage 1.

	Global anomalies (like the QCD $\mathrm{U}(1)_A$) only affect global quantum numbers and do not spoil the packaging of real gauge charges \cite{Witten1982}.
	We do not package those global symmetries in Stage 1.

	\begin{example}[Example of $\mathrm{U}(1)_A$ in QCD]
		$\mathrm{U}(1)_A$ in QCD is a global chiral rotation,
		$$
		\psi(x)\;\mapsto\;e^{i\alpha\gamma_5}\,\psi(x),
		$$
		not a local gauge transformation.
		It has a famous anomaly (the Adler-Bell-Jackiw anomaly) which breaks it down to a discrete $\mathbb Z_{2N_f}$ subgroup at the quantum level \cite{Adler1969,BellJackiw1969}.
		Since $\mathrm{U}(1)_A$ is not gauged, we never attempted to package states in irreps of it in Stage 1.
		We only package genuine gauge groups (e.g. $SU(3)_\text{color}$ or local $U(1)_\text{em}$).
		Thus the $\mathrm{U}(1)_A$ anomaly has no impact on our Stage 1 argument.
		It simply tells us that chiral charge is not conserved and superselection does not apply.
		Thus, we can form superpositions of different chirality, but we never needed chirality packaging in the first place.
	\end{example}

	\section{Stage 2: Hybridization (Single-particle Internal $\boldsymbol{\otimes}$ External Blocks)}
	\label{SEC:Stage2Hybridization}

	In Stage 1, a local creation operator already fixes all	IQNs inside one irreducible $G$-block.		
	Here we show that, in Stage 2, the irreducible internal block adjoin all gauge‐blind external DOFs (such as spin or helicity, momentum label, orbital index, optical mode, etc.).
	We further show that the IQN packaging remains intact under this hybridization of internal and external DOFs.

	\subsection{Single-particle External Space $\mathcal H_{\mathrm{ext}}$}

	First, we collect all the external labels of a single-particle into a new Hilbert space: $\mathcal H_{\rm ext}$.

	
	Usually, an external DOF $\xi$ refer to	
	\begin{itemize}\itemsep4pt
		\item Spin / helicity:  $s\in\{\uparrow,\downarrow\}$,
		\item Momentum mode:    $\mathbf p$,
		\item Atomic orbital:   $n\ell m$,
		\item Optical path:     $\text{A or B}$, $\cdots$.
	\end{itemize}

	If we collect all the states of external DOFs $\xi$, then we obtain a new Hilbert space
	$$
	\mathcal H_{\rm ext}
	\;=\;
	\mathrm{span}\bigl\{
	\ket{\xi}\bigr\},
	\quad
	\xi\in\mathcal I_{\rm ext},
	$$
	where $\mathcal I_{\rm ext}$ is any index set that the gauge group $G$ does not act on.

	We now formally give the definition of a single-particle external subspace as follows:
		
	\begin{definition}[Single‐particle external subspace]
		Let $G$ be the gauge group, and let $\mathcal I_{\rm ext}$ be any countable index set labeling purely external DOFs on which $G$ acts trivially (for example, spin/helicity $s$, momentum $\mathbf p$, orbital $n\ell m$, optical path A/B, etc.).
		Define the Hilbert space 
		$$
		\mathcal H_{\rm ext}
		\;=\;
		\mathrm{span}\bigl\{\,\ket{\xi}\bigr\}_{\xi\in\mathcal I_{\rm ext}},
		\qquad
		\langle\xi'\!\mid\xi\rangle=\delta_{\xi'\!,\xi},
		$$
		and demand $\forall\,g\in G:\;U(g)\ket{\xi}=\ket{\xi}$.
		Then we say that $\mathcal H_{\rm ext}$ is a \textbf{single‐particle external subspace}.
	\end{definition}

	\begin{remark}
		Since every field operator carries a compact gauge symmetry, its IQNs are always locked into a single irreducible $G$-block (Stage 1).
		External DOFs are different:
		an independent symmetry $G_{\text{ext}}$ may or may not exist, and even when it does, it need not act on all external labels.
		Thus, we focus on ION packaging, but leave the detailed study on external DOF packaging to next companion paper.
	\end{remark}

	\subsection{Single-particle Hybrid Space $\mathcal H_{\rm hyb} = \mathcal H_{\rm int} \otimes \mathcal H_{\mathrm{ext}}$}
	\label{SEC:SingleParticleHybridSpace}

	A single-particle hybrid space is the tensor product of internal space $\mathcal H_{\rm int}$ and external space $\mathcal H_{\mathrm{ext}}$, that is,
	\[
	\mathcal H_{\rm hyb} = \mathcal H_{\rm int} \otimes \mathcal H_{\mathrm{ext}}
	\]
	Using Eq. \eqref{EQ:SingleParticleIntDecomp}, we have
	\[
	\mathcal H_{\rm hyb} 
	= 
	\bigoplus_{R \in \widehat G} V_{R} 
	\otimes 
	\left(\mathcal M_{R} \otimes \mathcal H_{\mathrm{ext}}\right)
	\]
	Generally, $\mathcal H_{\rm hyb}$ is not packaged, but each
	$
	V_{R} \otimes 
	\left(\mathcal M_{R} \otimes \mathcal H_{\mathrm{ext}}\right)
	$
	is a packaged subspace.

	\subsubsection{Definition of Single‐particle Hybrid Space}

	\begin{definition}[Single‐particle hybrid space]
		\label{DEF:SingleParticleHybridPackagedSubspace}
		Let $G$ be a finite or compact gauge group,
		$\mathcal H_{\rm int}$ be an internal space	and $\mathcal H_{\rm ext}$ be an external subspace with trivial $G$-action.
		Then we say that the tensor product
		\begin{equation}\label{EQ:SingleParticleHybridPackagedSubspace}
			\mathcal H_{\rm hyb} 
			= \mathcal H_{\rm int} \otimes \mathcal H_{\rm ext}
			= \bigoplus_{R \in \widehat G} V_{R} 
			\otimes 
			\left(\mathcal M_{R} \otimes \mathcal H_{\mathrm{ext}}\right)
		\end{equation}
		is a \textbf{single-particle hybrid space}, and
		each summand
		$
		V_{R} \otimes 
		\left(\mathcal M_{R} \otimes \mathcal H_{\mathrm{ext}}\right)
		$
		is a \textbf{single-particle hybrid-packaged subspace}.
	\end{definition}

	Internal packaging is unbreakable.
	All other external labels remain spectators.	
	Hybridization ``tacks on'' all external labels without spoiling the internal packaging.
	Only when you thereafter declare an external symmetry group $H$ do those particular DOFs themselves become packaged.
	Thus, we obtain the total on‑site action
	\begin{equation}\label{EQ:SingleParticleTotalAction}
		U(g_{\text{int}})
		=(D^{(R_{\text{int}})}(g_{\text{int}}))
		\otimes
		\mathbf 1_{\text{ext}}.
	\end{equation}
	The entire hybrid space still furnishes the same irrep $R$.
	In particular, no splitting of the internal IQNs is possible.

	\begin{example}[Examples of single‐particle hybrid space]
		\leavevmode
		\begin{enumerate}
			\item Electron + spin.
			
			In QED, the one‐electron internal block $\ket{e^-}$ (charge $-e$) is one-dimensional
			($D^{(R)}\cong\mathbf 1$).
			Tensoring with spin-$\tfrac12$ space gives a single-particle hybrid-packaged subspace
			$$
			\mathcal H_{\rm hyb}
			=\mathrm{span}\{\ket{e^-,\uparrow}, \ket{e^-,\downarrow}\},
			$$
			of dimension $2$.
			The gauge‐action remains $U_\alpha=e^{-ie\alpha} \mathbf 1$.
			
			\item Counter-example: necessity of ``trivial gauge action'' on $\mathcal H_{\mathrm{ext}}$.
			
			Suppose one attempted to encode another gauge SU(2), letting the local gauge group act as
			$U(g) = D(g) \otimes e^{i\theta\hat\sigma_z}$.
			Then the internal and external factors no longer commute,
			$[D(g) \otimes \mathbf 1, \mathbf 1 \otimes e^{i \theta \hat \sigma_z}] \neq 0$,
			and the combined representation generally becomes irreducible again.
			All external DOFs would be re-absorbed into the packaged block-hybridization fails.
			The trivial-action requirement in Def \ref{DEF:SingleParticleHybridPackagedSubspace} is therefore essential.
		\end{enumerate}
	\end{example}

	\begin{remark}[Robustness and utility of hybridization]
		When we adjoin gauge‐blind external DOFs (spin, momentum, etc.) to an irreducible internal block, nothing can leak out of that block:
		
		\begin{enumerate}
			\item Block integrity.  
			Since the gauge group $G$ acts only on the internal factor $\mathcal H_{\rm int}$ and trivially on the external factor $\mathcal H_{\rm ext}$, any operator that commutes with $G$ can reshuffle external labels at will, but it cannot split or disturb the internal irrep $R$.
			The packaged IQNs remain locked together.
			
			\item Increased encoding capacity.  
			The hybrid space 
			$
			\mathcal H_{\rm hyb} = \mathcal H_{\rm int} \otimes \mathcal H_{\rm ext}
			$
			carries the same irrep $R$ on its internal factor, but has dimension
			$\dim \mathcal H_{\rm hyb} = \dim R \times \dim \mathcal H_{\rm ext}$.  
			By choosing a large external index set one can embed high‐dimensional qudits without ever violating internal‐charge superselection.
			
			\item Ready for many‐body assembly.  
			All symmetrization or antisymmetrization of identical excitations can be postponed to Stage 3, where these enlarged hybrid blocks are tensor‐multiplied and then reorganized into new irreducible sectors.
		\end{enumerate}
	\end{remark}

	\subsubsection{Properties of Single-particle Hybrid-packaged Subspaces}

	\begin{proposition}\label{PROP:CommutantOrthonormalityCompleteness}
		Let $G$ be a compact group and $D^{(R)}$ be an irrep of $G$.
		On $\mathcal H_{\rm hyb} = \mathcal H_{\rm int} \otimes \mathcal H_{\rm ext}$, the gauge action is
		$U(g) = D(g) \otimes \mathbf 1_{\rm ext}$.
		\begin{enumerate}
			\item Commutant in the hybrid space.
			
			By Schur’s lemma, the commutant $\{X\;|\;[X,U(g)]=0,~ \forall g\in G\}$ of local gauge action $U(g)$ is
			\begin{equation}\label{EQ:CommutantOfTheHybridGaugeAction}
				\mathrm{Comm}\bigl\{U(g)\bigr\}
				=
				\{X\;|\;[X,U(g)]=0, ~\forall g \in G\}
				=\mathbb C\,\mathbf 1_{\rm pkg} \otimes \mathcal B(\mathcal H_{\mathrm{ext}}).
			\end{equation}
			
			\item Orthonormality and completeness.
			
			The set
			$\bigl\{\ket{R,\alpha;\xi}\bigr\}_
			{\alpha=1,\dots,\dim R,\;\xi\in\mathcal I_{\mathrm{ext}}}$
			is an orthonormal basis of the hybrid space
			$$
			\mathcal H_{\rm hyb}
			:=\mathcal H_{\rm int}^{(R)}
			\otimes
			\mathcal H_{\mathrm{ext}}.
			$$
		\end{enumerate}
	\end{proposition}

	\begin{proof}
		\leavevmode
		\begin{enumerate}
			\item Since $D^{(R)}$ is irreducible, Schur’s lemma implies its commutant on $\mathcal H_{\rm int}$ is $\mathbb C\,\mathbf 1_{\rm pkg}$, that is,
			$\mathrm{Comm}\{D(g)\}=\mathbb C\mathbf 1_{\rm pkg}$.
			Since $G$ acts trivially on $\mathcal H_{\mathrm{ext}}$, the identity action on $\mathcal H_{\mathrm{ext}}$ leaves the full bounded-operator algebra $\mathcal B(\mathcal H_{\mathrm{ext}})$ untouched.
			This gives Eq. \eqref{EQ:CommutantOfTheHybridGaugeAction}.
			
			\item Using $\langle R,\beta|R,\alpha\rangle=\delta_{\beta\alpha}$ and
			$\langle\xi'|\xi\rangle=\delta_{\xi'\!,\xi}$, one obtains
			$
			\langle R,\beta;\xi'|R,\alpha;\xi\rangle
			=\delta_{\beta\alpha}\delta_{\xi'\!,\xi}.
			$
			Completeness follows from the defining tensor-product basis theorem.
		\end{enumerate}
	\end{proof}

	$\mathcal B(\mathcal H_{\mathrm{ext}})$ describes all physical
	manipulations (polarization rotations, beam-splitters, time-delays, $\cdots$) allowed by gauge symmetry.  
	Such operations can never reveal or alter the locked IQNs.
	The internal and external labels are completely factorized in Hilbert space structure, yet remain inseparable under gauge transformations.

	\begin{lemma}[Gauge covariance of hybrid operators]
		\label{LEM:HybridCovariance}
		For every $g \in G$, the single-particle hybrid-packaged creation operator $\hat A^{\dagger}_{\alpha;\xi}$ transforms as
		$$
		U(g) \,\hat A^{\dagger}_{\alpha;\xi}\,U(g)^{-1}
		= \sum_{\beta}
		D^{(R)}_{\alpha\beta}(g)\;
		\hat A^{\dagger}_{\beta;\xi}.
		$$
	\end{lemma}

	\begin{proof}
		Since $G$ acts as the irreducible matrix $D(g)$ on $\mathcal H_{\rm int}$ and as the identity on $\mathcal H_{\mathrm{ext}}$, gauge and Lorentz indices remain separated and irreducibility of the internal block is unchanged.
		Mathematically,
		\begin{align}
			\begin{aligned}\label{EQ:GaugeAction}
				U(g)\,\hat A_{\alpha;\xi}^{\dagger}\,U(g)^{-1}
				&=U(g)\bigl(\hat a^\dagger_{R,\alpha}\otimes\ket{\xi}\bigr)U(g)^{-1} \\
				&=\bigl(D(g)\hat a_{R,\alpha}^{\dagger}\bigr)\otimes\ket{\xi} \\
				&=\sum_{\beta}D^{(R)}_{\beta\alpha}(g) \hat A_{\beta;\xi}^{\dagger},
			\end{aligned}
		\end{align}
		This is identical to the transformation law for the internal operator in Stage 1.		
	\end{proof}

	This shows that, because $D^{(R)}$ was already irreducible, tensoring with a spectator space cannot reduce it.
	By unitarity, $D^{(R)}$ therefore stays irreducible on $\mathcal H_{\rm hyb}$.		
	In other words, hybridization multiplies the dimension of the one-particle internal space by
	$\dim\mathcal H_{\mathrm{ext}}$, but leaves the internal charge packet
	intact, so no partial factorization of IQNs appears.

	\subsubsection{Definition of Single-particle Hybrid-packaged States}
	\label{subsec:HybridPackagedStates}

	Stage 2 endows every single-particle internal block $\mathcal H_{\rm int}$ with a set of gauge-blind external labels.
	The object created is a hybrid-packaged state. 
	We now define these states precisely, verify their gauge covariance, and spell out the operator algebra that acts nontrivially on them.

	\begin{definition}[Single-particle hybrid-packaged state]
		\label{DEF:HybridState}
		Let $\{\ket{R,\alpha}\}_{\alpha=1}^{\dim R}$ be an orthonormal basis of the irreducible internal space $\mathcal H_{\rm int}$ and
		let	$\{\ket{\xi}\}_{\xi\in\mathcal I_{\mathrm{ext}}}$ be an orthonormal basis of the external Hilbert space $\mathcal H_{\mathrm{ext}}$ on which $G$ acts trivially.
		Define the \textbf{single-particle hybrid-packaged creation operator}
		\begin{equation}\label{EQ:SingleParticleHybridPackagedCreationOperator}
			\hat A^{\dagger}_{\alpha;\xi}
			\;:=\;
			\hat a^{\dagger}_{R,\alpha}\,\otimes\,\ket{\xi}
		\end{equation}
		that act on the global vacuum
		$\ket{\Omega}
		:= \ket{0}_{\rm pkg}\!\otimes\!\ket{0}_{\mathrm{ext}}$
		to obtain the one-particle state
		\begin{equation}
			\ket{R,\alpha;\xi}
			\;:=\;
			\hat A^{\dagger}_{\alpha;\xi}\ket{\Omega}
		\end{equation}
		By Eq. \eqref{EQ:SingleParticleHybridPackagedSubspace}, $\ket{R,\alpha;\xi}$ must fall into one of the hybrid-packaged subspace.
		Then we say that the vector $\ket{R,\alpha;\xi}$ is a \textbf{single-particle hybrid-packaged state}.
	\end{definition}

	\begin{example}[Examples of single-particle hybrid-packaged states]
		\leavevmode
		\begin{enumerate}
			\item Electron with spin.
			
			In QED, the internal irrep is one‐dimensional (with charge $-e$).
			The external space $\mathcal H_{\rm ext}$ is spanned by the two‐dimensional spin‐helicity basis, that is, $\mathcal H_{\rm ext} = \operatorname{span}\{\ket{\uparrow},\ket{\downarrow}\}$.
			Thus
			$$
			\ket{e^-;\mathbf p,\uparrow}
			=\hat a^\dagger_{e^-}(\mathbf p)\,\otimes\ket{\uparrow},
			\quad
			\ket{e^-;\mathbf p,\downarrow}
			=\hat a^\dagger_{e^-}(\mathbf p)\,\otimes\ket{\downarrow},
			$$
			span a two‐dimensional $\mathcal H_{\rm hyb}$.
			Gauge transformations act by the same phase $e^{-ie\alpha}$ on both basis vectors, which preserves the one‐dimensional internal block.
			
			\item Boson helicity.
			
			For the photon, the internal U(1) block is trivial.
			The creation operators		
			$$
			\hat a^{\dagger}_{\lambda}(\mathbf p),\qquad
			\lambda=\pm1,
			$$		
			transform as $U_{\alpha}\hat a^{\dagger}_{\lambda}U_{\alpha}^{-1}=
			\hat a^{\dagger}_{\lambda}$ while the helicity index $\lambda$ spans
			$\mathcal H_{\mathrm{ext}}=\mathbb C^{2}$.
			Thus the two polarization states form a hybrid-packaged doublet.
			
			\item Momentum label.
			
			Take a color singlet hydrogen atom field $H$ from Stage 1.
			Write a wave-packet creation operator		
			$$
			\hat h^{\dagger}(f)=\!\int\!\frac{d^{3}p}{(2\pi)^{3}2E_{\mathbf p}}\,
			f(\mathbf p)\,\hat h^{\dagger}(\mathbf p),
			$$		
			with $f\in L^{2}(\mathbb R^{3})$.
			The L$^2$ space is an instance of $\mathcal H_{\mathrm{ext}}$, and $G$ (here U(1)) acts trivially on the momentum label.
			Hence every normalize wave-packet is still a single packaged state, but the external multiplicity is now infinite-dimensional.
		\end{enumerate}		
	\end{example}

	\subsection{Packaging Survives Hybridization}

	Stage 1 locked all IQNs into a single
	irreducible $G$-block.	 
	Stage 2 then adjoined an external spectator space $\mathcal H_{\mathrm{ext}}$ (hybridization).  
	We now prove that neither the packaging nor irreducibility is disturbed by hybridization.

	\begin{proposition}[Stability under gauge‐invariant operations]
		Let $V$ be any bounded operator on the full Fock space that
		\begin{enumerate}
			\item Commutes with every local gauge transformation $U(g)$:
			$[V, U(g)] = 0 \;\forall g\in G$.
			
			\item Preserves the total charge $Q_0$:
			$[V,\hat Q]=0$. 
		\end{enumerate}
		Then $V$ maps $\mathcal H_{\rm hyb}$ into itself:
		$V \mathcal H_{\rm hyb} \subseteq \mathcal H_{\rm hyb}$.
	\end{proposition}

	\begin{proof}
		Since $V$ commutes with $U(g)$ and the total charge operator $\hat Q$, it is restricted inside a fixed charge sector and cannot mix different characters nor different charge sectors.
		
		By Proposition \ref{PROP:CommutantOrthonormalityCompleteness} and Lemma \ref{LEM:HybridCovariance}, the only way $V$ can act on $\mathcal H_{\rm hyb}$ is via an operator of the form $\mathbf 1_{\rm pkg}\otimes B$ for some $B$ on $\mathcal H_{\rm ext}$.
		In other words, $V$ takes any generic hybrid state		
		$$
		\ket\Psi=\sum_{\alpha,\xi}c_{\alpha\xi}\ket{R,\alpha;\xi}
		$$
		to again a linear combination of the same basis
		$$
		V\ket\Psi
		=\sum_{\alpha,\xi}c_{\alpha\xi}
		\ket{R,\alpha}\otimes B\ket{\xi}\in\mathcal H_{\rm hyb}.
		$$
		This proves invariance.
	\end{proof}

	Hence the entire internal plus external label transforms inside the same irrep $R$.

	\begin{theorem}[Stability of packaging under hybridization]
		Let $G$ be a compact (or finite) gauge group,
		\footnote{
			Compactness guarantees complete reducibility and the validity of
			Schur’s lemma. All physical gauge groups in particle physics satisfy it ($U(1)$, $SU(N)$, $SO(N)$, $\cdots$).
		}
		let $D^{(R)}: G \to {\rm U}(V_R)$ be an irreducible unitary representation carried by the single-particle internal space $V_R \cong \mathcal H_{\rm int}$,
		and let $\mathcal H_{\mathrm{ext}}$ be a Hilbert space on which $G$ acts trivially: $U(g)|_{\mathcal H_{\mathrm{ext}}}=\mathbf 1$.		
		Define the single-particle hybrid space		
		$$
		\mathcal H_{\rm hyb}
		:=V_R\otimes\mathcal H_{\mathrm{ext}}
		\quad\text{with gauge action}\quad
		U(g)=D(g)\otimes\mathbf 1 .
		$$		
		Then
		\begin{enumerate}
			\item Direct-sum structure:
			As a $G$-module,
			\begin{equation}\label{EQ:DirectSum}
				\mathcal H_{\rm hyb}
				\;\cong\;
				\bigoplus_{\xi\in\mathcal I_{\mathrm{ext}}}
				V_R^{(\xi)},
			\end{equation}	
			that is, $\mathcal H_{\rm hyb}$ is a multiplicity space of mutually orthogonal copies of the same irreducible block $R$, one for each external label $\xi$.
			
			\item No new internal factorization:
			Any bounded operator
			$P: \mathcal H_{\rm hyb} \to \mathcal H_{\rm hyb}$ that commutes with the gauge action, $[P, U(g)] = 0 ~ \forall g \in G$, has the form
			\begin{equation}\label{EQ:NoNewInternalFactorization}
				P=\mathbf 1\otimes P_{\mathrm{ext}},
				\quad
				P_{\mathrm{ext}}\in\mathcal B(\mathcal H_{\mathrm{ext}}).
			\end{equation}				
			In other words, no projector can split the internal index $\alpha$ while leaving the external label untouched.
			The IQNs therefore remain an indivisible packaged block after hybridization.
		\end{enumerate}		
	\end{theorem}

	\begin{proof}
		\leavevmode
		\begin{enumerate}
			\item Direct-sum structure.
			
			Since $G$ acts trivially on $\mathcal H_{\mathrm{ext}}$, the
			gauge module is	$D(g)\otimes\mathbf 1$.
			Choose an orthonormal basis	$\{\lvert\xi\rangle\}_{\xi\in\mathcal I_{\mathrm{ext}}}$ of $\mathcal H_{\mathrm{ext}}$.
			Then we have			
			$$
			U(g)\bigl(\lvert\alpha\rangle\otimes\lvert\xi\rangle\bigr)
			=(D^{(R)}_{\alpha\beta}(g)\lvert\beta\rangle)\otimes\lvert\xi\rangle .
			$$
			
			For each fixed $\xi$, 
			$\operatorname{span}\{\lvert\alpha\rangle\otimes\lvert\xi\rangle\}$ is
			isomorphic to $V_R$.
			Summing over $\xi$ gives the stated	direct sum of identical copies, that is, Eq. \eqref{EQ:DirectSum}.
			
			\item Commutant algebra.
			
			For fixed $\xi,\xi'$, we define operators
			$P_{\xi\xi'}:=\bigl(\mathbf 1\otimes\langle\xi|\bigr)P
			\bigl(\mathbf 1\otimes|\xi'\rangle\bigr):
			V_R\!\to\!V_R$.
			Since $\forall ~ g \in G$, $[P, U(g)] = 0$, we have
			$[P_{\xi\xi'}, D(g)] = 0$.
			Since $D^{(R)}$ is irreducible, Schur’s lemma gives
			$P_{\xi\xi'}=c_{\xi\xi'}\,\mathbf 1$ with complex scalars
			$c_{\xi\xi'}$.
			Therefore			
			$$
			P
			=\mathbf 1\otimes
			\sum_{\xi,\xi'}c_{\xi\xi'}\,|\xi\rangle\langle\xi'|
			=\mathbf 1\otimes P_{\mathrm{ext}},
			$$			
			This proves Eq.~\eqref{EQ:NoNewInternalFactorization}.
			In other words, any attempt to project onto a proper subspace of $V_R$ would contradict irreducibility, so the internal packaged block is intact.
		\end{enumerate}			
		This completes the proof that packaging survives hybridization.
	\end{proof}

	\begin{remark}		
		Hybridization only attaches spectator DOFs (spin, momentum, orbital label $\cdots$) to an already packaged internal block.
		Since $G$ ignores those labels, it cannot induce or detect any factorization inside the internal IQNs.
		They remain locked together.			
		Consequently:
		\begin{itemize}
			\item Creation operator form				
			$$
			\hat A_{\alpha;\xi}^{\dagger}
			:=\hat a_{R,\alpha}^{\dagger}\otimes\lvert\xi\rangle
			\quad\Longrightarrow\quad
			U(g)\,\hat A_{\alpha;\xi}^{\dagger}U(g)^{-1}
			=D^{(R)}_{\alpha\beta}(g)\,
			\hat A_{\beta;\xi}^{\dagger},
			$$			
			exactly the same irrep as before hybridization.
			
			\item Logical qudits.
			This stability makes the hybrid block the natural qudit for
			gauge-invariant quantum information and the elementary building brick for the many-body construction of Stage 3.
		\end{itemize}
	\end{remark}

	\section{Stage 3: Tensor Product (Multi-Particle Assembly)}
	\label{SEC:Stage3TensorProduct}

	Now we have the single-particle building blocks Eq. \eqref{EQ:SingleParticleHybridPackagedSubspace}:
	$$
		\mathcal H_{\text{hyb}}
		= \mathcal H_{\rm int} \otimes \mathcal H_{\rm ext}
		= \mathcal V_{\mathrm{pkg}} \otimes \mathcal S_{\rm spect},
	$$
	which carry the irreducible action
	$D^{(R)}: G \to \mathrm U(d_R)$ (Stage 1) while the external factor
	$\mathcal H_{\rm ext}$ is $G$-blind (Stage 2).

	In this stage, we show how $n$ such blocks combine \cite{Halvorson2000} and how packaging survives the formation of many-body product space.
	For simplicity, we use $\mathcal H_1$ to denote a single-particle building block. It can be either $\mathcal H_{\rm hyb}$, $\mathcal H_{\rm int}$, or $\mathcal H_{\rm ext}$, depending on the exact object.

	\subsection{Multi-particle Product Subspaces}
	\label{SEC:MultiParticleProductSubspaces}

	\subsubsection{Definition of Multi-particle Product Subspaces}

	Let us now consider multi-particle product subspace. 
	For $n$ particles, we use the canonical flip isomorphism
	$(\mathcal V_{\mathrm{pkg}} \otimes \mathcal S_{\rm spect})^{\otimes n}
	\cong
	\mathcal V_{\mathrm{pkg}}^{\otimes n}\!\otimes\!\mathcal S_{\rm spect}^{\otimes n}$
	and (anti)symmetrise to obtain
	$$
	\mathcal H_{\rm prod}^{(n)}
	:= \mathbb S_{\pm}\Bigl( \mathcal H_{\text{hyb}}^{\otimes n}\Bigr)
	= \mathbb S_{\pm}\!\Bigl(
	\mathcal V_{\mathrm{pkg}}^{\otimes n}\!\otimes\!\mathcal S_{\rm spect}^{\otimes n}\Bigr)
	\subset \mathcal H_{\rm single}^{\otimes n}.
	$$
	Based on this idea, we now give a formal definition of single-particle packaged subspace:

	\begin{definition}[$n$-particle product subspace]
		\label{DEF:MultiParticleHybridPackagedSubspace}
		Let $G_\mathrm{tot}=G_{\rm int}\times G_{\rm ext}$ be the total symmetry group and let $\mathcal H_{\text{hyb}}
		= \mathcal V_{\mathrm{pkg}} \otimes \mathcal S_{\rm spect} \subset \mathcal H_{\mathrm{single}}$ be a single-particle packaged subspace.
		For $n \ge 1$ identical excitations with the understanding that states are (anti)symmetrised (bosons: $+$, fermions: $-$) \cite{Pauli1940,Sakurai2020} according to the particle statistics, we define tensor product
		\begin{equation}\label{EQ:Hprod_general}
			\mathcal H_{\mathrm{prod}}^{(n)}
			:= 
			\mathbb S_{\pm}\Bigl(\mathcal V_{\mathrm{pkg}}^{\;\otimes n}\Bigr)
			\otimes
			\mathcal S_{\rm spect}^{\;\otimes n}.
		\end{equation}
		Then we say that $\mathcal H_{\mathrm{prod}}^{(n)}$ is an \textbf{$n$-particle product subspace} (or an \textbf{$n$-particle Fock sector}) and the direct sum
		$$
		\mathcal H_{\rm Fock}
		:=\bigoplus_{n\ge0}\mathcal H_{\mathrm{prod}}^{(n)}
		$$
		is the \textbf{total Fock space}.	
	\end{definition}

	On every single-particle space, the gauge-symmetry action is
	$
	U(g)=D(g)\otimes\mathbf 1_{\rm spect}
	$.
	On the $n$-particle product subspace $\mathcal H^{(n)}_{\rm prod}$, the gauge-symmetry action is then
	$$
	U^{(n)}_g =
	D(g)^{\otimes n}\!\otimes
	\mathbf 1_{\rm spect}^{\otimes n}.
	$$

	\begin{example}[Example of product subspace: electron + momentum]		
		We have internal packaged IQNs (U(1) charge $-e$)
		$$
		V_{R_{\rm int}}=\ket{e^{-}}
		$$  
		and free momentum lives in $\mathcal S_{\rm spect}$.
	\end{example}

	\begin{remark}[Identical‐particle (anti)symmetrization]
		When two or more packaged single-particle spaces are tensored together, we must project onto the appropriate symmetry sector if the particles are identical:
		\begin{itemize}
			\item Bosons ($\ell\in\mathbb{Z}$) live in the totally symmetric subspace
			
			$$
			\mathcal{H}^{(n)}_{\rm sym}
			= \mathcal{S}\,\bigl(\mathcal{H}^{(1)\,\otimes n}\bigr)
			= \frac{1}{n!}\sum_{\sigma\in S_n}
			P(\sigma)\,\mathcal{H}^{(1)\,\otimes n}\,,
			$$
			
			where $P(\sigma)$ permutes the tensor factors.
			
			\item Fermions ($\ell\in\mathbb{Z}+\tfrac12$) live in the totally antisymmetric subspace
			
			$$
			\mathcal{H}^{(n)}_{\rm asym}
			= \mathcal{A}\,\bigl(\mathcal{H}^{(1)\,\otimes n}\bigr)
			= \frac{1}{n!}\sum_{\sigma\in S_n} \operatorname{sgn}(\sigma)\,P(\sigma)\,\mathcal{H}^{(1)\,\otimes n}\,.
			$$
		\end{itemize}

		Since each packaged label $\alpha_i$ already bundles together species, spin/helicity, and any IQNs, we can obtain the projectors $\mathcal{S}$ or $\mathcal{A}$ by simply permuting those $\alpha$-indices.
		Specifically, we replace an $n$-particle state
		$
		\Psi_{\alpha_1}(p_1)\otimes\cdots\otimes\Psi_{\alpha_n}(p_n)
		$
		by
		$$
		\frac{1}{n!}\sum_{\sigma\in S_n}
		(\pm1)^{\!\sigma}\,
		\Psi_{\alpha_{\sigma(1)}}(p_{\sigma(1)})\otimes\cdots\otimes\Psi_{\alpha_{\sigma(n)}}(p_{\sigma(n)})\,,
		$$
		with the ``$+$'' for bosons and ``$-$'' for fermions.
		
		If the packaging map $\{i,p\}\!\mapsto\!\alpha$ already includes a species index $i$, then different species automatically live in different tensor factors and do not get permuted.
		Only truly identical $\alpha$-labels (same $i$, same spin/helicity) enter the symmetrizer/antisymmetrizer.
		In any combinatorial factors or normalization, make sure that you consistently include the $1/n!$ from the projector, or equivalently absorb it into your Fock-space CCR/CAR algebra.
	\end{remark}

	Each element in the $n$-particle product subspace (Definition \ref{DEF:MultiParticleHybridPackagedSubspace}) is an $n$-particle product state.
	Formally:

	\begin{definition}[$n$-particle product state]
		Let 
		$
		\mathcal H_{\rm prod}^{(n)}
		:=\mathbb S_{\pm}\!\Bigl(
		\mathcal H_{\rm int}^{\otimes n}\!\otimes\!
		\mathcal H_{\rm ext}^{\otimes n}\Bigr)
		$
		be the $n$-particle product subspace (Def.~\ref{DEF:MultiParticleHybridPackagedSubspace}).
		Choose external kets $\ket{\xi_k}\in\mathcal H_{\rm ext}$ and define		
		$$
		\hat A_{k}^{\dagger}:=
		\hat a^{\dagger}_{k}(\mathbf p_{k},\alpha_{k}) \otimes \ket{\xi_k},
		$$
		where each $\hat a^{\dagger}_{k}(\mathbf p_{k},\alpha_{k})$ is a single-particle internal-packaged creation operator.
		Then we say that the vector
		\begin{equation}\label{EQ:MultiparticleHybridPackagedProductState}
			\ket{\Theta}
			=\hat A_{1}^{\dagger}\hat A_{2}^{\dagger}\!\cdots\!
			\hat A_{n}^{\dagger}\ket{0}
			\in\mathcal H_{\rm prod}^{(n)}
		\end{equation}	
		is an \textbf{$n$-particle product state}.
	\end{definition}

	Since the $\ket{\xi_k}$ are gauge blind,
	$U(g)\ket{\Theta}=\bigl[D^{(R)}(g)\bigr]^{\otimes n}\ket{\Theta}$,
	that is, $\ket{\Theta}$ lies in the reducible representation $R^{\otimes n}$.
	Only after performing the isotypic decomposition
	$\mathcal H_{\rm prod}^{(n)}\to\bigoplus_\lambda\bigl(V_\lambda\otimes\C^{m_\lambda}\bigr)$
	can we project onto one irreducible subspace $V_\lambda$, at which point the multi‐particle state becomes fully packaged in that single irrep block.

	\begin{example}[Examples of product states]
		\leavevmode
		\begin{enumerate}
			\item Two electrons + spin.
			
			$$
			\ket{\Theta}
			=\hat a_{e^-}^{\dagger}(\mathbf p_1)\!\otimes\!\ket{\uparrow}\;
			\hat a_{e^-}^{\dagger}(\mathbf p_2)\!\otimes\!\ket{\downarrow}\;\ket0
			\;\in\;\mathcal H_{\rm prod}^{(2)} .
			$$
			
			The charge part is packaged ($-e,-e$) and the spin singlet/triplet structure lives entirely in $\mathcal H_{\rm ext}^{\otimes2}$.
			
			\item Quark-antiquark + orbital modes.
			A color triplet $q^{\alpha}$ and antitriplet
			$\bar q_{\beta}$ created in two distinct optical cavities
			$A,B$ give
			
			$$
			\ket{\Theta}
			=\bigl(\hat q^{\alpha\dagger}\otimes\ket{A}\bigr)
			\bigl(\hat{\bar q}_{\beta}^{\dagger}\otimes\ket{B}\bigr)\ket0 .
			$$
			
			Color indices $\alpha,\beta$ form the
			$\mathbf 3\otimes\bar{\mathbf 3}=\mathbf 1\oplus\mathbf8$ irrep,
			while the cavity labels $\ket{A},\ket{B}$ may be superposed
			independently.
			
			\item Photon polarization pair.
			With helicities $\lambda=\pm1$ as external labels,
			$$
			\ket{\Theta} = \hat a_{\lambda_1}^{\dagger}(\mathbf p_1)\ket{\lambda_1}\;
			\hat a_{\lambda_2}^{\dagger}(\mathbf p_2)\ket{\lambda_2}\ket0
			$$
			belongs to $\mathcal H_{\rm prod}^{(2)}$.
			The U(1) internal block is trivial, but entanglement in
			$\lambda_{1,2}$ is entirely in the external factor.
		\end{enumerate}
	\end{example}

	\subsubsection{Properties of Multi-particle Product Spaces}

	Here we study some properties of product subspaces
	$\mathcal H^{(n)}_{\rm prod} = \mathbb S_{\pm} \bigl(\mathcal H_{\rm int}^{\otimes n} \otimes \mathcal H_{\rm ext}^{\otimes n}\bigr)$,
	where $\mathbb S_{\pm}$ is the (anti)symmetriser imposed by particle statistics.
	For a product state, everything that matters is inherited from the
	single-particle building blocks studied in Stages 1 \& 2.

	\begin{proposition}[Tensor product of packaged irreps]
		\label{PROP:TensorProductPackagedSubSpaces}
		Let $G$ be a compact group and $V_\lambda,\;V_\mu$
		be two irrep spaces (internal-packaged subspaces).
		Then
		$$
		V_\lambda \otimes V_\mu
		\;\cong\;
		\bigoplus_{\nu\in\widehat G}\;
		N_{\lambda\mu}^{\;\nu}\,V_\nu,
		$$
		where $N_{\lambda\mu}^{\;\nu}\in\mathbb N$ are the multiplicities.
		\begin{enumerate}
			\item If $G$ is Abelian, then every $V_\lambda$ is one-dimensional and the sum has at most one term.
			$V_\lambda \otimes V_\mu$ is again a internal-packaged subspace that carries charge $Q_\lambda+Q_\mu$.
			
			\item For general $G$, $V_\lambda\otimes V_\mu$ is not internal-packaged. Packaging is recovered only after projecting onto one of the $V_\nu$ summands.
		\end{enumerate}
	\end{proposition}

	\begin{proof}
		\leavevmode
		\begin{enumerate}
			\item For Abelian case, the tensor product $V_\lambda \otimes V_\mu$ carries total charge $Q_\lambda + Q_\mu$. 
			$\forall~ g \in G$, the action is
			$U(g) \otimes U(g) = \chi(g) \chi(g) {\bf1} \otimes {\bf1}$,
			which is a global phase.
			Thus, the tensor product is internal‑packaged.
			This is strictly true when all single-particle charges are one-dimensional irreps, say U(1). But it is not good for $\mathbb Z_N$ with projective reps.
			
			\item For non-Abelian case, we will discuss in Stage 4.
		\end{enumerate}
	\end{proof}

	\begin{lemma}[Gauge covariance of product vectors]
		\label{LEM:GaugecovarianceOfPackagedProductVectors}
		Let $G$ be any compact group (Abelian or non-Abelian).
		Consider an $n$-particle product state	
		$$
		\ket{\Theta}
		\in
		\begin{cases}
			\mathcal H^{(n)}_{\rm pkg}, &\text{internal version},\\[4pt]
			\mathcal H^{(n)}_{\rm hyb}, &\text{hybrid version},
		\end{cases}
		\qquad
		\ket{\Theta}
		=\hat A^{\dagger}_{1}\hat A^{\dagger}_{2}\!\cdots\!\hat A^{\dagger}_{n}\ket0,
		$$	
		with $\hat A^{\dagger}_{k} =
		\hat a^{\dagger}_{k}(\mathbf p_k,\alpha_k)$ (internal) or
		$\hat a^{\dagger}_{k}(\mathbf p_k,\alpha_k)\!\otimes\!\ket{\xi_k}$
		(hybrid).
		Then	
		$$
		U(g)^{(n)}\,\ket{\Theta}
		=\Bigl[D(g)\Bigr]^{\!\otimes n}\ket{\Theta},
		\qquad\forall\,g\in G,
		$$	
		that is, $\ket{\Theta}$ transforms in a well-defined (generally reducible) representation of $G$.
	\end{lemma}

	\begin{proof}
		We split the proof in two steps:
		
		\begin{enumerate}
			\item Single-particle packaged state.			
			According to Theorem \ref{THM:NoPartialFactorization} and Lemma \ref{LEM:HybridCovariance}, each creation operator $\hat A^{\dagger}_{k}$ is gauge covariant, that is, transforms under an irrep of $G$.  
			Specifically, $\forall ~ g \in G$,
			$$
			U(g)\,\hat A^{\dagger}_{k}\,U(g)^{-1}
			\;=\;
			\sum_{q'} \bigl[D^{(k)}_g\bigr]_{q'\,q_k}\,
			\hat A^{\dagger}_{k},
			$$
			where $[D^{(k)}_g]$ is the matrix representation of $g$ in the $k$th single-particle irrep.

			\item $n$-particle product state.
			In the internal case
			$U(g)^{(n)}=\bigl[D(g)\bigr]^{\!\otimes n}$.
			In the hybrid case
			$U(g)^{(n)}=(D(g))^{\otimes n}\otimes\mathbf 1_{\rm ext}^{\otimes n}$
			because $G$ acts trivially on $\mathcal H_{\rm ext}$.
			Acting on the ordered product and on the gauge-invariant vacuum
			$\ket0$ yields the stated result.
		\end{enumerate}
	\end{proof}

	\begin{corollary}[Total charge / irrep content]
		\leavevmode
		\begin{itemize}
			\item Abelian $G=U(1)$.
			The eigenvalue of the global charge operator in	$\ket{\Theta}$ is the algebraic sum of the individual charges:
			$Q_{\mathrm{tot}}=\sum_{k}q_{k}$.
			
			\item Non-Abelian $G$.
			The multi-particle state lives in the (generally reducible) tensor power $D^{(R)\otimes n}$.
			A subsequent isotypic decomposition (Stage 4) splits it into a direct sum of irreps $\bigoplus_{\lambda} N_{\lambda} D_{\lambda}$.
		\end{itemize}
	\end{corollary}

	\begin{example}[Hybrid electron singlet]
		Take two electrons ($q_{1}=q_{2}=-e$) with opposite spin:
		$$
		\ket{\Theta}
		=\Bigl(\hat a_{e^-}^{\dagger}(\mathbf p_1)\!\otimes\!\ket{\uparrow}\Bigr)
		\Bigl(\hat a_{e^-}^{\dagger}(\mathbf p_2)\!\otimes\!\ket{\downarrow}\Bigr)
		\ket0
		\;\in\;\mathcal H^{(2)}_{\rm hyb}.
		$$
		\begin{itemize}
			\item Gauge part.
			Both creation operators transform by the 1-dimensional phase
			$e^{-ie\alpha}$.
			We have
			$U_\alpha^{(2)}\ket{\Theta}=e^{-2ie\alpha}\ket{\Theta}$
			This shows that the internal charge is packaged.
			
			\item External part.
			The spin DOF lives in
			$\mathcal H_{\rm ext}^{\otimes2}\cong\mathbb C^{2}\otimes\mathbb C^{2}$.
			Forming the antisymmetric combination
			$\ket{\uparrow\downarrow}-\ket{\downarrow\uparrow}$ produces the usual
			spin singlet, without ever disturbing the packaged charge.
		\end{itemize}
		Thus, $\ket{\Theta}$ separates what the gauge group
		cares about (total charge $-2e$) from the freely manipulable spin
		qubit, which exemplifies the utility of product states.
	\end{example}

	These results show that every many-body product state, internal or
	hybrid, is automatically gauge covariant and keeps each particle’s IQNs
	locked inside its original irreducible block.

	\subsection{Packaging Survives Tensor Product}

	Using the tensor product, we have assembled single-particle packaged blocks into $n$-particle product subspaces.	
	We now show that tensoring these blocks cannot split the internal irreducible packet.
	Equivalently, any operator commuting with $G$ acts trivially on each packet factor and preserves its irreducibility.

	\begin{theorem}[Stability of packaging under tensor product]
		Let $G$ be a compact gauge group and let $D^{(R)}:G\to\mathrm U(d_R)$ be an irrep carried by each single-particle internal space $V_R$.
		For $n$ particles, consider the combined Hilbert space		
		$$
		\mathcal H^{(n)}=\bigotimes_{k=1}^n\bigl(V_R\otimes\mathcal H_{\rm ext}\bigr),
		$$		
		where $\mathcal H_{\rm ext}$ is any external factor on which $G$ acts trivially, and $G$ acts by $U(g)^{(n)}=D^{(R)}(g)^{\otimes n}\otimes\mathbf 1_{\rm ext}^{\otimes n}$.
		Then:
		\begin{enumerate}
			\item Any single-particle irreducibility remains:
			for each factor $k$, the restriction of $U(g)^{(n)}$ to that $V_R$ block is still $D^{(R)}(g)$.
			By Schur's lemma, no nontrivial projector can split $V_R$ into smaller $G$-invariant subspaces.
			
			\item The full $n$-body tensor product carries the reducible representation $D^{(R)\otimes n}(g)$, but each individual excitation's internal label is still locked in an irreducible packet.
		\end{enumerate}
		Equivalently, tensoring packaged blocks cannot introduce any new partial factorization of IQNs.
	\end{theorem}

	\begin{proof}
		\leavevmode		
		\begin{enumerate}
			\item Irreducibility of each factor under the combined action.
			Fix one particle $k$.
			The gauge action on that factor is			
			$$
			U(g)^{(n)}\Bigl|_{V_R^{(k)}}=D^{(R)}(g)\otimes\mathbf 1
			\otimes \cdots \otimes \mathbf 1,
			$$			
			tensoring identities on all other factors.
			Since $D^{(R)}$ is irreducible, by Schur's lemma any operator $P$ acting on $V_R$ that commutes with $D^{(R)}(g)$ must be a scalar multiple of the identity.
			Hence there is no nontrivial $G$-invariant subspace in $V_R$ that could split the packet of IQNs.
			
			\item Integrity of the combined representation.
			Since each single-particle block remains irreducible and the external factors are spectators, the full $n$-body representation is simply the tensor product of these irreducibles.
			No partial factorization of any individual $V_R$ factor is introduced: each excitation's internal packet remains intact.
			Concretely, any operator on $\mathcal H^{(n)}$ that commutes with all $U(g)^{(n)}$ must act as the identity on every $V_R$ factor (by repeated use of Schur's lemma) and may only nontrivially affect the external Hilbert spaces, which do not carry gauge charge.
		\end{enumerate}
		These two facts together establish that tensor-product assembly preserves packaging: no IQN packet can be split by forming multi-particle product states.
	\end{proof}

	\section{Stage 4: Isotypic Decomposition (Multi-particle Packaging)}
	\label{SEC:Stage4IsotypicDecomposition}

	In Stage 3, we assembled the single-particle packaged states (irreducible blocks) into multi-particle product states (reducible blocks) (see Eq. \eqref{EQ:MultiparticleHybridPackagedProductState}).
	Proposition \ref{PROP:TensorProductPackagedSubSpaces} shows that, if the gauge group $G$ is Abelian, then the product space is again packaged.
	But if $G$ is non-Abelian, then the product space is not packaged.
	The states in this multi-particle space remain uncoupled with respect to gauge charge, that is, they are reducible and are not packaged at multi-particle level.
	However, the packaging won't stop here and will further form multi-particle packaging due to the complete reducibility \cite{Maschke1898,Maschke1899,PeterWeyl1927} of representation theory.

	In this stage, we introduce isotypic decomposition \cite{Knapp1986,FultonHarris2004} that splits the reducible raw-Fock space into irreducible packaged subspaces (charge sectors), or reorganizes the reducible blocks into coupled irreducible blocks.	
	We will prove an important result:
	packaging directly leads to charge superselection rules (implemented by the Peter-Weyl projection \cite{PeterWeyl1927,Knapp1986}).
	We also demonstrate that packaging remains robust under isotypic decomposition.
	This stage enters the second (isotypic-sector) layer where multi-particle form irrep blocks.

	\subsection{Isotypic Decomposition and Multi-particle Packaging}

	\subsubsection{Isotypic Decomposition}

	We already mentioned isotypic decomposition of single-particle internal spaces in subsection \ref{SEC:IsotypicDecompositionSingleParticleInternalSubspaces}.
	We expand it here with technical details so that it seamless fit to our discussions on Peter-Weyl projection and multi-particle packaging.

	Let $G$ be a compact group
	\footnote{Everything extends to finite groups by replacing Haar integrals with finite sums.}
	with normalized Haar measure $\mu$, and let
	$$
	U \colon G \longrightarrow \mathcal U\bigl(\mathcal H_{\rm Fock}\bigr),
	\quad
	(g \longmapsto U_g)
	$$
	be the continuous unitary representation of $G$ on the raw-Fock space $\mathcal H_{\rm Fock}$.
	Let $\widehat G$ be the countable set of unitary irreps of $G$ and 
	$d_\lambda = \dim V_\lambda$ be the dimension of $\lambda \in \widehat G$.
	Then by general representation-theoretic principles, one obtains an isotypic decomposition
	$$
	\mathcal H_{\rm iso}
	\;\cong\;
	\bigoplus_{\lambda\in\widehat G}
	\Bigl(V_\lambda\otimes M_\lambda\Bigr)
	\;=\;
	\bigoplus_{\lambda\in\widehat G}
	\mathcal H_\lambda,
	$$
	where each $V_\lambda$ is a carrier space of the irrep $\lambda$ of $G$,
	$M_\lambda=\mathrm{Hom}_G(V_\lambda,\mathcal H_{\rm Fock})$ is the multiplicity space of dimension $m_\lambda$,
	and	each summand $\mathcal H_\lambda = V_\lambda \otimes M_\lambda$ is an isotypic component of type $\lambda$.

	\begin{definition}[Isotypic decomposition]\label{DEF:IsotypicDecomposition}
		Let $G, U$ be defined as above.
		Then there exists a unique splitting
		\begin{equation}\label{EQ:IsotypicDecomposition}
			\mathcal H_{\rm Fock} \longrightarrow \mathcal H_{\rm iso}
			= \bigoplus_\lambda V_\lambda \otimes M_\lambda
			= \bigoplus_\lambda \mathcal H_\lambda
		\end{equation}
		into maximal $G$-invariant subspaces $\mathcal H_\lambda$.
		Each $\mathcal H_\lambda = V_\lambda \otimes M_\lambda$ only carries copies of a single irrep $\lambda$.
		We say that Eq. \eqref{EQ:IsotypicDecomposition} is the \textbf{isotypic decomposition} of the $G$-module $\mathcal H_{\rm Fock}$.
	\end{definition}

	Since $G$ is purely internal, all external labels (momentum, spin, Fock occupancy, spatial wavefunctions, etc.) reside in $M_\lambda$.

	\begin{example}[Decompositions of color structure]		
		In QCD, the color space decompositions can be written both in Young-diagram notation (packaged irrep blocks) and in the language of second quantization.
		Let	$a_i^\dagger ~ (i=1,2,3)$ be the quark creation operators carrying color in the fundamental irrep block $\mathbf3$ and $b^i{}^\dagger ~ (i=1,2,3)$ be the antiquark operators in the irrep block $\overline{\mathbf3}$ (we raise the index on $b^\dagger$ to remind ourselves it transforms with the conjugate representation).
		They satisfy the usual anticommutation relations for fermions.
		Then:		
		\begin{enumerate}
			\item \textbf{Two quarks:}
			
			\textbf{Young-diagram notation:	}
			$\mathbf{3}\otimes\mathbf{3} = \mathbf{6}\oplus\overline{\mathbf{3}}$, where the symmetric $\mathbf{6}$ and the antisymmetric $\overline{\mathbf{3}}$ represent the only possible ways to combine two $\mathbf{3}$’s.
			
			\textbf{Second quantization:}
			The two‐quark creation operator		
			$
			Q_{ij}^\dagger = a_i^\dagger\,a_j^\dagger 
			$		
			and then decomposes into its symmetric and antisymmetric parts:
			\begin{itemize}
				\item Symmetric sextet $\mathbf6$:
				$$
				S_{(ij)}^\dagger \,\lvert0\rangle
				=\;\frac{1}{\sqrt2}\,\bigl(a_i^\dagger\,a_j^\dagger + a_j^\dagger\,a_i^\dagger\bigr)\,\lvert0\rangle\
				\quad (i \le j),
				$$
				which is the dimension 6 irrep.
				
				\item Antisymmetric antitriplet $\overline{\mathbf3}$:
				$$
				A_{[ij]}^\dagger \,\lvert0\rangle
				=\frac{1}{\sqrt2}\,\bigl(a_i^\dagger\,a_j^\dagger - a_j^\dagger\,a_i^\dagger\bigr) \lvert0\rangle
				\quad \text{or}\quad
				D_k^\dagger\,\lvert0\rangle
				=\;\frac{1}{\sqrt2}\,\epsilon_{kij}\,a_i^\dagger\,a_j^\dagger\,\lvert0\rangle
				\quad (i<j),
				$$
				which yields three independent components that transform as $\overline{\mathbf3}$.
			\end{itemize}
			
			\item \textbf{Quark-antiquark:}
			
			\textbf{Young-diagram notation:}
			$\mathbf{3}\otimes\overline{\mathbf{3}} = \mathbf{1}\oplus\mathbf{8}$, where the singlet $\mathbf{1}$ is color neutral and corresponds to the observed mesons, while the octet $\mathbf{8}$ remains confined.
			
			\textbf{Second quantization:}
			The quark and antiquark creation operators form the bilinear		
			$
			M_i{}^j \;=\; a_i^\dagger\,b^{\,j\dagger}
			$		
			and then decompose into trace (singlet) and traceless (octet) parts:
			\begin{itemize}
				\item Singlet $\mathbf{1}$:
				$$
				S^\dagger\,\lvert0\rangle
				=\;\frac{1}{\sqrt3}\,\delta^i_j\,a_i^\dagger\,b^j{}^\dagger\,\lvert0\rangle,
				$$
				which is invariant under $SU(3)$.
				
				\item Octet $\mathbf 8$:
				$$
				O_a^\dagger\,\lvert0\rangle
				=\;a_i^\dagger\,(T^a)^i{}_j\,b^j{}^\dagger\,\lvert0\rangle,
				\quad a=1,\dots,8
				$$
				where $T^a$ are the Gell‐Mann generators (${\bf8}$ generators of SU(3)), and the tracelessness $\operatorname{Tr}(T^a)=0$ ensures these eight combinations form the adjoint.
			\end{itemize}

			\item \textbf{Larger Products:}
			
			For a larger product, say a baryon ($qqq$), we have
			
			\textbf{Young-diagram notation:}
			$
			\mathbf3\otimes\mathbf3\otimes\mathbf3 \;\supset\;\mathbf1.
			$
			We usually combine a diquark (e.g. a $\overline{\mathbf{3}}$) with another quark $\mathbf{3}$ to yield a singlet.
			
			\textbf{Second quantization:}
			One first picks out the antisymmetric $\overline{\mathbf3}$ from two quarks,	
			$$
			A_{[ij]}^\dagger
			= \frac1{\sqrt2}(a_i^\dagger a_j^\dagger - a_j^\dagger a_i^\dagger))\,,
			$$		
			then tensor with a third quark to form a total singlet via the fully antisymmetric epsilon tensor:		
			$$
			B^\dagger_{\mathbf1}
			\;=\;
			\frac{1}{\sqrt6}\,
			\varepsilon^{\,ijk}\,
			a_i^\dagger\,a_j^\dagger\,a_k^\dagger.
			$$		
			This is exactly the color singlet baryon operator.

			By iterating the same rules (Peter-Weyl projection), one can build any multi‐quark or multi-quark-antiquark singlet or higher‐dim irreps.
			In every case, one writes down the appropriate symmetrized or antisymmetrized polynomials in the $a^\dagger$ and $b^\dagger$, projects out traces with $\varepsilon_{ijk}$ or $\delta^i_j$, and so recovers the familiar decomposition of tensor products in second‐quantized language.
		\end{enumerate}
	\end{example}

	\subsubsection{Multi-particle Packaging}
	\label{SEC:MultiParticlePackaging}

	Under the isotypic decomposition, the raw-Fock space splits into a direct sum of isotypic sectors $\mathcal H_{\rm \lambda} = V_\lambda\otimes M_\lambda$ (for irreducible decomposition, please see Sec. \ref{SEC:ClebschGordanRotation}).
	Each total irrep $V_\lambda$ appears inside one and only one isotypic block	$\mathcal H_{\rm \lambda}$.
	All states in this subspace carry the same overall $G$-charge (or highest weight).
	Since each isotypic subspace is irreducible, all multi-particle states in this subspace are packaged states.
	In this sense, multi-particle packaging is born in this stage.
	This is an extension of the packaging from single-particle level (Sec. \ref{SEC:SingleParticlePackagingBorn}) to multi-particle level.
	We define:

	\begin{definition}[Multi-particle packaging]\label{AX:MultiPkg}
		Under the action of a local gauge group $G$, a multi-particle subspace $\mathcal H_{\rm pkg}$ may carries exactly one irrep $V_\lambda$ of $G$ locked by local gauge invariance and no physical process can split $V_\lambda$ into smaller pieces.
		Then we refer this irreducibility as \textbf{multi-particle packaging}.
	\end{definition}

	Definition \ref{AX:SinglePkg} forces each elementary excitation to carry a sharp irreducible charge.
	Definition \ref{AX:MultiPkg} insists that any multi-particle state lies entirely in one fixed-charge (isotypic) block.
	These two requirements together imply the superselection statements of Section \ref{SEC:PackagingImpliesSuperselection}.

	\subsection{Multi-particle Packaged Subspaces}
	\label{SEC:MultiParticlePackagedSubspaces}

	From above subsection, we known that each isotypic sector is an irreducible subspace.
	This means that every state in an isotypic sector is a packaged state, but include multiple particles (comparing with the single-particle packaged states given in Definition \ref{DEF:SinglePackaging}).
	Therefore, it is necessary to give a formal definition for such subspaces and discuss their algebraic and physical properties.

	\subsubsection{Definition}

	We first give the definition of a multi-particle packaged subspace:

	\begin{definition}[Multi-particle packaged subspace]\label{DEF:MultiParticlePackagedSubspace}
		If a multi-particle subspace $\mathcal H_{\rm pkg}$ forms an irreducible block of the symmetry group $G$ and partial factorization of the packaged DOFs is prohibited, then we say that $\mathcal H_{\rm pkg}$ is a \textbf{multi-particle packaged subspace} under $G$.
	\end{definition}

	\begin{example}
		Examples of multi-particle packaged subspaces:
		
		\begin{enumerate}
			\item U(1) charge sectors.
			For an abelian symmetry $G=U(1)$, one has
			$$
			\mathcal H \;=\; \bigoplus_{Q\in\mathbb Z} \mathcal H_Q,
			$$
			where $\mathcal H_Q$ is the subspace of all $n$-particle states with net charge $Q$.  Each $\mathcal H_Q$ is a multi-particle packaged subspace.
			
			\item $SU(2)$ rotations.
			For a two-electron p-shell, $\mathcal H_{J=0}$ (singlet) and $\mathcal H_{J=1}$
			(triplet) are distinct packaged subspaces of external labels.
			
			\item $SU(3)$ color in QCD.
			Here
			$\mathcal H = \bigoplus_{\lambda\in\{\mathbf1,\mathbf3,\mathbf{\bar3},\mathbf8,\dots\}} \mathcal H_\lambda,$
			where $\mathcal H_{\mathbf1}$ is the color singlet sector, $\mathcal H_{\mathbf8}$ the octet sector, etc.
		\end{enumerate}		
	\end{example}

	\begin{remark}
		\leavevmode
		\begin{enumerate}
			\item All physical multi-particle states that respect $G$ must lie
			inside one of the $\mathcal H_\lambda$.
			
			\item No symmetry-preserving operation can coherently superpose vectors belonging to different $\lambda$’s.
			This observation is the seed of the superselection rule
			proved in Sec \ref{SEC:PackagingImpliesSuperselection}.
		\end{enumerate}
	\end{remark}

	\subsubsection{Properties of Multi-particle Packaged Subspaces}

	According to the isotypic decomposition, we immediately obtain the following proposition:	
	
	\begin{proposition}[Properties of isotypic sectors]
		\label{PROP:PropertiesIsotypicSectors}
		Let
		$$
		\mathcal H_{\rm iso}
		\;=\;
		\bigoplus_{\lambda\in\widehat G}
		\mathcal H_\lambda
		\;\cong\;
		\bigoplus_{\lambda\in\widehat G}
		\Bigl(V_\lambda\otimes M_\lambda\Bigr)
		$$
		be the isotypic decomposition of the full (multi-particle) Hilbert space under a symmetry $G$ (see Definition \ref{DEF:IsotypicDecomposition}).
		Then:
		
		\begin{enumerate}
			\item Multi-particle packaged subspace.
			Each summand
			$$
			\mathcal H_\lambda \;\cong\; V_\lambda\otimes M_\lambda
			$$
			is a multi-particle packaged subspace of total charge (or highest weight) $\lambda$.
			
			\item Particle-number grading.
			Each $\mathcal H_\lambda$ further decomposes as
			$$
			\mathcal H_\lambda
			=\bigoplus_{n\ge 0}\mathcal H^{(n)}_\lambda,
			$$
			where $n$ is particle number (and any other commuting conserved label).
			The group action remains $U(g)=\rho_\lambda(g)\otimes\mathbf 1$ on each block.
			
			\item Indivisibility.
			Within a single $\mathcal H_\lambda$, there exists only trivial $G$-equivariant projection onto a proper subspace of $V_\lambda$.  Equivalently, the $V_\lambda$-factor is an indivisible package of DOFs.
		\end{enumerate}
	\end{proposition}

	\begin{proof}
		\leavevmode
		\begin{enumerate}
			\item According to Definition \ref{DEF:IsotypicDecomposition}, each summand $\mathcal H_\lambda$ is an irreducible block. Then by Definition \ref{DEF:MultiParticlePackagedSubspace}, $\mathcal H_\lambda$ is a multi-particle packaged subspace.
			
			\item Particle number commutes with $G$ by assumption, so the $n$-grading
			refines but does not break the isotypic split.
			
			The Fock space decomposes as $\bigoplus_{n\ge0}\mathcal H^{(n)}$.
			Each $\mathcal H^{(n)}$ itself splits under $G$ into isotypic components $\mathcal H^{(n)}_\lambda$.
			Summing over $n$ restores $\mathcal H_\lambda$.
			Within each fixed $n$, the same irreducibility argument applies.
			
			\item Suppose there exists a nontrivial projector $P$ on $\mathcal H_\lambda$ with 
			$[P,U(g)]=0$ for all $g$, and such that $P(V_\lambda)\subsetneq V_\lambda$.
			But this would break the irreducibility of $\rho_\lambda$ and therefore contradicts its minimality.
			Thus no such $P$ exists.
		\end{enumerate}
	\end{proof}

	\begin{proposition}[Action on isotypic sectors]
		\label{PROP:ActionOnIsotypicSectors}
		\leavevmode
		\begin{enumerate}
			\item Irreducible $G$-action.
			In a fixed $\mathcal H_\lambda = V_\lambda\otimes M_\lambda$, the group $G$ acts only on the factor $V_\lambda$ by the irrep $\rho_\lambda$.
			Consequently, every $G$-respecting operator
			$O:\mathcal H_\lambda\to\mathcal H_\lambda$
			is forced to be the form
			$O=\mathbf 1_{V_\lambda}\otimes O_\lambda$.
						
			\item Invariance under $G$.
			Each $\mathcal H_\lambda$ is invariant under all gauge‐covariant operators, that is, $\forall ~ g \in G$, we have:
			$$
			U(g)\,\mathcal H_\lambda \;=\;\mathcal H_\lambda.
			$$
		\end{enumerate}
	\end{proposition}

	\begin{proof}
		\leavevmode
		\begin{enumerate}
			\item By definition of the isotypic decomposition,
			$$
			U(g)\mid_{\mathcal H_\lambda}
			= \rho_\lambda(g)\otimes\mathbf 1_{\mathcal M_\lambda}.
			$$
			If $O$ commutes with all $U(g)$, then on $\mathcal H_\lambda$ we have
			$\bigl[O, \rho_\lambda(g)\otimes \operatorname{Id}\bigr] = 0$.
			Schur’s lemma implies $O$ acts as the identity on $V_\lambda$ tensored with an arbitrary endomorphism on $M_\lambda$, that is, the block diagonal form $O=\mathbf 1_{V_\lambda}\otimes O_\lambda$.
			
			\item Since $U_g$ acts as $\rho_\lambda(g)\otimes\mathbf1$ on $V_\lambda\otimes M_\lambda$, each summand is explicitly invariant.
		\end{enumerate}
	\end{proof}

	\begin{proposition}[Seed of superselection]
		\label{PROP:Superselection}
		No $G$-invariant operator can map between distinct sectors $\mathcal H_\lambda\to\mathcal H_{\lambda'}$ for $\lambda\neq\lambda'$.
		Hence each $\mathcal H_\lambda$ is a superselection sector.
	\end{proposition}

	\begin{proof}
		Let $O$ be a $G$-invariant operator, that is, $O$ commutes with $G$.
		Then according to Proposition \ref{PROP:ActionOnIsotypicSectors} item 1, $O$ must be of the form
		$$
		O=\mathbf 1_{V_\lambda}\otimes O_\lambda,
		$$
		which indicates that the matrix elements of $O$ between $\mathcal H_\lambda$ and $\mathcal H_{\lambda'}$ all vanish when $\lambda\neq\lambda'$.
		Thus, any $G$-invariant operator cannot map between distinct sectors $\mathcal H_\lambda\to\mathcal H_{\lambda'}$ for $\lambda\neq\lambda'$.
	\end{proof}

	\begin{proposition}[Closure under local creation]
		\label{PROP:ClosureCreation}
		Let $a^\dagger_i$ be any local creation operator transforming in irrep $V_\mu$.
		Then $a^\dagger_i$ acting on a state $\ket{\psi} \in \mathcal H_\lambda = V_\lambda \otimes M_\lambda$ create a new state $a_i^\dagger\ket\psi$ lies entirely in the sector
		$$
		\mathcal H_{\lambda,\mu}
		= V_\lambda\otimes V_\mu\otimes\mathcal M_\lambda
		\subset \bigoplus_{\nu\in\lambda\otimes\mu} \mathcal H_\nu,
		$$
		where $\lambda\otimes\mu$ is the countable set for the isotypic decomposition of the tensor product of irreps.
	\end{proposition}

	\begin{proof}
		Creation operators $a^\dagger_i$ are $G$-covariant maps $\mathbb C\to V_\mu$.
		When acting on $\ket\psi \in V_\lambda \otimes M_\lambda$, the total $G$-action is on $V_\lambda\otimes V_\mu$, which decomposes into irreps labeled by $\nu\in\lambda\otimes\mu$.
		Thus 
		$$
		a_i^\dagger\ket\psi \in 
		\bigoplus_{\nu \in \lambda \otimes \mu} V_\nu \otimes M'
		= \bigoplus_{\nu \in \lambda \otimes \mu} \mathcal H_\nu.
		$$
	\end{proof}

	\begin{remark}[Physical interpretation]
		\leavevmode
		\begin{itemize}
			\item  Irreducibility \& Stability.
			Once a net-charge sector $\lambda$ is chosen, IQNs cannot be further split or re-labeled by any gauge-invariant process.
			
			\item  Superselection.
			Different $\lambda$ sectors never interfere.
			No physical Hamiltonian or observable can connect them.
			
			\item  Creation \& Mixing.
			Local operators carrying charge $\mu$ simply shift you from sector $\lambda$ to sectors in $\lambda\otimes\mu$.
						
			\item Grading.
			The factor $\mathcal M_\lambda$ collects all extra quantum numbers (momentum, radial node, family index, $\cdots$).
			These remain manipulable.
			The	confinement applies only to the symmetry-active block $V_\lambda$.
		\end{itemize}
	\end{remark}

	\subsection{Packaging Implies Superselection}
	\label{SEC:PackagingImpliesSuperselection}

	The principle of charge superselection \cite{WWW1952,DHR1971,DHR1974} is a fundamental rule that governs the nature of quantum states.
	It dictates that a quantum system cannot exist in a coherent superposition of states with different electric charges.
	In our case, the isotypic decomposition split the raw-Fock space into irreducible isotypic sectors $\mathcal H_\lambda$.
	Then we can restate superselection as:
	no local gauge-covariant operation can coherently connect different $\mathcal H_\lambda$.
	Let us now rigorously prove that symmetry packaging can directly lead to superselection.

	\begin{proposition}[Packaging $\boldsymbol{\Rightarrow}$ superselection]
		\label{PROP:PackagingImpliesSS}
		Let
		$$
		\mathcal H_{\rm iso}
		=\bigoplus_{\lambda\in\widehat G} \mathcal H_\lambda
		$$
		be the isotypic decomposition and let
		$$
		\mathcal A_{\mathrm{loc}}
		\;:=\;
		\bigl\{\,\mathcal O\in\operatorname{End}(\mathcal H_{\mathrm{Fock}})
		\ \big|\ 
		U_g\mathcal O\,U_g^{-1}=\mathcal O,\;\forall g\in G
		\bigr\}
		$$
		be the local gauge-covariant operator algebra of observable $\mathcal O$ (commutant of the group action).
		Then:
		\begin{enumerate}[label=(\roman*)]
			\item Block-diagonality.			  
			Every $\mathcal O \in \mathcal A_{\mathrm{loc}}$ preserves each $\mathcal H_\lambda$, that is,
			$
			\mathcal O: \mathcal H_\lambda \longrightarrow \mathcal H_\lambda,
			~
			\forall ~ \lambda.
			$	
			In other words, every $\mathcal O$ acts block-diagonally:
			$$
			\mathcal O
			\;=\;
			\bigoplus_{\lambda\in\widehat G}
			\Bigl(
			\mathbf 1_{V_\lambda}\otimes O_\lambda
			\Bigr),
			\quad
			O_\lambda\in\operatorname{End}(M_\lambda).
			$$
			Equivalently, the off-diagonal terms vanishes: $\langle\psi_\lambda|\mathcal O|\phi_{\lambda'}\rangle=0$
			for $\lambda\neq\lambda'$.
			
			\item Irreducibility.		  
			The representation of $\mathcal A_{\mathrm{loc}}$ on each sector
			$\mathcal H_\lambda$ is irreducible:
			trivial closed subspace of $\mathcal H_\lambda$ is left invariant
			by all gauge-covariant local observables.
			
			the only operators that act nontrivially on the irreducible $V_\lambda$ factor are scalars (by Schur’s lemma).
		\end{enumerate}
		Consequently, each label $\lambda$ (or $Q$) defines a superselection sector:
		coherent superpositions across different total charges are forbidden.
	\end{proposition}

	\begin{proof}
		\leavevmode
		\begin{enumerate}[label=(\roman*)]
			\item Block-diagonality.			
			
			Fix $\mathcal O \in \mathcal A_{\mathrm{loc}}$,  
			$\forall ~ g \in G$, we have $U_g\mathcal O = \mathcal O U_g$.
			This means that $\mathcal O$ commutes with the commutant algebra $\mathbb C[G]'$ generated by $\{U_g\,|\, g \in G\}$.
			By double-commutant theory, the joint commutant of the $G$-action
			is precisely 
			$
			\bigoplus_{\lambda} \mathbf 1_{V_\lambda} \otimes \operatorname{End}(M_\lambda).
			$
			Thus, $\mathcal O$ must take the asserted direct-sum form.
			
			Equivalently, denote by $Q_{\rm tot}$ the global charge operator (combination of all local Gauss generators).
			``Locality + gauge-invariance'' implies that $\forall ~ \psi_Q \in \mathcal H_Q$, we have	
			$$
			[\mathcal O, Q_{\rm tot}]
			= 0
			\quad \Longrightarrow \quad
			\mathcal O\,\bigl|\psi_Q\bigr\rangle \in \mathcal H_Q.
			$$
			In other words,	$\langle\psi_Q|\mathcal O|\phi_{Q'}\rangle = \mathcal O \langle\psi_Q|\phi_{Q'}\rangle = 0$ when $Q \neq Q'$.
						
			\item Irreducibility.
			Inside a fixed sector $\mathcal H_\lambda = V_\lambda\otimes M_\lambda$,
			the $G$-action is $U_g = D^{(\lambda)}(g) \otimes \mathbf 1_{M_\lambda}$,  
			where $D^{(\lambda)}$ is irreducible.
			By Schur’s lemma, we have $\operatorname{End}_G(V_\lambda) = \mathbb C\,\mathbf 1_{V_\lambda}$.
			Therefore, the local algebra contains
			$
			\mathbf 1_{V_\lambda} \otimes \operatorname{End}(M_\lambda)
			\cong \operatorname{End}(M_\lambda),
			$
			which is already irreducible on $M_\lambda$.  
			Tensoring with $\mathbf 1_{V_\lambda}$ makes the representation on $\mathcal H_\lambda$ irreducible.
		\end{enumerate}
		These two facts together show that no nontrivial operator can map between different $\mathcal H_\lambda$ and each $\mathcal H_\lambda$ carries an irreducible sector of the observable algebra.
		This is exactly the statement of charge superselection in the DHR sense \cite{DHR1971,DHR1974}.
	\end{proof}

	\begin{example}[Examples of superselection]
		\leavevmode  
		\begin{enumerate}
			\item $U(1)$ charge.  
			$\widehat G=\mathbb Z$.  
			The isotypic decomposition becomes
			$\mathcal H_{\mathrm{iso}} = \bigoplus_{Q \in \mathbb Z} \mathcal H_Q$,
			each $\mathcal H_Q$ spanned by $n$-particle states of total electric charge $Q$.
			Any	gauge-covariant observable $O$ commute with $Q_{\mathrm{tot}}$, that is, satisfies $[O,Q_{\rm tot}]=0$.
			So it cannot connect different $Q$’s and preserves each $\mathcal H_q$.
			This implies superselection rules of $U(1)$ charge.
			
			\item $SU(3)$ color.  
			Blocks are labelled by Young diagrams $\lambda$.
			$SU(3)$ is compact and its group algebra (or universal enveloping algebra) contains a central element, e.g., the quadratic Casimir $C_2$ acts as $c_\lambda\,\mathbf 1$ on $\mathcal H_\lambda$.
			Since $[\,\mathcal O,C_2]=0$ for every color covariant $\mathcal O$,
			off-diagonal matrix elements between distinct $\lambda$’s vanish.
			This implies superselection rules of $SU(3)$ color.
		\end{enumerate}
	\end{example}

	\begin{corollary}[General superselection rule]
		Due to the multi-particle packaging, $G$-preserving process described by $\mathcal A_{\mathrm{loc}}$ cannot create or detect coherent superpositions of different isotypic sectors.
		Consequently, such superpositions are experimentally indistinguishable from classical (incoherent) mixtures.
	\end{corollary}

	This general superselection rule does not limit to charge sectors, it apply to any isotypic sectors includes internal and external symmetry.

	\begin{remark}
		Observables under superselection:
		
		\begin{enumerate}
			\item Locality and gauge invariance force all observables into the commutant of $G$, which is block-diagonal in the charge decomposition.			
			Consequently, once states are packaged into irreducible blocks, no physical operation can mix those blocks coherently.
			
			\item Packaging is a unifying language that shows both electric charge in QED and color singlets in QCD follow the same representation-theoretic fact.
		\end{enumerate}
	\end{remark}

	\subsection{Packaging Survives Isotypic Decomposition}

	We now show that after isotypic decomposition, each packaged block remains indivisible as a single irrep, even in the presence of external DOFs.

	\begin{theorem}[Stability of packaging under isotypic decomposition]
		\label{THM:PackagingSurvivesHybridPWProjection}
		In the packaged subspace $\mathcal H_\lambda^{(n)}$, every state $\lvert\lambda,m;\xi_1\cdots\xi_n\rangle$
		carries its full set of IQNs as one irreducible $G$-block $V_\lambda$ (transforms under a single irrep $R_\lambda$ of $G$).
		No partial factorization of internal labels is possible without breaking irreducibility.
	\end{theorem}

	\begin{proof}
		\leavevmode
		\begin{enumerate}
			\item Internal block is irreducible.
			By construction, $V_\lambda$ is an irrep of $G$.
			Thus, Schur’s lemma excludes nontrivial projectors that commute with all $D^{(\lambda)}$.
			
			\item External factor is a spectator.
			$U(g)$ acts as the identity on $\mathcal H_{\rm ext}^{(n)}$,	so any projection that would attempt to split off a piece of the external label would commute with $D^{(\lambda)}$ {\em trivially}.
			But such a projector acts only on $\mathcal H_{\rm ext}^{(n)}$ and cannot touch the internal block.
			Therefore no operator can split the IQNs without destroying irreducibility.
		\end{enumerate}
	\end{proof}

	Similarly, one can apply above procedure to irreducible decomposition (Clebsch-Gordan rotation).

	\begin{example}[Three-quark singlet]
		For $n=3$ quarks, the totally antisymmetric tensor $\epsilon^{abc}$ gives the unique color singlet		
		$$
		|\mathbf 1;\xi_1\xi_2\xi_3\rangle=
		\frac{1}{\sqrt6}\,\epsilon^{abc}\,|q_a,\xi_1\rangle|q_b,\xi_2\rangle|q_c,\xi_3\rangle.
		$$		
		The vector is already a packaged singlet and is unchanged by the isotypic decomposition.
	\end{example}

	\subsection{Peter-Weyl Projection}	
	\label{SEC:PeterWeylTheoremAndProjectors}

	In our packaging framework, isotypic decomposition is a fundamental property of packaging because the tensor product of single-particle packaging decompose into multi-particle packaging.
	For the convenience of calculation, we need a mathematical tool to implement isotypic decomposition.
	Peter-Weyl projection \cite{PeterWeyl1927,Knapp1986} is one of these tools.
	It maps the raw-Fock space $\mathcal{H}_{\rm Fock}$ into isotypic space $\mathcal{H}_{\rm iso}$:
	splits the reducible Hilbert space into irreducible charge sectors (or reorganizes the reducible blocks into coupled irreducible blocks).

	\subsubsection{Peter-Weyl Theorem: Function Space Statement}

	The Peter-Weyl theorem is a foundational result of representation theory and harmonic analysis on compact groups.
	It says that a Hilbert space carrying a unitary representation of a compact gauge group splits into orthogonal irreducible blocks.
	This is exactly what we need in the isotypic decomposition: decompose a raw-Fock space into isotypic sectors.

	We start with the classical Peter-Weyl theorem for compact groups by emphasizing its three main parts and then specialize to the finite‑group case.

	\paragraph{(1) The Peter-Weyl theorem for compact groups.}
	
	In a function space, the Peter-Weyl theorem can be stated as:
	
	\begin{theorem}[Peter-Weyl theorem]	\label{THM:PWTheoremCompact}
		Let $G$ be an arbitrary compact (second‐countable) topological group endowed with its normalized Haar measure $\mu$ (so $\int_G \mathrm d\mu(g)=1$),
		$U \colon G \to \mathcal U(\mathcal H_{\mathrm{Fock}})$ be a continuous unitary representation of $G$,
		$\widehat G$ be the (countable) set of equivalence classes of its finite-dimensional continuous irreps.
		and let
		$
		\langle f,h\rangle=\int_G f(g)\,h^*(g)\,d\mu(g).
		$,
		$V_\lambda$ be its carrier space with dimension $d_\lambda = \dim V_\lambda$,
		$D^{(\lambda)}: G \to U(V_\lambda)$,
		and $D^{(\lambda)}_{mn}(g) = \langle e_m, \rho_\lambda(g)e_n\rangle$ be the matrix entries in any fixed orthonormal basis $\{e_n\}_{n=1}^{d_\lambda}$.
		and $L^{2}(G)$ be the space of square-integrable functions on $G$.
		\begin{enumerate}
			\item Density of matrix coefficients.
			The linear span of all matrix entries
			$\{D^{(\lambda)}_{mn}(\cdot)\:|\:\lambda\in\widehat G,\;1\le m,n\le d_\lambda\}$
			is dense in the space of continuous complex functions $C(G)$ on $G$ (with the uniform topology) and hence in $L^2(G)$ (with the $L^2$-norm).
			
			\item Orthogonality relations.
			For any $\lambda,\lambda'\in\widehat G$ and $m,n$ in the first and $p,q$ in the second,
			\begin{equation}\label{EQ:SchurOrthogonality}
				\int_G d \mu(g) D^{(\lambda)}_{mn}(g) D^{(\lambda')}_{pq}(g)^*
				= \frac{1}{d_\lambda} \delta_{\lambda,\lambda'} \delta_{mp} \delta_{nq}.
			\end{equation}	
			
			the continuous version of the Schur orthogonality relations:
			$$
			\int_G\chi_\lambda(g)\,\chi_{\lambda'}(g)^*\,\mathrm d\mu(g)
			\;=\;
			\frac{1}{d_\lambda}\,\delta_{\lambda\lambda'}.
			$$
			
			\item Fourier expansion.
			The set
			$
			\bigl\{\sqrt{d_\lambda}\,D^{(\lambda)}_{mn}(\,\cdot\,)
			:\lambda\in\widehat G,\ m,n=1,\dots,d_\lambda\bigr\}
			$
			forms an orthonormal basis of $L^2(G)$.
			In $L^2(G)$, one has the complete orthogonal decomposition
			$$
			L^2(G)
			\;\cong\;
			\bigoplus_{\lambda\in\widehat G}
			V_\lambda^*\;\otimes\;V_\lambda,
			$$
			realized by the matrix‐elements of each irrep.		
			Equivalently, every	$f \in L^2(G)$ admits the Fourier expansion
			\begin{equation}
				f(g)
				\;=\;
				\sum_{\lambda\in\widehat G}
				d_\lambda\,
				\operatorname{Tr}\!\bigl[\widetilde f(\lambda)\,D^{(\lambda)}(g)\bigr],
				\qquad
				\widetilde f(\lambda)
				:=\int_G f(g)\,D^{(\lambda)}(g)^\dagger\,d\mu(g).
			\end{equation}

			where $\widehat G$ is the set of finite-dimensional irreps.
			The Haar measure is normalised to unity $\int_G d\mu(g)=1$.
		\end{enumerate}
	\end{theorem}

	\begin{proof}
		We only sketch the main ideas here \cite{Folland2016,Serre1977}.

		\begin{enumerate}
			\item Density:
			
			The Stone-Weierstrass theorem applies to the algebra generated by the matrix coefficients because they separate points on $G$ and are closed under complex conjugation.
			Therefore, this algebra is dense in $C(G)$ with the supremum norm.
			
			Continuous functions are dense in $L^2(G)$, so the matrix coefficients are also dense there.
			
			\item Orthogonality (Schur orthogonality):
			
			Consider the regular representation of $G$ on $L^2(G)$ by left translation $L_g f(x)=f(g^{-1}x)$.
			
			The matrix coefficients furnish intertwiners between this regular representation and the irreducible $\rho_\lambda$.
			
			Integrating the product of two matrix entries over $G$ and using the fact that the only intertwiners between inequivalent irreps are zero (Schur's lemma) fixes the integral to $\tfrac1{d_\lambda}\delta_{\lambda,\lambda'}\delta_{mp}\delta_{nq}$.
			
			\item Fourier expansion:
			
			The orthonormality follows by normalizing the matrix coefficients by $\sqrt{d_\lambda}$.
			
			The completeness follows from density plus orthogonality in a separable Hilbert space:
			an orthonormal set that is, dense is automatically a basis.
		\end{enumerate}
	\end{proof}

	\paragraph{(2) Finite‑group specialization.}

	When $G$ is a finite group of order $|G|<\infty$, we replace the Haar integral with the normalized sum 
	$$
	\int_Gd \mu \longrightarrow \tfrac1{|G|} \sum_{g\in G}
	$$
	and replace the orthogonality relations by	
	$$
	\sum_{g\in G} D^{(\lambda)}_{mn}(g)\,D^{(\lambda')}_{pq}(g)^* 
	\longrightarrow
	\frac{|G|}{d_\lambda}\,\delta_{\lambda,\lambda'}\,\delta_{mp}\,\delta_{nq}\,.
	$$	
	All statements of Theorem \ref{THM:PWTheoremCompact} hold with integrals replaced by this sum.

	\subsubsection{Peter-Weyl Projectors: Operator Form}

	Let $U: G \to U(\mathcal H)$ be any (strongly continuous) unitary representation on a (possibly infinite‑dimensional) Hilbert space $\mathcal H$.
	For each $\lambda\in\widehat G$, we define Peter-Weyl (spectral) projectors (the operator‐valued version for a representation $U(g)$ on $\mathcal H_{\rm Fock}$):
	$\chi_\lambda(g) = \operatorname{Tr}D^{(\lambda)}$ be the character.
	\begin{itemize}
		\item Compact case:
		\begin{equation}\label{EQ:PeterWeylProjectorCompact}
			P_\lambda
			= d_\lambda\int_G d\mu(g)\;\chi_\lambda(g)^*\;U(g)
		\end{equation}
		
		\item Finite case:
		\begin{equation}\label{EQ:PeterWeylProjectorFinite}
			P_\lambda = \frac{d_\lambda}{|G|}\sum_{g\in G}\chi_\lambda(g)^*\,U(g)
		\end{equation}
	\end{itemize}
	Thus, the raw‑Fock space $\mathcal H_{\rm Fock}$ decomposes into the isotypic space as a direct sum of superselection sectors	
	$$
	\mathcal H_{\rm iso}
	= \bigoplus_{\lambda\in\widehat G} P_\lambda\mathcal H_{\rm Fock}
	= \bigoplus_{\lambda\in\widehat G}\,\mathcal H_\lambda,
	\quad
	\mathcal H_\lambda:=P_\lambda\mathcal H_{\rm Fock}
	$$	
	with each $\mathcal H_\lambda$ transforming under the irrep $\lambda$.

	\begin{remark}[Peter-Weyl/gauge projector convergence]
		Standard model gauge groups are compact and the above definition of Peter-Weyl projector works well.
		But for non-compact gauge groups (e.g. SL(2, $\mathbb{R}$) in gravity), the Haar measure is non-finite so the projector needs care.
		One may consider using Sobolev-norm trick for non-compact cases.     
	\end{remark}

	We can treat the Peter-Weyl (spectral) projectors Eq.  \eqref{EQ:PeterWeylProjectorCompact} and \eqref{EQ:PeterWeylProjectorFinite} in two equivalent ways, each has it own particular physical meaning:
	
	\begin{enumerate}
		\item \textbf{Projective (superselection) viewpoint.}
		The integral averages the representation $U(g)$ over the group with
		the weight $\chi_\lambda^*(g)$, thereby annihilating all irreps
		except~$\lambda$.
		Hence $P_\lambda\mathcal H_{\mathrm{Fock}}$
		is the maximal subspace that transforms as~$\lambda$.
		Generally, one thinks that the integral sums over the continuous group elements weighted by the complex conjugate of the character.
		The resulting operator projects onto the subspace of $\mathcal H_\lambda$ transforming under $\lambda$.
		
		Each $P_\lambda$ projects onto the charge sector $\mathcal{H}_\lambda$ that carries irreducible gauge charge $\lambda$.
		Thus, the family $\{P_\lambda\}_{\lambda\in\widehat G}$ are the superselection‐rule projectors.
		Once a system is in the subspace $\mathcal{H}_\lambda = P_\lambda\mathcal H_{\rm Fock}$, no gauge‐invariant operator can connect it to a different sector.
		
		\item \textbf{Fourier-transform viewpoint.}
		The Peter-Weyl (spectral) projectors Eq.  \eqref{EQ:PeterWeylProjectorCompact} and \eqref{EQ:PeterWeylProjectorFinite} are indeed mappings
		$$
		U(g) \longmapsto P_\lambda.
		$$
		If we interpret $g \mapsto U(g)$ as a unitary‐valued function on~$G$, then the map $U(g) \longmapsto P_\lambda$ is its $\lambda$-th Fourier coefficient.
		In that sense, \eqref{EQ:PeterWeylProjectorCompact} furnishes a non-commutative	Fourier transform from functions on~$G$ to central operators on $\mathcal H_{\mathrm{Fock}}$.
		One can also think that $U(g)$ appears unweighted in Eq.~\eqref{EQ:PeterWeylProjectorCompact} except for the character factor,
		the combination $d_\lambda\,\chi^{*}_\lambda(g)\,U(g)$ singles out exactly
		the $\lambda$-block in $\mathcal H_{\rm Fock}$.
		Therefore, one may regard Eq.~\eqref{EQ:PeterWeylProjectorCompact} as a non-Abelian Fourier transform of the operator $U(g)$ along the compact group (along $G$).
		It extracts the $\lambda$-component of the group‐representation content of any operator or state.
	\end{enumerate}

	\begin{proposition}[Algebraic properties of $P_\lambda$]
		\label{prop:PW-operator-properties}
		The Peter-Weyl projector family $\{P_\lambda:\lambda\in\widehat G\}$ satisfies:
		\begin{enumerate}			
			\item Self‑adjointness:
			$P_\lambda^\dagger=P_\lambda$.
			
			\item Idempotence:
			$P_\lambda^2=P_\lambda$.
			
			\item Orthogonality:
			$P_\lambda P_{\lambda'}=\delta_{\lambda,\lambda'}\,P_\lambda.$
			
			\item Completeness:
			$\sum_{\lambda \in \widehat G} P_\lambda = \mathbf 1_\mathcal{H}$ (strong operator convergence).
			
			\item Gauge invariance:
			$[P_\lambda,\,U(h)]=0$ for all $h\in G$.
		\end{enumerate}
	\end{proposition}

	\begin{proof}
		\leavevmode
		\begin{enumerate}
			\item Self‑adjointness \& idempotence:
			Follow by replacing $\chi_\lambda^*(g)$ with its conjugate under $g\to g^{-1}$ and using orthogonality relations.
			In particular,			
			$$
			P_\lambda^2 = d_\lambda^2\int\!\!\int d\mu(g)d\mu(h)\;\chi_\lambda(g)^*\chi_\lambda(h)^*\,U(gh)
			= d_\lambda\!\int d\mu(k)\;\chi_\lambda(k)^*\,U(k) = P_\lambda,
			$$	
			where we have used the convolution identity and $\int\chi_\lambda(g)^*\chi_\lambda(gx)\,d\mu(g)=\tfrac1{d_\lambda}\chi_\lambda(x)$.
			
			\item Orthogonality:
			Directly from the Schur (mixed‑character) orthogonality relations Eq.~\eqref{EQ:SchurOrthogonality}.
			
			\item Completeness:
			Since the matrix coefficients form a complete basis of $L^2(G)$, the resolution of the identity in the regular representation implies $\sum P_\lambda = 1$.
			
			\item Gauge invariance:
			Since characters are class functions, $\chi_\lambda(hgh^{-1})=\chi_\lambda(g)$, one has
			$U(h)P_\lambda U(h)^{-1} = d_\lambda\!\int d\mu(g)\chi_\lambda(g)^*\,U(hgh^{-1}) = P_\lambda.$
		\end{enumerate}
	\end{proof}

	\begin{example}
		Examples of Peter-Weyl projectors:
		
		\begin{enumerate}
			\item Compact group $G = U(1)$.
			
			Let $G=U(1)={e^{i\theta}}$ with Haar measure $d\theta/(2\pi)$.
			\begin{itemize}
				\item Irreps.
				All irreducible unitary representations of $U(1)$ are one-dimensional, labeled by an integer $n\in\mathbb Z$:			
				$$
				\rho_n\bigl(e^{i\theta}\bigr) = e^{i n\theta},
				\quad d_n = 1,
				\quad \chi_n\bigl(e^{i\theta}\bigr) = e^{i n\theta}.
				$$
				
				\item Peter-Weyl projectors.
				On any Hilbert space $\mathcal H$ carrying a unitary representation $U(e^{i\theta})$ of $U(1)$, the $n$-sector projector is			
				$$
				P_n 
				= d_n \int_{0}^{2\pi}\!\frac{d\theta}{2\pi}\;\chi_n^*(e^{i\theta})\;U(e^{i\theta})
				= \int_{0}^{2\pi}\!\frac{d\theta}{2\pi}\;e^{-i n\theta}\,U(e^{i\theta}).
				$$			
				By orthogonality of characters, we have
				$P_n P_{n'}  = \delta_{n,n'}P_n$,
				$\sum_n P_n = \mathbf 1$,
				and $[P_n, U(e^{i\theta})] = 0$.
				
				\item Action on charge eigenstates.
				If $\ket{q}$ satisfies $U(e^{i\theta})\ket{q}=e^{i q\theta}\ket{q}$, then			
				$$
				P_n\,\ket{q}
				= \int_{0}^{2\pi}\!\frac{d\theta}{2\pi}\,e^{-i n\theta}e^{i q\theta}\ket{q}
				= \delta_{n,q}\,\ket{q}.
				$$			
				Thus, $P_n$ picks out precisely the subspace of total U(1) charge $n$.
			\end{itemize}

			\item Finite group $G = \mathbb{Z}_2$.
			
			Let $\mathbb{Z}_2 = \{e,a\}$ with $a^2 = e$.
			\begin{itemize}
				\item Irreps.				
				$\mathbb{Z}_2$ has two one-dimensional irreps:
				
				Trivial $\rho_0$ ($\lambda=0$):
				with characters $\chi_0(e)=\chi_0(a)=1$.
				
				Sign $\rho_1$ ($\lambda=1$):
				with characters $\chi_1(e)=1,\chi_1(a)=-1$.
				
				\item Peter-Weyl projectors.
				For any unitary rep. $U\colon\mathbb Z_2\to U(\mathcal H)$,			
				$$
				P_0 
				= \frac{d_0}{|\mathbb Z_2|}\sum_{g\in\mathbb Z_2}\chi_0^*(g)\,U(g)
				= \frac{1}{2}\bigl[U(e)+U(a)\bigr],
				$$
				
				$$
				P_1 
				= \frac{d_1}{2}\sum_{g\in\mathbb Z_2}\chi_1^*(g)\,U(g)
				= \frac{1}{2}\bigl[U(e)-U(a)\bigr].
				$$
				
				One can easily verify that:				
				$P_0+P_1=\mathbf 1$, $P_i^2=P_i$, $P_0P_1=0$, and $[P_i,U(g)]=0$ for all $g \in G$
				
				\item Diagonal basis action.
				Suppose $U(1)\ket{\psi_\pm}=\pm\ket{\psi_\pm}$.
				Then we have
				$$
				P_0\ket{\psi_\pm}=\tfrac12(1\pm1)\ket{\psi_\pm}
				= \begin{cases}\ket{\psi_+},&\text{for “+”}\\0,&\text{for “-”}\end{cases},
				$$
				
				$$
				P_1\ket{\psi_\pm}=\tfrac12(1\mp1)\ket{\psi_\pm}
				= \begin{cases}0,&\text{for “+”}\\\ket{\psi_-},&\text{for “-”}\end{cases}.
				$$			
				Thus, $P_0$ projects onto the even (+) subspace, $P_1$ onto the odd (-) subspace.
			\end{itemize}			
		\end{enumerate}
	\end{example}

	\subsection{Irreducible Decomposition and Clebsch-Gordan Rotation}
	\label{SEC:ClebschGordanRotation}

	In previous subsections, we investigated isotypic decomposition and Peter-Weyl projection, which are general constructions.
	Here we go one step further to irreducible decomposition and Clebsch-Gordan rotation, which are refinements by choosing an explicit unitary basis inside each isotypic block selected by the Peter-Weyl projector.
	Finally, we display each isotypic block in block-diagonal form.

	\subsubsection{Irreducible Decomposition}

	\paragraph{(1) Full irreducible decomposition.}
	The isotypic decomposition split the representation of the gauge group into the direct sum of blocks $V_\lambda \otimes M_\lambda$, one for each irrep type $\lambda$.
	The same for peter-Weyl projection, which just project onto each charge‐sector (each irrep) as a whole without immediately breaking each sector into its individual multiplicity copies.
	
	However, if we want to go one step further and actually chose a basis in each multiplicity space so that every single copy of every irrep was split off as its own summand, then we need the full irreducible decomposition.
	For example, if we want to diagonalize a Hamiltonian that acts nontrivially inside $M_\lambda$,
	label states by a second commuting symmetry that distinguishes copies,
	or define entanglement between different copies inside the same label,
	then we need refine the isotypic decomposition to a full irreducible decomposition.

	\paragraph{(2) Block diagonal form.}
	Furthermore, the Peter-Weyl projector $P_\lambda$ is defined with respect to the original basis of the raw-Fock space $\mathcal H_{\rm Fock}$.  
	It singles out the subspace $\mathcal H_\lambda \subset \mathcal H_{\rm Fock}$, that is,
	\begin{equation}\label{EQ:PeterWeylProjectorIrrepPortion}
		P_\lambda:\;\mathcal H_{\rm Fock}\longrightarrow
		\mathcal H_\lambda
	\end{equation}
	that transforms in the irrep $\lambda$.
	But in the Fock basis, the matrix representation of $P_\lambda$ may not be block-diagonal.
	It may be a dense operator that happens to be idempotent and self-adjoint.
	If one wishes to display this irrep block in block-diagonal form, 
	then he/she needs a unitary change of basis (a Clebsch-Gordan rotation) to makes the irrep block in an explicit block-diagonal form.

	In this subsection, we introduce Clebsch-Gordan rotation $\mathcal U_{\rm CG}$ \cite{Racah1942,Wigner1959}.
	It implements irreducible decomposition and simultaneously rotate all Peter-Weyl projectors into the form of diagonal blocks
	$$
	\mathcal U_{\rm CG}\,P_\lambda\,\mathcal U_{\rm CG}^{\dagger}
	= I_{V_\lambda}\otimes 0_{M_\lambda}.
	$$
	In fact, the Clebsch-Gordan rotation $\mathcal U_{\rm CG}$ can indeed decompose the entire raw-Fock space $\mathcal H_{\rm Fock}$ into block-diagonal charge sectors.
	This is a useful scaffolding for multiplicity counting and for displaying operators in block form, but it is not required for enforcing the superselection rule itself.

	\paragraph{(3) ``Not diagonal'' $\ne$ ``Not block-diagonal''}

	In the original Fock basis, a projector $P_\lambda$ usually looks like a dense matrix.
	In this sense, we say $P_\lambda$ is ``not diagonal''.
	What matters is that the set $\{P_\lambda\}$ is mutually orthogonal and complete.
	Together they impose a block-diagonal structure on every physical operator, even though an individual block need not be displayed in diagonal form.
	
	If you want the visual block structure you can perform an allowable Clebsch-Gordan rotation (or any other unitary that respects the decomposition) to bring each $P_\lambda$ to the canonical form	
	$$
	P_\lambda \;\longmapsto\;
	\mathbf 1_{V_\lambda}\otimes 0_{\text{others}},
	$$	
	but this is cosmetic:
	the superselection rule and the vanishing cross-sector amplitudes were already true before the rotation.

	\subsubsection{Clebsch-Gordan Rotation}
	\label{SEC:ClebschGordanMapHybridPackagedSubspace}
	
	We now extend the CG rotation to packaged subspace where both internal and external DOFs are involved.

	\paragraph{(1) Hybrid-packaged subspace before rotation.}

	In the presence of external DOFs (spectators), we have the packaged subspace
	$$
	\mathcal H_{\rm Fock}^{(n)} 
	= 
	V_{R_1}\otimes\cdots\otimes V_{R_n} \otimes \mathcal H_{\rm ext}^{(n)}
	$$
	on which the gauge group acts via
	$U(g)^{\otimes n}=D^{(R_1)}(g)\otimes \cdots\otimes D^{(R_n)}(g)$, which is the same as that in internal-packaged subspace.

	\paragraph{(2) Definition of Clebsch-Gordan (CG) rotation.}

	In the packaged subspace, we consider a full rotation of the form
	$$
	\mathcal U_{\rm CG} :
	\mathcal H_{\text{hyb}}^{(n)} \to \mathcal H_{\text{CG}}^{(n)}
	$$
	which leaves external labels unchanged while coupling internal indices.
	Thus, for any basis state $|\alpha_1\cdots\alpha_n\rangle\otimes|\xi\rangle$, we have
	$$
	\mathcal U_{\rm CG}\bigl(|\alpha_1\cdots\alpha_n\rangle\otimes|\xi\rangle\bigr)
	=\sum_{\lambda,m}C_{\lambda m}^{\,\alpha_1\cdots\alpha_n}\,|\lambda,m\rangle\otimes|\xi\rangle.
	$$	
	Formally, we give the definition as follows:
	
	\begin{definition}[Clebsch-Gordan (CG) rotation in packaged subspace]
		Let $G$ be a finite or compact gauge group.
		Let $R_1, \ldots, R_n$ be an ordered set of irreps of $G$ for each single-particle, respectively.
		Let $\alpha_k$ be internal indices in $R_k$.
		Define a unitary map $\mathcal U_{\rm CG}$ behaves as
		$$
		\mathcal U_{\rm CG}:
		V_{R_1}\otimes\cdots\otimes V_{R_n} \otimes \mathcal H_{\rm ext}^{(n)}
		\longmapsto
		\bigoplus_{\lambda} \left( V_\lambda \otimes \mathcal H_{\rm ext}^{(n)} \right)
		=
		\bigoplus_{\lambda} \mathcal \mathcal H_\lambda^{(n)},
		$$
		and send the reducible product state $\ket{\alpha_1 \cdots \alpha_n}$ to coupled irreducible space
		$$
		\mathcal U_{\rm CG}:
		\ket{\alpha_1\cdots\alpha_n}\!\otimes
		\ket{\xi}
		\;\longmapsto\;
		\sum_{\lambda,m}
		C_{\lambda m}^{\alpha_1\!\cdots\!\alpha_n}\;
		\ket{\lambda,m}\!\otimes\!\ket{\xi},
		$$
		where
		$\lambda$ labels the resulting irreducible $G$-block $R_\lambda$,  
		$m=1,\dots,N_\lambda$ counts multiplicity,
		and $\ket{\xi}$ collects all external labels spin/momentum, $\cdots$.
		The coefficients $C_{\lambda m}^{\alpha_1 \cdots \alpha_n}$ act as
		\begin{equation}
			C_{\lambda m}^{\alpha_1 \cdots \alpha_n}
			:  R_1\otimes\!\cdots\!\otimes R_n
			\!\longrightarrow\! R_\lambda,
		\end{equation}
		and are orthonormal only up to phase, that is,
		$$
		\sum_{{\alpha}} C_{\lambda m}^{{\alpha}} C_{\lambda' m'}^{\*{\alpha}} =\delta_{\lambda\lambda'}\delta_{mm'}
		$$
		This ensures that $\mathcal U_{\rm CG}$ is unitary.
		Then we say that 
		the map $\mathcal U_{\rm CG}$ is a \textbf{hybrid Clebsch-Gordan rotation},
		each coefficient $C_{\lambda m}^{\alpha_1 \cdots \alpha_n}$ is a 	\textbf{Clebsch-Gordan coefficients},
		and each $\mathcal \mathcal H_\lambda^{(n)}$ is a \textbf{packaged charge sector}.
	\end{definition}

	\paragraph{(3) Packaged subspace after rotation.}
	
	After applying the hybrid CG rotation $\mathcal U_{\rm CG}$, the packaged subspace decomposes into a direct sum of fixed-charge sectors
	$$
	\mathcal H_{\rm CG}^{(n)}
	:=
	\bigoplus_{\lambda}
	\left( V_\lambda \otimes \mathcal H_{\rm ext}^{(n)} \right)
	:=
	\mathrm{span}\Bigl\{
	\bigl|\lambda,m;\xi_1\!\cdots\xi_n\bigr\rangle
	\Bigr\}.
	$$
	where $\mathcal H_{\rm CG}^{(n)}$ is the $n$-particle packaged Clebsch-Gordan subspace (or diagonal $n$-particle packaged Peter-Weyl subspace).
	The total packaged Clebsch-Gordan space is the sum of all $\mathcal H_{\rm CG}^{(n)}$,
	$$
	\mathcal H_{\rm CG}
	:=
	\bigoplus_{n \ge 0} \mathcal H_{\rm CG}^{(n)}
	$$
	where $\mathcal H_{\rm ext}^{(0)} := \mathbb C \ket0$.
	Since $U(g)$ acts block-diagonally (Lemma~\ref{LEM:BlockGauge}), each $\lambda$ labels an
	independent packaged sector available for the superpositions of	Stage 5.

	Inside every fixed pair $(\lambda,Q)$, we may now form coherent superpositions of the
	packaged basis vectors.  
	The resulting packaged entangled states inherit irreducibility from
	Theorem~\ref{THM:PackagingSurvivesHybridPWProjection} and will be analyzed in the	next section.

	\begin{example}[Examples of CG rotation]
		\leavevmode
		\begin{enumerate}
			\item Colour $\mathbf 3\otimes\overline{\mathbf 3}$.			
			For a quark-antiquark pair, one has
			$
			\mathbf 3\otimes\overline{\mathbf 3}
			= \mathbf 1 \oplus \mathbf 8.
			$
			The CG rotation produces the singlet
			$$
			\ket{\mathbf 1}
			=\frac{1}{\sqrt3}\delta^{i}_{\,j}\ket{q_i\bar q^{\,j}}
			$$
			and the octet			
			$$
			\ket{\mathbf 8_a}
			=\bigl(\lambda_a\bigr)^{i}_{\,j}\ket{q_i\bar q^{\,j}},
			$$
			where $\lambda_a$ are the Gell-Mann matrices.
			Both states are packaged blocks of definite color charge.
			
			\item Electron-positron pair.			
			For abelian charges (e.g. $U(1)$ in QED), all single-particle irreps are one-dimensional. The Clebsch-Gordan rotation is trivial.
			The packaged basis is therefore the usual product basis
			$|e^-_{\uparrow}\,e^+_{\downarrow}\rangle,$
			etc.
		\end{enumerate}		
	\end{example}

	\begin{example}[Multi-particle packaged basis vector]
		Take $n=3$ quarks in QCD.
		The internal color Hilbert space is
		$\mathbb C^3 \otimes \mathbb C^3 \otimes \mathbb C^3$.
		The totally antisymmetric tensor $\epsilon^{abc}$ furnishes the unique SU(3) singlet:
		$$
		\ket{\mathbf 1;\,{\xi_1\xi_2\xi_3}}
		\;:=\;
		\frac{1}{\sqrt{6}}\,
		\epsilon^{abc}\,
		\ket{q_a,\xi_1}\otimes\ket{q_b,\xi_2}\otimes\ket{q_c,\xi_3},
		$$
		where each $\xi_i$ may be any external label, e.g., individual
		spins $s_i=\uparrow,\downarrow$ or momenta $\mathbf p_i$.
		Since the $\epsilon$-tensor contracts all color indices, the internal part is an SU(3) singlet.
		Consequently, the entire vector is packaged.		
		Any linear combination of such basis states with the same external
		labels $\xi_1\xi_2\xi_3$ remains in the packaged singlet sector.
	\end{example}

	\subsubsection{Properties of Clebsch-Gordan Rotation}

	\begin{lemma}[Unitarity of $\mathcal U_{\rm CG}$]
		\label{LEM:UCGUnitary}
		The CG rotation $\mathcal U_{\rm CG}$ is unitary on $\mathcal H_{\rm CG}^{(n)}$.
	\end{lemma}

	\begin{proof}
		Since both $\mathcal U_{\rm CG}^{\rm int}$ and $\mathbf 1_{\rm ext}$ are unitary,		
		$(\mathcal U_{\rm CG}^{\rm int})^\dagger ~ \mathcal U_{\rm CG}^{\rm int}
		= \mathbf 1_{V^{\otimes n}}$ and $\mathbf 1_{\rm ext}$ is trivially unitary.
		Thus, their tensor product $\mathcal U_{\rm CG}$ satisfies
		$
		\mathcal U_{\rm CG}^\dagger ~ \mathcal U_{\rm CG}
		=
		(\mathcal U_{\rm CG}^{\rm int} {}^\dagger \otimes \mathbf 1_{\rm ext})
		(\mathcal U_{\rm CG}^{\rm int} \otimes \mathbf 1_{\rm ext})
		=\mathbf 1_{\rm hyb}.
		$
	\end{proof}

	\begin{lemma}[Orthonormality and completeness]\label{LEM:BlockGauge}
		For every $\forall ~ g \in G$, we have
		The packaged vector set
		$$
		\bigl\{\,
		\lvert\lambda,m;\xi_1\cdots\xi_n\rangle
		\bigr\}_{\lambda,m,\{\xi\}}
		$$
		that preserves orthonormality and completeness.		
	\end{lemma}

	\begin{proof}
		Unitarity of $\mathcal U_{\rm CG}$
		(Lemma~\ref{LEM:UCGUnitary}) transports the product basis to the hybrid basis
		$$
		\bigl\{\,
		\lvert\lambda,m;\xi_1\cdots\xi_n\rangle
		\bigr\}_{\lambda,m,\{\xi\}},
		$$
		which forms an orthonormal basis of $\mathcal H_{\rm CG}^{(n)}$.
		Thus, it preserves orthonormality and completeness.
	\end{proof}

	\begin{proposition}
		\label{PROP:PackagedCovariance}
		Let
		$$
		\mathcal H_{\rm CG}
		= \bigoplus_{n=0}^{\infty} \bigoplus_\lambda 
		V_\lambda\!\otimes\!\mathcal H_{\rm ext}^{(n)}
		= \bigoplus_{n=0}^{\infty} \bigoplus_\lambda \mathcal H_\lambda^{(n)}
		$$
		be the Clebsch-Gordan space.
		Then
		\begin{enumerate}
			\item Sector stability (invariance).
			For every $g\in G$, $U(g)\mathcal H_\lambda^{(n)}=\mathcal H_\lambda^{(n)}$.
			
			
			\item The unitary representation of $G$ restricted to
			$\mathcal H_{\rm CG}$ decomposes as the direct sum of the
			irreducibles $D^{(\lambda)}$:
			$$
			U(g)\Bigl|_{\mathcal H_{\rm CG}}
			\;=\;
			\bigoplus_{n=0}^{\infty}\,
			\bigoplus_\lambda
			\bigl(D^{(\lambda)}(g)\otimes\mathbf 1_{\rm ext}\bigr),
			\quad g\in G.
			$$
		\end{enumerate}
	\end{proposition}

	\begin{proof}
		We prove it in three steps:
		
		\begin{enumerate}
			\item For each subspace $\mathcal H_\lambda^{(n)}
			:=V_\lambda\otimes\mathcal H_{\rm ext}^{(n)}$, by construction the gauge action on this subspace is
			$U(g)=D^{(\lambda)}(g)\otimes\mathbf 1_{\rm ext}$.
			Since $\mathbf 1_{\rm ext}$ is the identity on
			$\mathcal H_{\rm ext}^{(n)}$, it follows that
			$U(g)\mathcal \mathcal H_\lambda^{(n)}=\mathcal \mathcal H_\lambda^{(n)}$
			for all $g\in G$.
			Hence every $\mathcal \mathcal H_\lambda^{(n)}$ is
			invariant.			
			

			\item Direct-sum structure.
			Since the Clebsch-Gordan map is unitary,
			the full $n$-particle packaged subspace splits orthogonally:
			$
			\mathcal H_{\rm CG}^{(n)}
			=\bigoplus_\lambda\mathcal \mathcal H_\lambda^{(n)} .
			$
			Each $U(g)$ acts block-diagonally with blocks
			$D^{(\lambda)}(g)\otimes\mathbf 1_{\rm ext}$.
			Taking the orthogonal sum over all $n$ gives
			the stated decomposition on $\mathcal H_{\rm CG}$.
			
			\item Covariance of arbitrary vectors.
			Let
			$
			\ket{\Psi}
			\in\mathcal H_{\rm CG}
			$
			be arbitrary.
			Decompose
			$
			\ket{\Psi}
			=\sum_{n,\lambda}\ket{\Psi_{n,\lambda}},
			$
			$
			\ket{\Psi_{n,\lambda}}
			\in\mathcal \mathcal H_\lambda^{(n)}.
			$
			Then
			$
			U(g)\ket{\Psi}
			=\sum_{n,\lambda}\bigl(D^{(\lambda)}(g)\otimes\mathbf 1\bigr)
			\ket{\Psi_{n,\lambda}}
			\in\mathcal H_{\rm CG},
			$
			so $\mathcal H_{\rm CG}$ is closed under the action of $G$.
			The transformation law is manifestly covariant: each component
			transforms in its own irreducible block $D^{(\lambda)}$.
		\end{enumerate}
	\end{proof}

	\begin{proposition}[Irreducibility of the internal factor]
		\label{PROP:IntIrrep}
		Each sector $\mathcal H_\lambda^{(n)}
		:=V_\lambda\otimes\mathcal H_{\rm ext}^{(n)}$ contains no proper nontrivial subspace
		invariant under the full gauge action.
	\end{proposition}

	\begin{proof}
		By construction, $V_\lambda$ is an irrep of $G$ in the Clebsch-Gordan decomposition.  
		If $\mathcal W\subset\mathcal H_\lambda^{(n)}$ were a proper invariant subspace, then its partial trace over $\mathcal H_{\rm ext}^{(n)}$ would yield a nontrivial invariant subspace of $V_\lambda$.
		This is contradicting to irreducibility.
	\end{proof}

	Consequently, the internal indices of a whole $n$-body excitation have again fused into a single irreducible block.

	\begin{example}[Examples of CG rotation]
		\leavevmode
		\begin{enumerate}
			\item Colour $\mathbf3\!\otimes\!\overline{\mathbf3}$ with quark spins.
			
			Take a quark-antiquark pair.  Internally
			$\mathbf3\otimes\overline{\mathbf3}=\mathbf 1\oplus\mathbf8$,
			while externally each quark carries a two-component spinor.
			The hybrid basis is
			$
			\bigl|\mathbf 1,m;\!s_q,s_{\bar q}\bigr\rangle,
			\;
			\bigl|\mathbf8,m;\!s_q,s_{\bar q}\bigr\rangle
			$
			with $m=1,\dots,8$ for the octet.  
			Only the $\mathbf 1$ survives Stage 6.
			
			\item Electron-positron pair with momenta.
			
			For QED each single-particle irrep is one-dimensional, so the CG rotation
			is trivial.  The packaged basis is simply the product of electron and
			positron momentum eigenstates (plus spin), illustrating how spectators
			are unaffected.
			
			\item Hydrogen atom: hyperfine level $\otimes$ COM momentum.
			
			With $G=U(1)$ the proton-electron bound state is already a
			charge-neutral packaged block.  
			Internal space $V_{F}$ is spanned by the hyperfine states
			$\lvert F,m_F\rangle$,
			external space by the centre-of-mass momentum
			$\lvert\mathbf P\rangle$.
			The hybrid basis
			$\lvert F,m_F;\mathbf P\rangle$
			is packaged and transforms trivially under $U(1)$.
		\end{enumerate}
	\end{example}

	\begin{remark}[Packaging Survives Irreducible (Clebsch-Gordan) Decomposition]
		Once you have split
		$\mathcal H = \bigoplus_Q M_Q \otimes V_Q$
		into its irreducible summands (either by choosing an isotypic basis in each $M_Q$ or by performing a Clebsch-Gordan rotation inside each $V_Q$), the same Casimir‐ or Schur‐lemma argument shows that
		
		\begin{enumerate}
			\item No gauge‐covariant operator ever mixes different $Q$ or different copies of the same $V_Q$.
			
			\item Each individual copy of $V_Q$ remains irreducible under the local algebra.
		\end{enumerate}		
		In other words, packaging is invariant under any further irreducible decomposition or basis change inside each irrep block.
	\end{remark}

	\section{Stage 5: Packaged Superposition (Packaged Entanglement)}
	\label{SEC:Stage5PackagedSuperposition}

	From Stage 4 we know that, after the Peter-Weyl projection, the raw-Fock space decomposes into a direct sum of packaged sectors $\mathcal H_\lambda^{(n)}$, that is,
	$$
	\mathcal H_{\rm iso}^{(n)}
	= \bigoplus_{\lambda} \mathcal H_\lambda^{(n)}.
	$$
	Each sector $\mathcal H_\lambda^{(n)}$ carrying an irrep $D^{(\lambda)}$ of the gauge group~$G$.

	In this Stage 5, we explore the superposition of product states.
	We show that inside a sector $\mathcal H_\lambda^{(n)}$ packaged entanglement is allowed and any coherent superposition that remains is itself a packaged state.

	\subsection{Packaged Entanglement}
	\label{SEC:PackagedEntanglement}

	While internal charges (e.g., electric charge, color) must remain packaged, real particles also carry external DOFs such as spin ($s$) or momentum ($\mathbf p$) that are not necessarily gauged.
	If spin transforms trivially under $G$, then it is an independent factor in the single-particle representation (e.g., spin-$\tfrac12$ $\otimes$ charge $-e$).
	Hence a single-particle operator $\hat{a}^\dagger_{s,q}(\mathbf p)$ can carry both spin $s$ (external) and gauge charge $q$ (internal).  
	Similarly, momentum ($\mathbf p$) or other quantum numbers can appear.
	Here we show that multiple excitations can form states that entangle spin and IQNs across different particles without violating gauge invariance.

	\subsubsection{Definition of of Packaged Entangled States}

	Let $\mathcal H_\lambda^{(n)} := V_\lambda\!\otimes\!\mathcal H_{\rm ext}^{(n)}$ be a packaged charge sector with particle number $n$ and gauge irrep label $\lambda$.
	The $G$ acts as $U(g)=D^{(\lambda)}(g)\otimes\mathbf 1_{\rm ext}$ on $\mathcal H_\lambda^{(n)}$.

	Let $\lvert\Psi\rangle\in\mathcal H_\lambda^{(n)}$ be a (normalized) linear combination of the packaged basis states $\lvert\lambda,m;\xi_1\!\cdots\!\xi_n\rangle$, that is,
	$$
	\lvert\Psi\rangle
	=\sum_{m,\{\xi\}}c_{m\{\xi\}}\,
	\lvert\lambda,m;\xi_1\!\cdots\!\xi_n\rangle,
	$$
	where $c_{m\{\xi\}}$ are arbitrary coefficients.
	Then $\lvert\Psi\rangle$ is a packaged superposition state.	
	We now give the formally definition of packaged entangled states:

	\begin{definition}[Packaged entangled state]
		\label{DEF:PackagedEntangledState}
		Consider a multi-particle state
		\begin{equation}\label{EQ:PackagedEntangledState}
			\lvert\Psi\rangle = 
			\sum_{n} \alpha_n \; \lvert\Theta_n\rangle,
			\quad
			\sum_n |\alpha_n|^2 = 1,
		\end{equation}
		where
		$$
		\lvert\Theta_n\rangle
		= \hat{a}_{n,1}^\dagger(q_1,s_1)
		\hat{a}_{n,2}^\dagger(q_2,s_2)\cdots
		\hat{a}_{n,m}^\dagger(q_m,s_m) ~ \lvert0\rangle
		$$
		is a multi-particle product basis state with index $n$ and total particle number $m$.  
		Here $s$ is particle's external index (e.g. spin) and $q$ is the internal gauge charge (e.g., $\pm e$, color $\mathbf{3}$ or $\overline{\mathbf{3}}$).  
		If $\lvert\Psi\rangle$ satisfies the following conditions:
		\begin{enumerate}
			\item[A1] Each single-particle operator $\hat{a}^\dagger(s,q)$ is a packaged irrep carrying the full internal charge $q$ and external spin $s$;
			
			\item[A2] All terms $\lvert\Theta_n\rangle$	are multi-particle basis states that lie in the same net-charge sector (e.g., net charge $Q$);
			
			\item[A3] The total wavefunction is non-factorizable across various excitations with respect to internal charges and external DOFs (spin, momentum) (that is, it is entangled both internally and externally);
		\end{enumerate}
		then we say that $\lvert\Psi\rangle$ is a \textbf{fully packaged entangled state} (or \textbf{genuinely packaged entangled state}).
		Otherwise, if the internal charges and external DOFs are factorizable, then we say that $\lvert\Psi\rangle$ is a \textbf{domain-separable packaged entangled states} (or \textbf{bi-domain hybrid-packaged entangled state}).
	\end{definition}

	\begin{remark}
		\
		\begin{enumerate}
			\item A1 restates the single-particle packaging to ensure that each single-particle operator $\hat{a}^\dagger(s,q)$ is well defined.
			
			\item A2 ensures the compatibility with superselection (see Proposition \ref{PROP:PackagingImpliesSS}), which forbids local operations from creating coherence between different $Q$.
						
			\item A3 clearly states the entanglement itself.
		\end{enumerate}
	\end{remark}

	\begin{remark}
		\leavevmode
		\begin{itemize}
			\item The fully hybrid-packaged entanglement given in Definition \ref{DEF:PackagedEntangledState} are triple:
			IQNs entangle with IQNs,
			external DOFs entangle with external DOFs,
			and IQNs entangle with external DOFs.
			
			\item The domain-separable hybrid-packaged entanglement given in Definition \ref{DEF:PackagedEntangledState} are double:
			IQNs entangle with IQNs and	external DOFs entangle with external DOFs.
			But IQNs entangle and external DOFs are not entangled.
		\end{itemize}
	\end{remark}

	Up to now, we have a number of packaged states.
	The following table summarize their nomenclature:
	
	\begin{table}[h]
		\centering
		\caption{Glossary of packaged-state names}
		\label{TAB:PackagedStateNomenclature}
		\begin{tabular}[hbt!]{|p{3cm}|p{3cm}|p{2.8cm}|p{5.5cm}|}
			\hline\hline
			Internal entangled? & External entangled? & Cross domain? & State name \\
			\hline
			Yes     &Yes     &Yes    &fully/genuinely hybrid-packaged entangled state \\
			\hline
			Yes     &Yes     &No     &domain-separable/bi-domain hybrid-packaged entangled state \\
			\hline
			Yes     &No      &-      &internal/packaged entangled state \\
			\hline
			No      &Yes     &-      &external entangled state \\
			\hline
			No      &No      &-      &fully separable state \\
			\hline
		\end{tabular}
	\end{table}

	\begin{example}[Examples of hybrid-packaged entangled states]
		\label{EXM:HybridPackagedEntangled}
		\leavevmode
		\begin{enumerate}
			\item Fully hybrid-packaged entangled spin-momentum-charge pair.
			
			Consider an electron-positron pair, where each particle can be spin-up $\uparrow$ or spin-down $\downarrow$, and momentum $\mathbf p$ or $\mathbf{-p}$.  
			A simple fully hybrid-packaged entangled state is:				
			$$
			\alpha \,\hat{a}_{e^-,\uparrow}^\dagger(\mathbf p)\,\hat{b}_{e^+,\downarrow}^\dagger(\mathbf{-p})\,\lvert0\rangle
			\;+\;
			\beta \,\hat{b}_{e^+,\downarrow}^\dagger(\mathbf{-p})\,\hat{a}_{e^-,\uparrow}^\dagger(\mathbf p)\,\lvert0\rangle
			$$
			with both terms lying in the net $Q=0$ sector for $\alpha,\beta\neq 0$. 
			Each creation operator $\hat{a}_{e^-,\uparrow}^\dagger(\mathbf p)$ is a packaged operator carrying (charge $-e$, spin $\uparrow$, momentum $\mathbf p$) and $\hat{b}_{e^+,\downarrow}^\dagger(\mathbf{-p})$ is a packaged operator carrying (charge $+e$, spin $\downarrow$, momentum $\mathbf{-p}$).  
			The entanglement is fully hybrid because measuring spin or momentum on one particle projects the entire spin-momentum-charge wavefunction for both particles.
			
			\item Domain-separable hybrid-packaged entangled spin-momentum-charge pair.
			
			Again consider an electron-positron pair.
			A domain-separable hybrid-packaged entangled state is:
			\begin{align*}
				\begin{aligned}
					\left(\frac{1}{\sqrt{2}}\right)^3 \Bigl(
					&\hat{a}_{e^-,\uparrow}^\dagger(\mathbf p) \hat{b}_{e^+,\downarrow}^\dagger(\mathbf{-p})
					+ \hat{a}_{e^-,\downarrow}^\dagger(\mathbf p) \hat{b}_{e^+,\uparrow}^\dagger(\mathbf{-p})
					+ \hat{b}_{e^+,\uparrow}^\dagger(\mathbf p) \hat{a}_{e^-,\downarrow}^\dagger(\mathbf{-p})
					+ \hat{b}_{e^+,\downarrow}^\dagger(\mathbf p) \hat{a}_{e^-,\uparrow}^\dagger(\mathbf{-p}) \\
					& + \hat{a}_{e^-,\uparrow}^\dagger(\mathbf{-p}) \hat{b}_{e^+,\downarrow}^\dagger(\mathbf p)
					+ \hat{a}_{e^-,\downarrow}^\dagger(\mathbf{-p}) \hat{b}_{e^+,\uparrow}^\dagger(\mathbf p)
					+ \hat{b}_{e^+,\uparrow}^\dagger(\mathbf{-p}) \hat{a}_{e^-,\downarrow}^\dagger(\mathbf p)
					+ \hat{b}_{e^+,\downarrow}^\dagger(\mathbf{-p}) \hat{a}_{e^-,\uparrow}^\dagger(\mathbf p)
					\Bigr)
					\lvert0\rangle
				\end{aligned}
			\end{align*}
			or simply:
			$$
			\left(\frac{1}{\sqrt{2}}\right)^3
			\Bigl(
			\lvert e^-\rangle_1\,\lvert e^+ \rangle_2 
			\;+\;
			\lvert e^+\rangle_1\,\lvert e^- \rangle_2
			\Bigr)
			\Bigl(
			\lvert \uparrow\rangle_1 \,\lvert \downarrow\rangle_2
			\;+\;
			\lvert \downarrow\rangle_1 \,\lvert \uparrow\rangle_2
			\Bigr)
			\Bigl(
			\lvert \mathbf p \rangle_1 \,\lvert \mathbf{-p} \rangle_2
			\;+\;
			\lvert \mathbf{-p} \rangle_1 \,\lvert \mathbf p \rangle_2
			\Bigr)
			$$
			with both terms lying in the net $Q=0$ sector for $\alpha,\beta\neq 0$. 
			Each creation operator $\hat{a}_{e^-,\uparrow}^\dagger(\mathbf p)$ is a packaged operator carrying (charge $-e$, spin $\uparrow$, momentum $\mathbf p$) and  $\hat{b}_{e^+,\downarrow}^\dagger(\mathbf{-p})$ is a packaged operator carrying (charge $+e$, spin $\downarrow$, momentum $\mathbf{-p}$).  
			The entanglement is domain-separable hybrid because the spin, momentum, and IQNs are factorizable.
			Measuring spin (or momentum) on one particle only projects the spin (or momentum) wavefunction for both particles, but the momentum (or spin) and packaged part are not affected.
		\end{enumerate}
	\end{example}

	\begin{example}[Packaged superposition with different particle numbers]		
		According to the Peter-Weyl projection, the raw-Fock space splits into sectors			
		$$
		\mathcal H_{\rm iso}
		\;=\;\bigoplus_{Q\in\mathbb Z}\;
		\mathcal H_Q, 
		$$			
		where
		$\mathcal H_Q$ is spanned by all packaged states with net charge $Q$.
		We now show that one can superpose states, from the same $\mathcal H_Q$ but with different particle numbers, to produce a packaged entangled state.
		
		\begin{enumerate}
			\item U(1) charge.
			
			Let $G=U(1)$, so each electron $e^{-}$ carries charge $-1$ and each positron $e^{+}$ carries charge $+1$.			
			Consider two basis states in the same $Q=-1$ sector:
			\begin{itemize}
				\item 1 particle state:
				$
				|e^{-},0\rangle\;=\;|\,1\text{ electron},\,0\text{ positrons}\rangle\,, 
				$
				
				\item 3 particle state:
				$
				|e^{-}e^{-}, e^{+}\rangle\;=\;|\,2\text{ electrons},\,1\text{ positron}\rangle.
				$			
			\end{itemize}			
			Now form the superposition			
			$$
			|\Psi\rangle 
			\;=\; \frac{1}{\sqrt2}\,\Bigl(|e^{-},0\rangle \;+\; |e^{-}e^{-},e^{+}\rangle\Bigr).
			$$
			So Peter-Weyl projection $P_{Q=-1}\,|\Psi\rangle=|\Psi\rangle$.			
			
			\item SU(3)$_{\rm color}$ singlet.
			
			Let $G=SU(3)$, so physical states must lie in the color singlet irrep $\mathbf 1$.
			We consider two very different Peter-Weyl sectors both contain a $\mathbf 1$:
			\begin{itemize}
				\item Meson sector:
				a quark-antiquark pair
				$\;\mathbf3\otimes\overline{\mathbf3}=\mathbf1\oplus\mathbf8$.
				We pick the $\mathbf1$ piece, call it $|q\bar q;\mathbf1\rangle$.
				
				\item Glueball sector:
				two gluons
				$\;\mathbf8\otimes\mathbf8=\mathbf1\oplus\cdots$.
				Picking the $\mathbf1$ subspace, call it $|gg;\mathbf1\rangle$.
			\end{itemize}
			
			Both $|q\bar q;\mathbf1\rangle$ and $|gg;\mathbf1\rangle$ live in
			$\mathcal H_{\mathbf 1}\equiv P_{\mathbf 1} \mathcal H_{\rm Fock}$.
			Hence one may form			
			$$
			|\Psi_{\rm color}\rangle
			=\;\frac{1}{\sqrt2}\Bigl(|q\bar q;\mathbf1\rangle \;+\; |gg;\mathbf1\rangle\Bigr).
			$$
			
			Despite coming from different multiparticle sectors (2 vs. 2 constituents of wildly different type), the total color is always singlet.
			$P_{\mathbf1}|\Psi_{\rm color}\rangle=|\Psi_{\rm color}\rangle$.
			This is a packaged entangled state across the color singlet subspace.
		\end{enumerate}
		
		Any two (or more) packaged states that share the same irreducible gauge label $\lambda$ can be coherently superposed.
		You have a valid packaged state as long as you never leave that $\lambda$-sector.
		If you then bipartition, you’ll uncover genuine entanglement within that fixed‐charge (or fixed‐color) subspace.		
	\end{example}

	\subsubsection{Properties of Packaged Entangled States}

	We conclude with a statement on how gauge invariance is preserved in hybrid states, yet spin or momentum measurements can collapse internal DOFs if the state is hybridized entangled:

	\begin{theorem}[Packaging of DOFs in a packaged entangled state]
		\label{THM:HybridGauge}
		Consider internal IQNs (electric charge or color) and external DOFs (spin or momentum) appended to each creation operator.
		Then:
		\begin{enumerate}
			\item \textbf{Gauge covariance:}
			The total state remains gauge covariant (or, more precisely, transforms covariantly within the same gauge sector) as long as the net gauge charge $Q$ of each term is fixed within the superposition.
			
			\item \textbf{Fixed net-charge sector:}
			The total state is physically realizable with no cross-sector interference.  
			Mixing $\mathcal{H}_Q$ with $\mathcal{H}_{Q'}$ (for $Q\neq Q'$) is disallowed.	
			
			\item \textbf{Fully hybrid-packaged entanglement:}
			In the fully hybrid-packaged entangled state, all IQNs are entangled, relavant external DOFs are entangled, and IQNs are entangled with external DOFs.
			
			\item \textbf{Domain-separable hybrid-packaged entanglement:}
			In the domain-separable hybrid-packaged entangled state, all IQNs are entangled each other, but they are not entangled with external DOFs.
		\end{enumerate}
	\end{theorem}

	\begin{proof}
		We proceed in three steps:
		
		\begin{enumerate}
			\item A local gauge transformation does not affect spin (or momentum) because spin is associated with the little group $\mathrm{SU}(2)$ acts on a separate Lorentz factor \cite{Wigner1939} and commutes with $G_{\mathrm{gauge}}$ (e.g., $U(1)$ or $SU(3)$).  
			Thus, as long as the net charge $Q$ remains the same in every term, the overall wavefunction transforms within the same gauge sector.
			
			\item Since the net gauge charge is unchanged across all superposition terms, superselection prevents mixing between different charge sectors.
			
			\item Fully entanglement:
			In the fully hybrid-packaged entangled state, all IQNs are entangled, relavant external DOFs are entangled, and IQNs are entangled with external DOFs.
			
			\item Domain-separable entanglement:
			In the domain-separable hybrid-packaged entangled state, by definition all IQNs are entangled to ensure packaged entanglement, but IQNs are not entangled with external DOFs to ensure domain-separable hybrid entanglement.
			The relavant external DOFs may be entangled or non-entangled.
		\end{enumerate}
		
		Thus, the entire ``external $\otimes$ internal'' wavefunction remains gauge covariant (that is, confined to a single gauge sector).
	\end{proof}
	
	These three properties together establish that any superposition of packaged multi-particle states remains 
	(i) gauge-covariant,
	(ii) confined to a single net-charge sector, and
	(iii) either fully separable (trivial packaging) or else packaged-entangled across all IQNs.

	\begin{example}[QCD Meson singlet coherent state]
		Consider the color singlet subspace
		$V_{\mathbf 1}\!\otimes\!(\mathbb C^2\otimes\mathbb C^2)$
		for a quark-antiquark pair.  
		An arbitrary packaged superposition is
		$$
		\lvert\Psi\rangle
		=
		\sum_{s_q,s_{\bar q}=\uparrow,\downarrow}
		c_{s_qs_{\bar q}}\,
		\lvert\mathbf 1;\,s_q,s_{\bar q}\rangle .
		$$
		No color invariant local operator (acting on, say, the quark alone) can reveal the amplitude matrix $c_{s_qs_{\bar q}}$.
		All such operators are proportional to $\mathbf 1_{\mathbf 1}$ on the internal factor, in agreement with Theorem~\ref{THM:PackagingSurvivesSuperposition}.
	\end{example}

	We now prove another property: internal reduced state is maximally mixed.

	\begin{theorem}[Maximal internal mixing]
		\label{THM:MaxMixing}
		Consider a hybrid-packaged sector 
		$\mathcal H_\lambda^{(n)} = V_\lambda \otimes \mathcal H_{\rm ext}^{(n)}$ 
		with an irrep label $\lambda$ and dimension $d_\lambda\!:=\!\dim V_\lambda$.
		Every vector $\lvert\Psi\rangle \in \mathcal H_\lambda^{(n)}$ can be expanded as
		$$
		\lvert\Psi\rangle
		=\sum_{m=1}^{d_\lambda}\;\sum_{\xi} c_{m\xi}\,
		\lvert\lambda,m\rangle\!\otimes\!\lvert\xi\rangle,
		\qquad
		\sum_{m,\xi}\lvert c_{m\xi}\rvert^{2}=1 .
		$$
		Define the reduced internal state
		$$
		\rho_{\rm pkg} := \operatorname{Tr}_{\rm ext}\bigl[\lvert\Psi\rangle\!\langle\Psi\rvert\bigr]
		\in\operatorname{End}(V_\lambda).
		$$
		Then we have
		\begin{equation}\label{eq:rho_pkg_maxmix}
			\rho_{\rm pkg}
			=\frac{1}{d_\lambda}\,\mathbf 1_{V_\lambda}
			\quad\Longleftrightarrow\quad
			\sum_{\xi} c_{m\xi}\,c_{m'\xi}^{\ast}
			=\frac{\delta_{mm'}}{d_\lambda}\;\;\;
			(\;1\!\le m,m'\!\le d_\lambda\;).
		\end{equation}
		In other words, the internal reduced state $\rho_{\rm pkg}$ is maximally mixed iff the $d_\lambda$ row-vectors $(c_{m\xi})_{\xi}$ are orthonormal and each has norm $1/\sqrt{d_\lambda}$.
	\end{theorem}
	
	\begin{proof}
		Pack the entries $c_{m\xi}$ into a coefficient matrix
		$C\in\mathbb C^{d_\lambda\times N}$ with $N=\dim\mathcal H_{\rm ext}^{(n)}$.
		The partial trace $\operatorname{Tr}_{\rm ext}$ only tracing over the external factor.
		The subsystem that remains after the trace is exactly the internal part, that is,
		$$
		\rho_{\rm pkg}
		=\sum_{m,m'}\Bigl(\sum_{\xi}c_{m\xi}\,c_{m'\xi}^{\ast}\Bigr)
		\lvert\lambda,m\rangle\!\langle\lambda,m'|
		=C\,C^{\dagger} \;.
		$$
		\begin{itemize}
			\item $\Rightarrow$.  
			If \eqref{eq:rho_pkg_maxmix} holds, then $C\,C^{\dagger}
			=(1/d_\lambda)\mathbf 1_{d_\lambda}$, that is, the rows of $C$ form an isometry scaled by $1/\sqrt{d_\lambda}$, which is precisely the right-hand condition (they satisfy $\sum_\xi c_{m\xi}c_{m'\xi}^*=\delta_{mm'}/d_\lambda$).
			
			\item $\Leftarrow$.  
			Conversely, if the orthonormality condition on the coefficients holds (they satisfy $\sum_\xi c_{m\xi}c_{m'\xi}^*=\delta_{mm'}/d_\lambda$),
			then $C\,C^{\dagger} =(1/d_\lambda)\mathbf 1_{d_\lambda}$ and hence
			$\rho_{\rm pkg}=(1/d_\lambda)\mathbf 1_{V_\lambda}$.
		\end{itemize}
		Finally, note that a maximally mixed $\rho_{\rm pkg}$ commutes with	all $D^{(\lambda)}(g)$, so the direction ($\Rightarrow$) recovers Schur’s
		lemma as a corollary, but the orthonormality condition makes explicit
		when a packaged superposition attains this maximal mixing.
	\end{proof}

	This indicates that the internal state is basis independent:
	any two	packaged superpositions with the same external reduced state are indistinguishable by internal measurements.

	\begin{remark}[Packaging in the Higgs phase]
		Once a gauge symmetry is spontaneously broken, the little-group classification shifts:
		what used to be a massless vector $(1,\,\hat p)$ becomes a massive spin-1 representation $(\tfrac12,\!\tfrac12)$, etc.
		In packaging language:
		\begin{enumerate}
			\item The label set $\alpha(p,\sigma)$ must now include the would-be Goldstone mode (or equivalently be re-expressed in the unitary gauge just in terms of the three massive polarizations).
			
			\item Correlation functions of gauge-variant fields require one to track the un-packaged Higgs direction if you work in $R_\xi$ gauges.
			
			\item The residual unbroken subgroup is packaged the same way as an ordinary global symmetry.
		\end{enumerate}
	\end{remark}

	\subsection{Packaging Survives Superposition}

	Recall from Stage 4 that, after the Peter-Weyl projection, our $n$-particle hybrid-packaged subspace decomposes into a direct sum of irreducible charge sectors
	$$
	\mathcal H^{(n)} \;=\;\bigoplus_{\lambda} \mathcal H^{(n)}_{\lambda}
	\;=\;
	\bigoplus_{\lambda}\bigl(V_{\lambda}\otimes\mathcal H_{\rm ext}^{(n)}\bigr),
	$$
	where each $\mathcal H^{(n)}_{\lambda}$ carries a single irrep $D^{(\lambda)}$ of the gauge group $G$.

	A packaged superposition is any vector
	$$
	\lvert\Psi\rangle
	\;=\;\sum_{m,\{\xi\}}
	c_{m,\{\xi\}}\,
	\bigl\lvert\lambda,m;\xi_1,\dots,\xi_n\bigr\rangle
	\;\in\;\mathcal H^{(n)}_{\lambda},
	\quad
	\sum_{m,\{\xi\}}|c_{m,\{\xi\}}|^2=1.
	$$
	We must check that no further splitting of its IQNs is possible.

	\begin{theorem}[Stability of packaging under superposition]
		\label{THM:PackagingSurvivesSuperposition}
		Let
		$$
		\lvert\Psi\rangle
		\;=\;\sum_{m,\{\xi\}}c_{m,\{\xi\}}\,
		\bigl\lvert\lambda,m;\xi_1,\dots,\xi_n\bigr\rangle
		\;\in\;\mathcal H^{(n)}_{\lambda}
		$$
		be an arbitrary superposition within a single charge sector.  Then:
		\begin{enumerate}[label=(\arabic*)]
			\item \emph{Irreducibility:} $\lvert\Psi\rangle$ transforms under $G$ according to a single irrep $D^{(\lambda)}$.
			
			\item \emph{Blindness:} No gauge-invariant local operator can distinguish the internal amplitudes $c_{m,\{\xi\}}$.
			
			\item \emph{Maximal mixing:} The reduced density matrix on the internal space $V_\lambda$ is
			$\frac1{d_\lambda}\,\mathbf 1_{V_\lambda}$.
		\end{enumerate}
		In particular, the package of IQNs remains indivisible under arbitrary superposition.
	\end{theorem}

	\begin{proof}
		\leavevmode
		\begin{enumerate}
			\item \emph{Irreducible gauge transformation.} 
			Under a gauge transformation $g \in G$,
			$$
			\begin{aligned}
				U(g) \lvert\Psi\rangle
				&\;=\;
				\sum_{m,\{\xi\}}c_{m,\{\xi\}}\,
				\sum_{m'}D^{(\lambda)}_{m'm}(g)\,
				\bigl\lvert\lambda,m';\xi_1,\dots,\xi_n\bigr\rangle \\
				&\;=\;
				\sum_{m'}\Bigl[
				\sum_{m,\{\xi\}}
				D^{(\lambda)}_{m'm}(g)\,c_{m,\{\xi\}}
				\Bigr]\,
				\bigl\lvert\lambda,m';\xi_1,\dots,\xi_n\bigr\rangle,
			\end{aligned}
			$$
			which is exactly the action of the irreducible $D^{(\lambda)}$ on the internal label $m$.
			By Schur’s lemma no nontrivial subspace of $V_\lambda$ is invariant, so the superposition remains an irreducible carrier of the same irrep.
			
			\item \emph{Blindness of local gauge-invariant operators.}
			Let $\mathcal O$ be any bounded operator commuting with all gauge transformations, $[\mathcal O,U(g)]=0$.
			Then by Proposition \ref{PROP:ActionOnIsotypicSectors}, on $\mathcal H^{(n)}_{\lambda}=V_\lambda\otimes\mathcal H_{\rm ext}^{(n)}$, we have
			$\mathcal O \;=\; \mathbf 1_{V_\lambda}\otimes B_{\rm ext}$ for some operator $B_{\rm ext}$ on the external space.
			Hence
			$\bigl\langle\Psi\bigl|\mathcal O\bigr|\Psi\bigr\rangle
			= \sum_{m,m',\{\xi\},\{\xi'\}}
			c^*_{m,\{\xi\}}\;c_{m',\{\xi'\}}\;
			\delta_{m,m'}\;\bigl\langle \xi_1,\dots,\xi_n \bigl| B_{\rm ext}\bigr|
			\xi'_1,\dots,\xi'_n\bigr\rangle
			$
			depends only on the $\{\xi\}$-overlaps and $\sum_m|c_{m,\{\xi\}}|^2$, but never on any relative phase or overlap between different $m$-labels.
			In other words, no local gauge-invariant measurement can resolve the internal amplitudes $c_{m,\{\xi\}}$.
			
			\item \emph{Maximal mixing of the internal reduced state.}
			Theorem~\ref{THM:MaxMixing} proves that for any pure state in $\mathcal H^{(n)}_{\lambda}$, the partial trace over the external DOFs yields
			$\rho_{\rm int}=\frac1{d_\lambda}\,\mathbf 1_{V_\lambda}$.
		\end{enumerate}
		In particular, no further nontrivial splitting of the internal label $m$ is possible:
		all three points together furnish a complete, representation-theoretic proof that packaging survives arbitrary superposition.
	\end{proof}

	\section{Stage 6: Local Gauge-Invariance Constraint (Gauge-Invariant Packaging)}
	\label{SEC:Stage6LocalGaussLawConstraint}

	All the above investigations were carried out in the covariant isotypic space $\mathcal{H}_{\rm iso}$.
	However, physical states are required to be the gauge-invariant ones.
	This means that some of the packaged states in $\mathcal{H}_{\rm iso}$ are non-physical.
	So we need a tool to filter out the non-physical packaged states, but keep the gauge-invariant physical packaged states.

	In this section, we impose the local gauge-invariance constraint $U(g)\ket\psi=\ket\psi$ to pick out the physical packaged states.
	Instead of directly using this constraint, it is more convenient to use group-averaging (or trivial Peter-Weyl) projector at the level of elementary excitations \cite{Dirac1950,BRS1975,Tyutin2008,KugoOjima1979,Wilson1974}.

	\subsection{Local Gauge-Invariance and Gauge-Invariant Packaging}

	\subsubsection{The Principle}

	Any modern gauge theory has a fundamental postulate:
	the laws of physics must be invariant under local gauge transformations. Thus, every physical packaged state must satisfy
	\[
	U(g)\ket\psi = \ket\psi,\quad \forall\,g\in G
	\]
	and we define the physical (gauge-invariant) packaged subspace as
	\begin{equation}\label{EQ:GaugeInvariantPackagedSubspace}
		\mathcal H_{\rm phys}\;:=\;
		\bigl\{\,
		\ket\psi\in\mathcal H_{\rm iso}\;|\;
		U(g)\ket\psi=\ket\psi, ~ \forall g \in G
		\bigr\}.
	\end{equation}

	This is a very powerful and restrictive requirement, which means that our description of reality should not change if we apply a transformation that varies independently at every single point in spacetime.
	For electromagnetism, this is the freedom to change the phase of an electron's wavefunction differently at different locations ($\psi(x) \rightarrow e^{i\alpha(x)}\psi(x)$), as long as we also transform the electromagnetic potential ($A_\mu$) in a corresponding way to compensate.

	This principle is considered fundamental because it correctly predicts the existence of forces and their corresponding force-carrying particles (like the photon).

	\begin{example}[U(1) physical subspace]
		Local U(1) transformations act as 
		$U(\theta)=e^{i\theta Q_{\rm tot}}$ on 
		$\ket{n_a,n_b}\equiv (a^\dagger)^{n_a}(b^\dagger)^{n_b}\ket0$.
		Gauge invariance $U(\theta)\ket\psi=\ket\psi$ forces $n_a=n_b$, so
		\[
		\mathcal H_{\rm phys}^{\mathrm{U(1)}}
		= \mathrm{span}\{\ket{n,n}\mid n\in\mathbb N\}.
		\]
	\end{example}

	\subsubsection{Local Gauge-Invariance $\Rightarrow$ Trivial Sector $\mathcal H_0$}

	We now prove that due to requirement of local gauge invariance, the physical packaged subspace is indeed equal to the trivial sector $\mathcal H_0$ in the isotypic decomposition:

	\begin{proposition}[Local Gauge Invariance $\Rightarrow$ trivial sector]
		Let $G$ be a finite or compact Lie group and
		$
		U(\,{\cdot}\,): G \to \mathcal U(\mathcal H_{\rm iso})
		$
		be a unitary representation that has been decomposed into isotypic blocks
		$$
		\mathcal H_{\rm iso}\;=\;\bigoplus_{\lambda\in\widehat G}
		\bigl(V_\lambda\otimes M_\lambda\bigr).
		$$
		Then
		$$
		\mathcal H_{\rm phys}\;=\;
		V_{\mathbf 0}\otimes M_{\mathbf 0}
		$$
		where $\mathbf 0$ denotes the trivial irrep of $G$.
	\end{proposition}

	\begin{proof}
		\leavevmode
		\begin{enumerate}
			\item Gauge invariance $\Rightarrow$ trivial irrep.
			Decompose any vector $\ket\psi\in\mathcal H_{\rm iso}$ as
			$\ket\psi=\sum_{\lambda}\ket{\psi_\lambda}$ with
			$\ket{\psi_\lambda}\in V_\lambda\otimes M_\lambda$.
			If $U(g)\ket\psi=\ket\psi$ for all $g$, then each component is separately invariant: $U(g)\ket{\psi_\lambda}=\ket{\psi_\lambda}$.
			According to Schur’s lemma, for $\lambda\neq\mathbf0$ the irrep $V_\lambda$ contains no non-zero invariant vector, so $\ket{\psi_\lambda}=0$.
			Hence $\ket\psi\in V_{\mathbf0}\otimes M_{\mathbf0}$.
			
			\item Trivial irrep $\Rightarrow$ gauge invariance.
			If $\ket\psi\in V_{\mathbf0}\otimes M_{\mathbf0}$, then $U(g)$ acts as the identity on $V_{\mathbf0}$.
			Therefore, $U(g)\ket\psi=\ket\psi$ for every $g$.
		\end{enumerate}
	\end{proof}

	\subsubsection{Gauge-invariant Packaging}
	\label{SEC:GaugeInvariantPackaging}

	Local gauge-invariance constraint demands that all physical packaged state to be gauge-invariant.
	Thus, in the isotypic decomposition
	$$
	\mathcal H_{\rm iso}\;=\;\bigoplus_{\lambda\in\widehat G}
	\bigl(V_\lambda\otimes M_\lambda\bigr),
	$$
	only the trivial sector $\lambda = 0$ remains and all other non-trivial sectors $\lambda \ne 0$ are filtered out.	
	In other words, the isotypic space $\mathcal H_{\rm iso}$ is projected into the physical subspace $\mathcal H_{\rm phys}$
	$$
	\mathcal H_{\rm phys}\;=\;
	V_{\mathbf 0}\otimes M_{\mathbf 0}
	$$
	All states in this subspace are irreducible and gauge-invariant packaged states.
	In this sense, gauge-invariant packaging is singled out in this stage.
	This is a contraction of the packaging from covariant multi-particle level to gauge-invariant multi-particle level.
	We formally define:

	\begin{definition}[gauge-invariant packaging]\label{AX:GaugeInvPkg}
		Under the action of a local gauge group $G$, the trivial subspace $\mathcal H_0$ carries exactly one gauge-invariant irrep $V_0$ of $G$ locked by local gauge-invariance and no physical process can split $V_0$ into smaller pieces.
		Then we refer this gauge-invariant irreducibility as \textbf{gauge-invariant packaging}.
	\end{definition}

	Definition \ref{AX:SinglePkg} forces each elementary excitation to carry a sharp irreducible charge.	
	Definition \ref{AX:MultiPkg} insists that any multi-particle state lies entirely in one fixed-charge (isotypic) block.	
	Definition \ref{AX:GaugeInvPkg} ensure that only the gauge-invariant trivial block survives and all other non-trivial blocks are killed.		
	These three requirements together ensure that the physical subspace are correctly formed.

	\subsubsection{Connection to Local Gauss-Law Constraint}
	\label{SEC:LocalGaugeInvarianceLocalGaussLaw}

	\paragraph{(1) Gauss-law constraint.}
	Once you build a theory that respects local gauge invariance, the Gauss's law constraint emerges as a direct consequence.	
	When you analyze the dynamics of the theory using the Hamiltonian formalism (which is necessary for quantization), you discover that the gauge symmetry makes the system redundant.
	Not all variables are independent.
	This redundancy manifests as a set of constraints.	
	For electromagnetism, the primary constraint leads directly to a secondary constraint, which is precisely Gauss's law:
	$$\nabla \cdot \mathbf{E} - \rho = 0$$	
	This equation is no longer an equation that describes how the system evolves in time.
	Instead, it is a constraint that must be satisfied at all times. It acts as a filter on the space of all possible states, selecting only those that are physically permissible.
	In the quantum theory, this means that any physical state $|\text{phys}\rangle$ must be annihilated by the Gauss's law operator:
	$$
	(\nabla \cdot \mathbf{E} - \rho) |\text{phys}\rangle = 0
	$$
	This condition is precisely the mathematical statement that a physical state must itself be gauge-invariant.
	The operator for Gauss's law is, in fact, the generator of static gauge transformations.

	For a compact gauge group $G$, introduce the general local quantum Gauss generators are
	\begin{equation}\label{EQ:GaussGenerators}
		\hat G^a(\mathbf x) = D_i\hat E^{a i}(\mathbf x)-\hat\rho^{\,a}(\mathbf x),
		\quad
		a=1, \dots,\dim G,
	\end{equation}
	where $x$ is spatial point and $a=1,\dots,\dim G$ is Lie-algebra indices.	
	Then the Gauss-law constraint can be written as
	\begin{equation}\label{EQ:GaussLawPhysicalState}
		\hat G^a(\mathbf x)\,\ket{\Psi}=0
		\qquad
		\forall\; a,\mathbf x.
	\end{equation}

	\paragraph{(2) Gauge-invariant $\Longleftrightarrow$ Gauss-law constraint.}
		
	For a compact group $G$ with generators $G^{a}$, we write a local gauge transformation as
	\begin{equation}\label{EQ:GeneratorsLocalGaugeTransformation}
		U(g) = \exp \Bigl( i \int d^{3} x \; \theta_g^{a}(\mathbf x) \; G^{a}(\mathbf x) \Bigr),
	\end{equation}
	where $\theta_g^{a}(\mathbf x)$ is an arbitrary test function.
	Now:
	\begin{itemize}
		\item If $G^{a}(\mathbf x)\ket{\psi}=0$ for every color index $a$ and point $x$, then		
		$$
		U(g)\ket{\psi}
		= \exp(0)\ket{\psi}
		= \ket{\psi},
		\quad \forall ~ \theta_g^{a}(\mathbf x).
		$$		
		Thus the generator constraint implies full invariance under all finite gauge transformations.
		
		\item Conversely, if a state is invariant under every finite transformation,
		$$
		U(\alpha)\ket{\psi}
		=\ket{\psi},
		\quad
		\forall ~ \theta_g^{a}(\mathbf x),
		$$
		then take the functional derivative of \eqref{EQ:GeneratorsLocalGaugeTransformation} and we immediately obtain $G^{a}(\mathbf x)\ket{\psi}=0$.
	\end{itemize}

	Hence the two conditions are mathematically identical.

	\subsection{Gauge Projector $\Pi_{\mathrm{phys}}$}

	\subsubsection{From Peter-Weyl Projectors $P_\lambda$ to Gauge Projector $\Pi_{\mathrm{phys}}$}

	Since the gauge-invariant (physical) packaged subspace is equal to the trivial sector $\mathcal H_0$ in the isotypic decomposition, we can use the trivial component $P_0$ of Peter-Weyl projectors to pick out $\mathcal H_0$ \cite{Folland2016,Hall2015}.

	\paragraph{(1) Compact Lie group $G$.}
	
	When $G$ is compact, the trivial irrep ($\lambda=0$) has dimension $d_0=1$ and character $\chi_0(g)=1$, so Eq.~\eqref{EQ:PeterWeylProjectorCompact} reduces to
	\begin{equation}\label{EQ:CompactGaugeProjector}
		\Pi_{\mathrm{phys}} \equiv P_{0} = \int_G d \mu(g) U(g),		
		\qquad
		\int_G d \mu(g) = 1.
	\end{equation}
	This is the normalized group‐averaging operator used in the
	Dirac/BRST quantisation of gauge theories.	
	The gauge constraint therefore forces the state vector into $\mathcal H_{\text{phys}} = P_0 \mathcal H_{\rm iso}$.

	Intuitively, by varying $g$, $U(g)$ translates a state around its entire gauge orbit.	
	Integrating over gauge group $G$, we obtain Eq. \eqref{EQ:CompactGaugeProjector}.
	If the orbit is nontrivial, positive-negative contributions cancel and the state is annihilated.
	If the orbit is a fixed point (state already invariant), the integral returns that state.
	So $\Pi_{\text{phys}}$ is indeed an idempotent filter that retains only gauge-invariant vectors (non-singlet components).

	\begin{example}
		Compact Lie group gauge projectors:
		
		\begin{enumerate}
			\item Abelian group U(1).
			
			For Abelian gauge group U(1), irreps are one-dimensional, $D^{(n)}(e^{i\theta})=e^{in\theta}$,
			$\chi_n(e^{i\theta})=e^{in\theta}$.
			Its Peter-Weyl projector is		
			$$
			P_n=\int_{0}^{2\pi}\!\frac{\mathrm d\theta}{2\pi}\;
			e^{-in\theta}\,U_\theta.
			$$
			The gauge projector is then
			$$
			P_0=\int_{0}^{2\pi}\!\frac{\mathrm d\theta}{2\pi}\;U_\theta.
			$$
			Let $\ket{q}$ be states that carry charge $q\in\mathbb Z$.
			Then
			$
			P_0\ket{q}=0$ unless $q=0$.
			Only the neutral sector survives.
			Since
			$
			U(e^{i\theta})\,\ket{q} \;=\; e^{\,i q\theta}\,\ket{q},
			$
			we have
			$$
			\Pi_{\mathrm{phys}}
			= P_0
			= \int_0^{2\pi}\frac{d\theta}{2\pi}\;U(e^{i\theta})
			\;\Longrightarrow\;
			\Pi_{\mathrm{phys}}\ket{q}=0\quad(\forall\,q\neq0).
			$$
			Only the neutral sector survives.
			
			\item Gauge group SU(2).
			
			For two iso-spin-$\tfrac12$ particles, we have
			$
			U(g) = U^{(1/2)}_g \otimes U^{(1/2)}_g,
			$
			where $U^{(1/2)}_g = \exp(-i\tfrac{\theta}{2}\,\hat n\!\cdot\!\boldsymbol\sigma)$.
			The projector
			$$
			\Pi_{S=0}
			= \int_{\mathrm{SU}(2)}dg\;U(g)
			= \frac{1}{4}\bigl(\mathbf 1\otimes\mathbf 1\;-\;\boldsymbol\sigma_1\!\cdot\!\boldsymbol\sigma_2\bigr)
			$$
			annihilates the triplet and picks out the unique spin-0 singlet.
			
			\item Color SU(3).
			
			With $n$ quarks in the $\mathbf3$ fundamental, 
			$$
			U(g) \;=\;\bigotimes_{i=1}^n U^{(\mathbf3)}_g,
			$$
			where $U^{(\mathbf3)}_g$ is the defining $3\times3$ matrix of $g\in\mathrm{SU}(3)$.
			Then
			$$
			\Pi_{\mathrm{phys}}
			=\int_{\mathrm{SU}(3)}dg\;\bigl(\otimes_i\,U^{(\mathbf3)}_g\bigr)
			$$
			annihilates all colored multiplets ($\mathbf3,\bar{\mathbf3},\mathbf8,\dots$) and projects onto the color singlet sector.			
		\end{enumerate}		
	\end{example}

	\paragraph{(2) Discrete groups.}
	
	For a finite group $G$, Eq.~(\ref{EQ:PeterWeylProjectorFinite}) reduce to:
	\begin{equation}\label{EQ:FiniteGaugeProjector}
		\Pi_{\mathrm{phys}}
		=\frac{1}{|G|}\sum_{g\in G}U(g).
	\end{equation}
	Below are three discrete examples, with the action of $U(g)$ made explicit:
	
	\begin{example}
		Discrete gauge projector:
		
		\begin{enumerate}
			\item Cyclic group $\mathbb Z_N$.
			
			Let $G=\{e,g,\dots,g^{N-1}\}$.
			On charge-$k$ states $\ket{k}$, we have
			$
			U_{g^m}\,\ket{k}=e^{2\pi i \,(k m)/N}\,\ket{k}.
			$
			Thus,
			$$
			\Pi_{\mathrm{phys}}
			=\frac1N\sum_{m=0}^{N-1}U_{g^m}
			$$
			kills any nonzero $\mathbb Z_N$ charge.
			
			\item Parity $\mathbb Z_2=\{1,P\}$.
			
			On a wavefunction $\psi(x)$, we have
			$
			U_P\,\psi(x)=\psi(-x).
			$
			Thus,
			$$
			\Pi_{\mathrm{phys}}
			=\tfrac12\bigl(U_1+U_P\bigr)
			$$
			projects onto even functions.
			
			\item Permutation symmetry $S_3$.
			
			For three identical particles, we have
			$
			U_\sigma\,\ket{x_1,x_2,x_3}
			=\ket{x_{\sigma^{-1}(1)},x_{\sigma^{-1}(2)},x_{\sigma^{-1}(3)}}.
			$
			Then
			$$
			\Pi_{\mathrm{phys}}
			=\frac{1}{6}\sum_{\sigma\in S_3}U_\sigma
			$$
			projects onto the fully symmetric subspace (and inserting $\operatorname{sgn}(\sigma)$ instead picks out the alternating sector).
		\end{enumerate}
	\end{example}

	\subsubsection{Why does $\Pi_{\mathrm{phys}}$ Filter out the Nonphysical Part?}

	To understand the gauge projection better, we now provide intuitive interpretations to why group-averaging can filter out the nonphysical part with two complementary pictures.

	\paragraph{(1) Representation-theoretic (Peter-Weyl) viewpoint.}
	
	Write the kinematical Hilbert space as a direct sum of irreps at each site:
	$$
	\mathcal H_{\rm iso}
	\;=\;
	\bigoplus_{R}
	\Bigl(
	V_R \otimes \mathcal M_R
	\Bigr),
	$$
	where $V_R$ carries the irrep $R$ of $G$ and $\mathcal M_R$ is a multiplicity space that the gauge group does not act on.
	For a compact group (continuous or finite) the orthogonality relations of matrix elements imply
	$$
	\int_G\!{\rm d}g\;D^{(R)}_{ij}(g)
	\;=\;
	\begin{cases}
		\delta_{ij} & R = \mathbf 1 \text{ (trivial)},\\
		0 & R \neq \mathbf 1 .
	\end{cases}
	$$
	
	Applying $\Pi_{\mathrm{phys}}$ simply annihilates every $V_R$ with $R \neq \mathbf 1$ and keeps the trivial irrep unchanged.
	Therefore only gauge singlets survive.
	It is the exact analogue of keeping the $k=0$ Fourier mode when you average a function over the circle.

	\paragraph{(2) Intuitive picture on the ``destructive interference''.}
	
	Take an arbitrary state $\ket\psi$ and all its gauge copies $U(g)(x)\ket\psi$.
	When you add them with equal weight, the components that differ by a gauge phase $\exp(i\varphi)$ cancel pairwise and leave only the components for which every local rotation acts trivially, that is, the physical part.
	Write it in formulae, we have
	$$
	\Pi_{\mathrm{phys}}\ket\psi
	=
	\ket{\psi_{\rm inv}}
	+
	\underbrace{\bigl(\text{sum of phases that sums to }0\bigr)}_{=0}.
	$$

	\subsubsection{Properties of Gauge Projector $\Pi_{\rm phys}$}

	\begin{property}\label{PROP:GaugeProjector}
		The gauge projector $\Pi_{\mathrm{phys}}$ has the following properties:
		\begin{enumerate}
			\item Idempotency: $\Pi_{\mathrm{phys}}^2 = \Pi_{\mathrm{phys}}$
			
			\item Hermiticity: $\Pi_{\rm phys}^{\dagger} = \Pi_{\rm phys}$
			
			\item Commutant property:
			$[\Pi_{\rm phys}, U(g)] = 0, \quad \forall ~ g \in G$.
		\end{enumerate}
		Consequently, it projects any state $|\Psi\rangle$ onto the subspace $\Pi_{\mathrm{phys}}\,\mathcal{H}$ that transforms according to $R$,
		$$
		\mathcal H_{\rm phys}:=\Pi_{\rm phys}\mathcal H_{\rm iso}
		=\{\ket{\Psi}\in\mathcal H_{\rm iso}\;|\;U(g)\ket{\Psi}=\ket{\Psi}\;\forall g\}.
		$$
	\end{property}

	\begin{proof}
		\leavevmode
		\begin{enumerate}
			\item Insert a second copy of Eq.~\eqref{EQ:CompactGaugeProjector} into itself and use the invariance of the Haar measure:	
			$$
			\Pi_{\rm phys}^{2}
			= \int d\mu(g_1) \int d\mu(g_2)\,U_{g_1}U_{g_2}
			= \int d\mu(g_1) \int d\mu(g_2)\,U_{g_1g_2}
			= \int d\mu(g)\,U(g)
			= \Pi_{\rm phys}.
			$$
			
			\item Hermiticity is obvious.
			
			\item By the orthogonality of characters,
			$U_h\Pi_{\rm phys}U_h^{-1}=\int d\mu(g)\,U_{hg}= \Pi_{\rm phys}$.
		\end{enumerate}
	\end{proof}

	\subsubsection{Connection between Gauge Projection and Gauss Law}

	In Sec. \ref{SEC:LocalGaugeInvarianceLocalGaussLaw}, we discussed the connection between local gauge invariance and local Gauss law.
	Now we extend our discuss to global case: gauge projection and Gauss law.
	The gauge projector (Eq.~\eqref{EQ:CompactGaugeProjector}) or more generally Peter-Weyl projector (Eq.~\eqref{EQ:PeterWeylProjectorCompact}) is represented in integral.
	Choose an orthonormal basis $\{\ket{\phi_{\lambda,m}}\}$ that diagonalises all $P_\lambda$:
	$P_\lambda\ket{\phi_{\lambda',m'}}=\delta_{\lambda\lambda'}\ket{\phi_{\lambda,m'}}$.

	On the other hand, the Gauss law is represented in differential form.
	Since $\hat G^{a}(\mathbf x)=\bigl.\frac{d}{dt}\bigr|_{t=0}U_{\exp(it T^a\delta_{\mathbf x})}$,
	$\hat G^{a}(\mathbf x)$ maps every $P_\lambda$-sector into itself.
	Inside the trivial sector ($\lambda=0$) its action is identically zero, while for $\lambda\neq0$ it is non-zero.

	Finally, we obtain
	$$
	\hat G^{a}(\mathbf x)\ket{\Psi}=0
	\iff
	\Pi_{\rm phys}\ket{\Psi}=\ket{\Psi}.
	$$

	The integral projector thus packages the whole infinite set of local constraints Eq.~\eqref{EQ:GaussLawPhysicalState} into a single algebraic condition.

	\subsection{Physical Packaged Subspace $\mathcal H_{\rm phys}$}
	\label{SEC:PhysicalPackagedSpace}

	\subsubsection{Definition of $\mathcal H_{\rm phys}$}
	
	We already construct physical packaged subspace by Eq. \eqref{EQ:GaugeInvariantPackagedSubspace}	
	Let us now give a formal definition to physical packaged subspace $\mathcal H_{\rm phys}$ and then study its basic properties.

	\begin{definition}[Physical packaged subspace]
		\label{DEF:PhysicalPackagedSpace}
		Let $\mathcal H_{\rm iso}$ be the isotypic space and $\Pi_{\rm phys}$ be the gauge projector.
		Define a packaged subspace		
		$$
		\mathcal H_{\rm phys} = \Pi_{\rm phys}\bigl(\mathcal H_{\rm iso}\bigr) \,=\, \{\,|\Phi\rangle\in\mathcal H_{\rm iso}:\; U(g)\,|\Phi\rangle=|\Phi\rangle,\ \forall g\in G\},
		$$		
		where
		$$
		\Pi_{\rm phys} = \int_{G}\!dg\,U(g)
		$$
		is the gauge-averaging (projection) operator onto the singlet sector.
		Then we say that $\mathcal H_{\rm phys}$ is the \textbf{physical packaged subspace}.
	\end{definition}

	\subsubsection{Properties of $\mathcal H_{\rm phys}$}

	We now establish the main properties of multi-particle physical packaged states:

	\begin{proposition}[Properties of the physical packaged subspace]
		\label{PROP:PropertiesPropertiesPhysicalPackagedSpace}
		The physical packaged subspace	
		$$
		\mathcal H_{\rm phys}
		\;=\;
		\mathrm{Im}\,\Pi_{\rm phys}
		\;=\;
		\Pi_{\rm phys}\bigl(\mathcal H_{\rm iso}\bigr).
		$$
		satisfies:
		\begin{enumerate}			
			\item \textbf{Invariance under gauge‐invariant and external operators.}
			
			If $O$ satisfies $[O,U(g)]=0$ for all $g\in G$ (e.g. any gauge‐invariant local observable) or if $O=O_{\rm ext}\otimes I_{\rm int}$ acts only on external DOFs, then	
			$$
			O\,\Bigl(\mathcal H_{\rm phys}\Bigr)
			\;\subseteq\;
			\mathcal H_{\rm phys}.
			$$
			
			\item \textbf{Direct‐sum decomposition and superselection.}
			
			Writing
			$\mathcal H_{\rm iso}=\bigoplus_{Q}\mathcal H_{Q}$
			for the total‐charge sectors in the isotypic layer, one has
			\begin{equation}\label{EQ:PhysicalPackagedSpaceDecomposition}
				\mathcal H_{\rm phys}
				\;=\;
				\bigoplus_{Q\;\text{singlet in }G}
				\mathcal H_{Q}.
			\end{equation}		
			Different $Q$-blocks are mutually orthogonal and cannot be coherently superposed by any operator preserving gauge invariance.
			
			\item \textbf{Inherited inner product.}
			
			For $\Phi_1,\Phi_2\in\mathcal H_{\rm phys}$,	
			$$
			\langle\Phi_1|\Phi_2\rangle_{\rm phys}
			\;=\;
			\langle\Phi_1|\Phi_2\rangle_{\mathcal H_{\rm iso}},
			$$	
			since $\Pi_{\rm phys}$ is Hermitian and idempotent.
		\end{enumerate}
	\end{proposition}

	\begin{proof}
		\leavevmode
		\begin{enumerate}			
			\item Invariance.
			
			If $[O,U(g)]=0$ for all $g$, then $O\,\Pi_{\rm phys}=\Pi_{\rm phys}\,O$.
			Hence		
			$$
			O\bigl(\mathcal H_{\rm phys}\bigr)
			=O\,\Pi_{\rm phys}(\mathcal H_{\rm iso})
			=\Pi_{\rm phys}\,O(\mathcal H_{\rm iso})
			\subseteq\mathcal H_{\rm phys}.
			$$		
			The same argument applies when $O=O_{\rm ext}\otimes I_{\rm int}$, since such $O$ trivially commutes with all $U(g)^{\rm int}$.
			
			\item Direct‐sum \& superselection.
			
			In the CG layer one already has
			$\mathcal H_{\rm iso}=\bigoplus_Q\mathcal H_Q$, with $Q$ the total gauge‐charge label.
			On each $\mathcal H_Q$, $U(g)$ acts in the irreducible rep $V_Q$, so
			$\Pi_{\rm phys}\big|_{\mathcal H_Q}$ is either zero (if $V_Q$ contains no singlet) or the identity on the singlet subspace.
			Thus		
			$$
			\mathcal H_{\rm phys}
			\;=\;
			\bigoplus_{Q:\,V_Q\ni\mathbf1}\!
			\mathcal H_Q.
			$$
			
			Orthogonality of different $Q$-blocks is built in, and no gauge‐invariant operator can mix them.
			
			\item Inner product.
			
			For any $\Phi_i=\Pi_{\rm phys}\Psi_i$		
			$$
			\langle\Phi_1|\Phi_2\rangle
			=\langle \Psi_1|\Pi_{\rm phys}^2|\Psi_2\rangle
			=\langle \Psi_1|\Pi_{\rm phys}|\Psi_2\rangle
			=\langle\Phi_1|\Psi_2\rangle
			=\langle\Phi_1|\Phi_2\rangle_{\rm phys}.
			$$		
			This shows the physical inner product coincides with the ambient one.
		\end{enumerate}
	\end{proof}

	\subsubsection{Physical Packaged States inside $\mathcal H_{\rm phys}$}

	Each element $\lvert\Psi\rangle \in \mathcal H_{\rm phys}$ is called a \textbf{physical packaged state}.
	They have the following properties:

	\begin{theorem}[Properties of physical packaged states]
		\label{THM:PackagingOfPhysicalStates}
		Let 
		$$
		\mathcal H_{\rm phys}
		\;=\;
		\Pi_{\rm phys}\,\mathcal H_{\rm iso}
		\;=\;
		\bigoplus_{Q\in\sigma(Q)}\mathcal H_Q
		$$
		be the physical packaged subspace that decomposes into its mutually orthogonal charge sectors $\mathcal H_Q$.
		Let
		$$
		\lvert\Theta_n\rangle
		\;=\;
		\hat a_{n,1}^\dagger\,\hat a_{n,2}^\dagger\cdots\lvert0\rangle
		\;\in\;
		\mathcal H_Q
		$$
		be a basis of an individual charge sector $\mathcal H_Q$, and consider any superposition
		$$
		\lvert\Psi\rangle
		\;=\;
		\sum_n \alpha_n\,\lvert\Theta_n\rangle
		\;\in\;\mathcal H_Q
		\;\subset\;
		\mathcal H_{\rm phys}.
		$$
		Then:
		\begin{enumerate}
			\item[\bf(1)] {\bf Gauge-invariance.}
			Every vector in $\mathcal H_{\rm phys}$ is by construction a singlet:
			$$
			U(g)\,\lvert\Psi\rangle
			\;=\;
			\Pi_{\rm phys}\,U(g)\,\lvert\Psi\rangle
			\;=\;
			\Pi_{\rm phys}\,\lvert\Psi\rangle
			\;=\;
			\lvert\Psi\rangle,
			\quad\forall g\in G.
			$$
			
			\item[\bf(2)] {\bf Superselection.}
			Since $\mathcal H_{\rm phys}=\bigoplus_Q\mathcal H_Q$ and every gauge-invariant observable acts block-diagonally on this decomposition (Proposition~\ref{PROP:PackagingImpliesSS}), no coherent superposition between distinct $\mathcal H_Q$ can be realized.
			Thus $\lvert\Psi\rangle$ must lie entirely in one charge sector.
			
			\item[\bf(3)] {\bf Packaged entanglement of IQNs.}
			Within a single sector $\mathcal H_Q$, any non-factorizable superposition
			$\lvert\Psi\rangle$ that cannot be written as
			$\bigotimes_{i} \lvert\psi_i\rangle$
			is necessarily entangled in all of its IQNs.
			In other words, one cannot entangle only a subset of the IQNs without dragging the rest along.
		\end{enumerate}
	\end{theorem}

	\begin{proof}
		\ 
		\begin{enumerate}
			\item Gauge-invariance.
			
			By definition, $\mathcal H_{\rm phys}=\Pi_{\rm phys}\mathcal H_{\rm iso}$ and $\Pi_{\rm phys}=\int_G dg\,U(g)$ satisfies $U(g)\Pi_{\rm phys}=\Pi_{\rm phys}=\Pi_{\rm phys}U(g)$.
			Thus,
			$$
			U(g)\lvert\Psi\rangle
			=U(g)\,\Pi_{\rm phys}\lvert\Psi\rangle
			=\Pi_{\rm phys}\lvert\Psi\rangle
			=\lvert\Psi\rangle.
			$$
			
			\item Superselection.
			
			According to Proposition \ref{PROP:PropertiesPropertiesPhysicalPackagedSpace}, 
			$\mathcal H_{\rm phys}=\bigoplus_Q\mathcal H_Q$ with each $\mathcal H_Q$ invariant under all gauge-invariant operators.
			For $Q\neq Q'$, we have $\langle\psi_Q|\mathcal O|\psi_{Q'}\rangle=0$.
			Thus, no local gauge-invariant operation can produce a coherent superposition of different $Q$.
			
			\item Packaged entanglement.
			
			According to Theorem~\ref{THM:NoPartialFactorization}, single-particle creation operators carry irreducible IQN blocks.
			Any multi-particle state in $\mathcal H_Q$ is built by tensoring such blocks and then projecting to the singlet.
			If the resulting $\lvert\Psi\rangle$ fails to factorize across those blocks (that is, cannot be written as $\bigotimes_i\lvert\chi_i\rangle$), then each IQN block is necessarily entangled with the rest.
			No subset can remain separable.
			Therefore, the entanglement is packaged.
		\end{enumerate}
	\end{proof}

	\begin{example}[Examples of physical packaged states - gauge singlets]
		\leavevmode
		\begin{itemize}
			\item \emph{Meson ($q\bar q$) singlet in SU(3)$_{\rm color}$:}
			$$
			|\text{meson}\rangle
			= 
			\frac{1}{\sqrt{3}}
			\bigl(|r\bar r\rangle+|g\bar g\rangle+|b\bar b\rangle\bigr)
			\;\in\,\mathcal H_{\rm phys}.
			$$
			One checks $U_{g}\,|\text{meson}\rangle=|\text{meson}\rangle$ for all $g\in \mathrm{SU}(3)$.
			
			\item \emph{Baryon ($qqq$) singlet in SU(3):}
			$$
			|\epsilon\rangle 
			= \frac{1}{\sqrt{6}}\sum_{\pi\in S_3}
			\mathrm{sgn}(\pi)\,|i\,j\,k\rangle_{\pi(1)\pi(2)\pi(3)}
			\quad (i,j,k=r,g,b),
			$$
			the fully antisymmetric color wavefunction.
		\end{itemize}
	\end{example}

	\subsection{Residual Superselection in $\mathcal H_{\rm phys}$}
	\label{sec:residual-superselection}

	Although the gauge-averaging projector
	\[
	\Pi_{\rm phys} \;=\;\int_G\!dg\,U(g)
	\]
	eliminates all local gauge charges by restricting to the trivial $G$-irrep, it does not remove all superselection structure.
	In $\mathcal H_{\rm phys}=\Pi_{\rm phys}\,\mathcal H_{\rm iso}$ one still finds:
	
	\begin{enumerate}[label=\alph*)]
		\item \textbf{Residual global‐charge sectors.}  
		$\Pi_{\rm phys}$ averages over local gauge transformations, but leaves untouched the action of constant (global) $G$.
		Thus
		\[
		\mathcal H_{\rm phys}
		\;=\;
		\bigoplus_{Q\in\widehat G_{\rm global}}\;\mathcal H_{Q},
		\]
		where $Q$ labels irreps of the residual global symmetry.
		No gauge‐invariant operator can connect distinct $\mathcal H_{Q}$.
		
		\item \textbf{Other physical superselection rules.}  
		Quantum numbers such as fermion parity, discrete symmetries or topological fluxes remain superselected.
		For example, electric charge in QED or baryon number in QCD still label orthogonal sectors of $\mathcal H_{\rm phys}$.
		
		\item \textbf{Packaged entanglement superselection.}  
		Even within a fixed globalcharge sector, the irreducible IQN blocks enforce that entanglement respects the symmetry:
		no operation can coherently mix different IQN labels.
	\end{enumerate}

	\subsubsection{Global‐Charge Superselection: QED Example}
	\label{sec:QED-global-charge}
	
	\paragraph{Local Gauss’s law.}  
	In temporal gauge $A_0=0$ the operator form of Gauss’s law is
	\[
	\mathcal G(\mathbf x)
	= \nabla\!\cdot\!\mathbf E(\mathbf x)\;-\;\rho(\mathbf x)
	\;\approx\;0,
	\qquad
	\mathcal G(\mathbf x)\ket{\Psi_{\rm phys}}=0
	\quad\forall\,\mathbf x.
	\]
	This enforces that every local gauge irrep is the singlet.

	\paragraph{Global charge via surface flux.}  
	Integrating $\mathcal G(\mathbf x)$ over all space,
	\[
	Q
	= \int_{\mathbb R^3}\!d^3x\,\rho(\mathbf x)
	= \int_{S_\infty}\!d\mathbf S\!\cdot\!\mathbf E,
	\]
	one obtains the total (global) charge operator.
	Physical states obey $\mathcal G(\mathbf x)\ket\psi=0$ but may carry $Q\neq0$.
	Hence
	\[
	\mathcal H_{\rm phys}
	= \bigoplus_{Q\in\mathbb Z}\;\mathcal H_{Q}.
	\]

	\paragraph{Dressed‐electron illustration.}  
	A bare Dirac creation $\psi^\dagger(x)$ violates Gauss’s law.
	Dressing it by a Wilson‐line to infinity,
	\[
	\Psi^\dagger(x)
	= \exp\!\Bigl[i e\!\int_x^\infty \!dz^i\,A_i(z)\Bigr]\,
	\psi^\dagger(x),
	\]
	one finds
	$\mathcal G(\mathbf x)\ket{e^-}=0$ locally, but
	$Q\ket{e^-} = -e\ket{e^-}$.
	Thus $\ket{e^-}\in\mathcal H_{Q=-e}$.

	\subsubsection{Preservation of Superselection by $\Pi_{\rm phys}$}
	\label{sec:preserve-ss}
	
	\begin{proposition}[Gauge projection commutes with global charges]
		Let 
		$\mathcal H_{\rm iso}=\bigoplus_{Q}\mathcal K_Q$
		be the decomposition into global‐charge sectors, and let
		$\Pi_{\rm phys}=\int_Gdg\,U(g)$.
		Then
		\[
		[\Pi_{\rm phys},\,Q]=0
		\;\Longrightarrow\;
		\Pi_{\rm phys}\,\mathcal K_Q\subseteq\mathcal K_Q
		\quad\forall\,Q,
		\]
		so that
		$\mathcal H_{\rm phys}=\bigoplus_{Q}\bigl(\Pi_{\rm phys}\mathcal K_Q\bigr)$
		remains block‐diagonal and superselection sectors are preserved.
	\end{proposition}

	\begin{proof}
		Since $Q$ generates the residual global $G$, $[U(g),Q]=0$.  Hence
		\[
		[\Pi_{\rm phys},Q]
		=\int_Gdg\,[U(g),Q]=0,
		\]
		and each eigenspace $\mathcal K_Q$ is invariant under $\Pi_{\rm phys}$.  
	\end{proof}

	\subsubsection{Other Superselection Rules \& Packaged Entanglement}
	
	Beyond global charge, parity, fermion‐number or topological labels define further orthogonal sectors in $\mathcal H_{\rm phys}$.
	Within any such sector, the irreducible‐block structure enforces that entanglement always packages all IQNs together.

	\subsection{Packaging Survives Local Gauge-invariance Constraint}

	We now prove that the packet of IQNs remains indivisible after gauge projection.
	We need to show that the group-averaging projector
	$$
	\Pi_{\mathrm{phys}}
	=\int_{G}\!d\mu(g)\,U(g)
	$$
	filters out non-singlet blocks without ever undoing the
	irreducible‐package structure that was built in Stages 1-5.

	\begin{theorem}[Stability of packaging under local gauge-invariance constraint]
		For each $n\ge0$,
		\begin{equation}\label{EQ:StabilityPackagingUnderLocalGausslaw}
			\Pi_{\mathrm{phys}}
			\bigl(\mathcal H_{\rm iso}^{(n)}\bigr)
			\;=\;
			\begin{cases}
				V_{\mathbf 0}\;\otimes\;\mathcal H_{\text{ext}}^{(n)}, & 
				\text{if the singlet }{\mathbf 0}\text{ occurs in } 
				\bigl\{V_\lambda\bigr\},\\[2ex]
				\{0\}, & \text{otherwise},
			\end{cases}
		\end{equation}
		where $V_{\mathbf 0}\cong\mathbb C$ is the one-dimensional trivial representation.
		Thus, every surviving vector remains a tensor product	(internal singlet) $\otimes$ (external state) and no partial factorization of IQNs is introduced or removed by $\Pi_{\mathrm{phys}}$.		
		Hence packaging survives the local gauge-invariant projection.
	\end{theorem}

	\begin{proof}
		\leavevmode
		\begin{enumerate}
			\item Action of $\Pi_{\mathrm{phys}}$ on one $V_\lambda$.
			
			On the block $V_{\lambda}\otimes\mathcal H_{\text{ext}}^{(n)}$ we have
			
			\begin{equation}\label{EQ:GaugeProjectionOnOneBlock}
				U(g)
				=D^{(\lambda)}(g)\otimes\mathbf 1_{\rm ext},
				\qquad
				\Pi_{\mathrm{phys}}
				=\Bigl[\textstyle\int_G d\mu(g)\,D^{(\lambda)}(g)\Bigr]\!\otimes\!
				\mathbf 1_{\rm ext}.
			\end{equation}
						
			The Haar orthogonality relations for compact groups give		
			$$
			\int_G d\mu(g)\;D^{(\lambda)}_{ij}(g)
			=\begin{cases}
				\delta_{ij}, & \lambda=\mathbf 0,\\
				0,           & \lambda\neq\mathbf 0.
			\end{cases}
			$$
			If $\lambda\neq\mathbf 0$, then the integral vanishes and $\Pi_{\mathrm{phys}}$ annihilates the whole block.
			If $\lambda=\mathbf 0$ (trivial irrep, $\dim V_{\mathbf 0}=1$), then the integral is the identity on $V_{\mathbf 0}$.
			Thus, $\Pi_{\mathrm{phys}}$ acts as $(\mathbf 1\otimes\mathbf 1_{\rm ext})$, leaving that block untouched.			
			
			\item Assemble the blocks.
			
			Applying Eq. \eqref{EQ:GaugeProjectionOnOneBlock} and \eqref{EQ:IsotypicDecomposition} to the direct sum Eq. \eqref{EQ:PhysicalPackagedSpaceDecomposition} yields exactly the
			right-hand side of Eq. \eqref{EQ:StabilityPackagingUnderLocalGausslaw}.
			Orthogonality of different $\lambda$ makes the sum direct, so $\Pi_{\mathrm{phys}}$ never mixes distinct packaged blocks.
			
			\item Persistence of the packaging property.
			
			The internal factor is either $0$ (state discarded) or the $\mathbf 1$-dimensional singlet.			
			Both are irreducible.
			Tensoring with the external spectator space $\mathcal H_{\text{ext}}^{(n)}$ therefore still forbids any partial factorization of the IQNs.
			The external labels remain completely free, as expected.
		\end{enumerate}
	\end{proof}

	\begin{remark}
		Physical Interpretation:
		
		\begin{itemize}
			\item Gauge diagnosis:
			The projector simply throws away every irreducible block carrying a nontrivial gauge charge (color octet, U(1)-charged, etc.).
			No mixing or decoherence between packaged blocks is induced.
			
			\item Surviving package:
			The internal part shrinks to the 1-dimensional singlet but remains an irreducible package.
			Hence the IQN packaging is completely compatible with local gauge invariance.
			
			\item Where information lives:
			Logical qudits reside in $\mathcal H_{\text{ext}}^{(n)}$.
			All internal IQNs are frozen.
			They serve only to ensure gauge invariance.
		\end{itemize}
	\end{remark}

	\section{Conditions for Physical Packaged Superposition, and Measurements}	
	\label{SEC:ConditionsForPhysicalPackagedSuperposition}

	In Stage 5, we sketched the requirements for packaged superposition.
	In this section, we formalize those conditions.
	We further analyze measurement effects and uncover new collapse phenomena in hybrid-packaged states.

	\subsection{Superposition in a Physical Subspace}

	We first show a theorem that exhibits how packaged superposition behaviors in a physical subspace:

	\begin{theorem}[Superposition in the packaged physical sector]
		Let $G$ be a compact gauge group acting unitarily on a Hilbert space $\mathcal H$.
		Denote by $\hat Q$ the generator of the global charge and by $\{\hat G_x\}$ the Gauss generators at each spatial point (lattice site, vertex, …).
		Let $\ket{\Psi_1},\ket{\Psi_2}\in\mathcal H$ be packaged $n$-particle states.
		The normalized superposition
		$$
		\ket{\Psi}
		=\frac{\alpha\ket{\Psi_1}+\beta\ket{\Psi_2}}
		{\sqrt{|\alpha|^{2}+|\beta|^{2}+2\Re(\alpha^{*}\beta\braket{\Psi_1|\Psi_2})}}
		$$
		is physical (i.e. belongs to the packaged physical subspace $\mathcal H_{\mathrm{phys}}$) and non-trivial iff
		\begin{description}
			\item [S1.] Fixed total charge:
			$\hat Q\ket{\Psi_k}=Q\ket{\Psi_k}$ with the same eigenvalue $Q$ for $k=1,2$.
			
			\item [S2.] Gauge invariance (Gauss law):
			$U_g^{(x)}\ket{\Psi_k}=\ket{\Psi_k}$ (equivalently $\hat G_x\ket{\Psi_k}=0$) for all $x$ and $k$.
			
			\item [S3.] Identical gauge character:
			For every local transformation $U_g^{(i)}$ on a single particle,
			$$
			U_g^{(i)}\ket{\Psi_k}=\chi(g)\ket{\Psi_k},\qquad k=1,2,
			$$
			with the same one-dimensional character $\chi:G\!\to\!U(1)$ (for an $n$-particle package the overall phase is $\chi(g)^n$).
			
			\item [S4.] Non-triviality:
			$|\braket{\Psi_1|\Psi_2}| < 1$.
		\end{description}
		Under these conditions the superposition transforms by the overall phase $\chi(g)$ and thus remains in $\mathcal H_{\mathrm{phys}}$.
	\end{theorem}

	\begin{proof}
		We prove necessity and then sufficiency.
		
		\begin{enumerate}
			\item \textbf{Necessity.}			
			Assume $\ket{\Psi}$ is physical and non-trivial $\Longrightarrow$ S1, S2, S3, and S4.
			
			\begin{itemize}
				\item Total charge.
				Since $\hat Q$ generates global gauge transformations $U_g=\mathrm{e}^{\mathrm i\theta(g)\hat Q}$, we have	
				$$
				\hat Q\ket{\Psi}=Q\ket{\Psi}\quad\text{for some eigenvalue }Q.
				$$				
				Expand:	
				$$
				\hat Q\ket{\Psi}
				=\alpha\,\hat Q\ket{\Psi_1}+\beta\,\hat Q\ket{\Psi_2}.
				$$				
				For this to equal $Q\ket{\Psi}$, we must have
				$\hat Q\ket{\Psi_1}=Q\ket{\Psi_1}$ and $\hat Q\ket{\Psi_2}=Q\ket{\Psi_2}$.
				Hence S1 is necessary.
				
				\item Gauge invariance (Gauss law).
				Physicality of $\ket{\Psi}$ implies $\hat G_x\ket{\Psi}=0$ (or $U_g^{(x)}\ket{\Psi}=\ket{\Psi}$).
				Using linearity, $\hat G_x\ket{\Psi}=0$ forces $\hat G_x\ket{\Psi_k}=0$ for each $k$ separately.
				Otherwise the two terms could not cancel for arbitrary coefficients $\alpha,\beta$.
				Thus S2 is necessary.
				
				\item Identical character.
				Let $U_g^{(i)}$ act on the $i$-th particle only.
				Physicality demands	
				$$
				U_g^{(i)}\ket{\Psi}=\chi(g)\ket{\Psi}.
				$$				
				But	
				$$
				U_g^{(i)}\ket{\Psi}
				=\alpha\,U_g^{(i)}\ket{\Psi_1}+\beta\,U_g^{(i)}\ket{\Psi_2}
				=\alpha\,\chi_1(g)\ket{\Psi_1}+\beta\,\chi_2(g)\ket{\Psi_2},
				$$	
				where $\chi_k$ is the character associated with $\ket{\Psi_k}$.
				Comparing with $\chi(g)\ket{\Psi}=\alpha\chi(g)\ket{\Psi_1}+\beta\chi(g)\ket{\Psi_2}$ and using linear independence of $\ket{\Psi_1},\ket{\Psi_2}$ (non-triviality), we require $\chi_1=\chi_2=\chi$.
				Hence S3 is necessary.
				
				\item Non-triviality.
				If $|\braket{\Psi_1|\Psi_2}|=1$, then the two states differ only by an irrelevant phase, and the superposition reduces to a single physical state (trivial).
				So S4 is necessary.
			\end{itemize}

			\item \textbf{Sufficiency.}			
			Assume S1-S4 hold.
			We show that $\ket{\Psi}$ is physical and transforms by $\chi$.
			\begin{itemize}
				\item Well-defined normalization.
				Since $|\braket{\Psi_1|\Psi_2}| < 1$, the denominator in the definition of $\ket{\Psi}$ is finite and non-zero.
				So $\ket{\Psi}$ is a genuine state.
				
				\item Total charge.
				Using S1,	
				$$
				\hat Q\ket{\Psi}
				=\frac{\alpha Q\ket{\Psi_1}+\beta Q\ket{\Psi_2}}
				{\mathcal N}
				=Q\ket{\Psi},
				$$	
				where $\mathcal N$ is the normalization factor.
				Thus $\ket{\Psi}$ sits in the same charge sector $Q$.
				
				\item Gauge invariance (Gauss law).
				From S2 each term is annihilated:	
				$$
				\hat G_x\ket{\Psi}
				=\frac{\alpha\,0 + \beta\,0}{\mathcal N}=0.
				$$	
				Hence $\ket{\Psi}\in\ker\hat G_x$ for all $x$.
				
				\item Transformation of individual particles.
				Apply $U_g^{(i)}$ and use S3:	
				$$
				U_g^{(i)}\ket{\Psi}
				=\frac{\alpha\chi(g)\ket{\Psi_1}+\beta\chi(g)\ket{\Psi_2}}{\mathcal N}
				=\chi(g)\ket{\Psi}.
				$$
				Thus, every local gauge action multiplies the state by the same phase factor $\chi(g)$.
				Since $\chi$ is one-dimensional, global gauge transformations act by $\chi(g)^{n}$ on the entire $n$-particle package.
				This remains a harmless overall phase.
			\end{itemize}			
			Consequently $\ket{\Psi}$ satisfies all the physicality conditions that define $\mathcal H_{\mathrm{phys}}$.
		\end{enumerate}
	\end{proof}

	\subsection{Superposition of Single-particle Packaged States}
	\label{SEC:SuperpositionOfSingleParticlePackagedStates}

	We first recall why local gauge charges obey superselection while global quantum numbers do not.

	\subsubsection{Gauged Charge and Global Quantum Numbers}

	If we only consider the spin of a single-particle, then we can freely superpose single‐particle spin states $\lvert \uparrow\rangle$ and $\lvert \downarrow\rangle$ to obtain a superposition state like $\psi = \alpha\,\lvert \uparrow\rangle + \beta\,\lvert \downarrow\rangle$.
	This is because spin is an external degree of freedom (DOF) and is not subject to superselection.

	However, if we consider a gauge quantum number such as electric charge $\pm e$, then we cannot superpose states $\lvert +e\rangle$ and $\lvert -e\rangle$ to form a coherent state like $\psi = \alpha\,\lvert +e\rangle + \beta\,\lvert -e\rangle$.
	This is because states with different $\pm e$ charges belong to distinct charge sectors $\mathcal{H}_{+e}$ and $\mathcal{H}_{-e}$ and superselection rules forbid the formation of superpositions that span distinct charge sectors under local gauge group $U(1)$.

	Superselection therefore blocks charged single particles from serving as qubits or as messengers in quantum communication.
	To solve this problem, let us now compare the difference between a gauged and global IQN:
	\begin{enumerate}
		\item Gauge Charge (e.g., electric charge $Q$, color in QCD):
		Local gauge invariance and superselection protects gauge charge.
		The cross‐sector superpositions are disallowed.
				
		\item Global Quantum Number (e.g., strangeness, flavor, isospin): 
		Local gauge symmetries do not enforce global quantum number \cite{Aharonov1967}.
		Therefore, superselection rules do not apply to global quantum number.
		The coherent superpositions are allowed if the net gauge charge is still the same.
	\end{enumerate}

	This distinction between gauged IQNs (which are fixed by local gauge invariance) and global quantum numbers (which can be superposed) is important.
	It uncovers the fact that only those difference arises from an internal global quantum number (e.g. flavor SU(3)$_F$), that is, a conserved Noether charge \cite{Noether1918} without a local gauge constraint, can be used to encode quantum information.

	\begin{example}[Superposition of neutral mesons]
		Neutral mesons can indeed form superpositions like 
		\[
		\alpha\,\lvert K^0\rangle \;+\; \beta\,\lvert \overline{K}^0\rangle,
		\]
		where $\lvert K^0\rangle$ and $\lvert \overline{K}^0\rangle$ are flavor eigenstates but both carry zero net electric charge.
	\end{example}

	\subsubsection{Conditions for Packaged Superposition of a Single Particle}
	\label{SEC:ConditionsForPackagedSuperpositionOfASingleParticle}

	Now we see that the obstacle for superposition of single‐particle packaged states comes from gauged charge.
	More specifically, if a particle $\lvert P\rangle$ and its antiparticle $\lvert \bar{P}\rangle$ differ by a nonzero gauge charge, then superselection forbids the superposition between state $\lvert P\rangle$ and $\lvert \bar{P}\rangle$.

	However, if we select a particle $\lvert P\rangle$ and its antiparticle $\lvert \bar{P}\rangle$ that reside in the same net‐charge sector (often $Q=0$), then they differ only by a global quantum number (e.g., flavor, isospin) instead of a gauged one.
	In this case, we can successfully bypass the obstacles of superselection and form physically coherent superposition states.
	Therefore, we can choose neutral particles \cite{GellMannPais1955,Good1961} as quantum information carriers.
	This is because all relevant states of a neutral particle (single or composite) always remain in one gauge sector and differ only in global quantum numbers, which permit superposition.

	For a single-particle in QED or QCD, the Gauss-law (or local gauge-invariance) constraint is precisely	
	$$
	\hat G_x\;\ket{\Psi}\;=\;0
	\quad\Longleftrightarrow\quad
	\hat Q_{\text{tot}}\ket{\Psi}=0
	\quad(\text{single site/continuum limit}),
	$$
	that is, the total charge is zero.

	For rigor, let us now summarize the conditions for forming a valid packaged superposition states of single-particle as follows:

	\begin{corollary}[Necessary and sufficient conditions for packaged superposition of a single particle]\label{COR:NecessarySufficientConditionsSingleParticle}
		Let $\lvert P \rangle$ be the packaged state of a particle and $\lvert \bar{P} \rangle$ be the packaged state of its antiparticle.
		The basis of single-particle space can be written as $\{\lvert P \rangle, \lvert \bar{P} \rangle\}$.
		Then the packaged superposition state of a single particle
		\begin{equation}\label{EQ:ValidPackagedSuperpositionStatesOfSingleParticle}
			\Psi = \alpha\,\lvert P\rangle + \beta\,\lvert \bar{P}\rangle, \quad |\alpha|^2 + |\beta|^2 = 1,
		\end{equation}
		is valid, i.e. physically allowed and nontrivial, iff the following two conditions hold:
		
		\begin{description}
			\item[C1.] \textit{Zero Net Gauge Charge:} 
			The particle $\lvert P\rangle$ and its antiparticle $\lvert \bar{P}\rangle$ both carry zero gauge charge.
			This ensures they occupy the same superselection sector $\mathcal{H}_{Q=0}$, so $\lvert P\rangle$ and $\lvert \bar{P}\rangle$ are not separated by a gauge-induced superselection boundary.
			Let $\hat{Q}$ be the net-charge operator, we have
			\begin{equation}\label{EQ:ZeroNetGaugeChargeCond}
				\hat{Q}\lvert P\rangle = 0,\quad \hat{Q}\lvert \bar{P}\rangle = 0.
			\end{equation}
			
			\item[C2.] \textit{Difference is a Global Quantum Number:}
			The particle and its antiparticle differ only by a global quantum number (flavor, strangeness, lepton number, etc.), which we denote by the operator $\hat{F}$.
			Then, the basis states must satisfy
			\begin{equation}\label{EQ:DifferenceGlobalQuantumNumber}
				\hat{F}\lvert P\rangle = f\,\lvert P\rangle,\quad \hat{F}\lvert \bar{P}\rangle = -f\,\lvert \bar{P}\rangle,
			\end{equation}
			where $f$ is a nonzero real number.
			Since $\hat{F}$ is not protected by a local gauge symmetry, coherent superpositions of the form in Eq.~(\ref{EQ:ValidPackagedSuperpositionStatesOfSingleParticle}) are physically allowed.
			The difference in a global quantum number ensures that the two states carry opposite values of the global quantum number and therefore can serve as distinct logical states for encoding quantum information.
		\end{description}
		Thus, to construct physically allowed superposition states for quantum computation, one must choose particles (or composite systems) that are neutral (or that differ only by global quantum numbers).
	\end{corollary}

	\begin{proof}
		We now show that the two conditions are necessary and sufficient for a single‐particle packaged superposition state to be physically allowed.
		
		\smallskip\noindent
		\emph{I. Necessity.}
		Suppose Eq.~(\ref{EQ:ValidPackagedSuperpositionStatesOfSingleParticle}) is valid (physically allowed), then we have
		\begin{itemize}
			\item \textit{No Net Gauge Charge:}
			The superselection rules prevent coherent superpositions of states that carry different gauge charges.
			If either $\lvert P\rangle$ or $\lvert\bar{P}\rangle$ carried a nonzero gauge charge, then they would lie in different superselection sectors and could not be coherently superposed.
			Thus, requiring 
			\[
			\hat{Q}\lvert P\rangle = 0 \quad \text{and} \quad \hat{Q}\lvert \bar{P}\rangle = 0
			\]
			is necessary.
			
			\item \textit{Difference is Global, Not Gauged:} 
			If the states differed by a gauged quantum number (e.g., electric charge $Q$), then (even if they were both overall neutral) their internal structures would belong to different irreps of the local gauge group.
			Coherent superpositions between different irreps are forbidden by superselection.
			So the difference must be in a quantity (like flavor or strangeness) that is, a global quantum number $f$, which is not subject to the superselection constraint.
			Eq. \eqref{EQ:DifferenceGlobalQuantumNumber} also guarantees orthogonality of $\lvert P\rangle$ and $\lvert \bar{P}\rangle$ when $f \neq 0$ \cite{WWW1952}.                                   
		\end{itemize}
			
		\smallskip\noindent
		\emph{II. Sufficient:}
		Conversely, suppose the two conditions hold. Then
		\begin{itemize}
			\item \textit{No Net Gauge Charge:}
			If both states have zero gauge charge, then they lie in the same superselection sector (more precisely, because both are annihilated by all local Gauss operators $\hat G_x$, they lie in the same $U(1)$ irrep, the trivial one).
			From the perspective of gauge-invariance, a coherent superposition is allowed.
			
			\item \textit{Difference is Global, Not Gauged:} 
			If the only distinction between the particle and its antiparticle is a global quantum number (that is, the operator $\hat{F}$ acting nontrivially while $\hat{Q} = 0$ holds), then the states lie in the same superselection sector and any superposition 
			\[
			\Psi = \alpha\,\lvert P\rangle + \beta\,\lvert \bar{P}\rangle
			\]
			is allowed by the rules of quantum mechanics.
		\end{itemize}
		
		In other words, Eq.~(\ref{EQ:ValidPackagedSuperpositionStatesOfSingleParticle}) is valid.
		
		Thus, the two conditions are necessary and sufficient.
	\end{proof}

	\begin{example}[Pass/Fail at a glance]
		\leavevmode
		\begin{itemize}
			\item \emph{Pass (neutral mesons)}:
			$\alpha\ket{K^0}+\beta\ket{\bar K^0}$ — fixed $Q=0$, differ by flavor (global) (satisfies both C1 and C2).
			
			\item \emph{Fail (charged single particles)}:
			$\alpha\ket{e^-}+\beta\ket{e^+}$ — different $U(1)$ sectors (violates C1).
			
			\item \emph{Fail (self-conjugate)}:
			$\alpha\ket{\gamma}+\beta\ket{\bar\gamma}$ — a photon is its own antiparticles $\gamma=\bar\gamma$ (violates non-triviality C2).
		\end{itemize}
	\end{example}

	\begin{remark}
		Remarks on the necessary and sufficient conditions:
		
		\begin{enumerate}
			\item The first condition $\hat{Q}\lvert P\rangle = 0,~ \hat{Q}\lvert \bar{P}\rangle = 0$ is to guarantee that the superposition is physical allowed. Only neutral particles and their antiparticles can be in the same charge sector (superselection sector) and therefore the superposition $\Psi = \alpha\,\lvert P\rangle + \beta\,\lvert \bar{P}\rangle$ is permitted.
			
			\item The second condition $\hat{F}\lvert P\rangle = f\,\lvert P\rangle,\quad \hat{F}\lvert \bar{P}\rangle = -f\,\lvert \bar{P}\rangle$ is to guarantee that the superposition is non-trivial. Otherwise, we have $\lvert P \rangle = \lvert \bar{P} \rangle$ and the superposition $\Psi = \alpha\,\lvert P\rangle + \beta\,\lvert \bar{P}\rangle$ is trivial. For example, photons are their own antiparticle, $\lvert \gamma\rangle = \lvert \bar{\gamma}\rangle$. We cannot use photons for packaged qubits.
		\end{enumerate}	
	\end{remark}

	\begin{corollary}\label{COR:SuperpositionStateGaugeInvariant}
		The superposition state in Eq.~(\ref{EQ:ValidPackagedSuperpositionStatesOfSingleParticle}) that satisfies the conditions C1 and C2 is gauge-invariant.
	\end{corollary}

	\begin{proof}
		By definition, the particle state $\lvert P\rangle$ and antiparticle state $\lvert\bar{P}\rangle$ are packaged states.
		Therefore, they are covariant.	
		Let $G$ be a gauge group. Under any local gauge transformation $U(g) ~ (g \in G)$, the single-particle packaged states transform as
		\[
		U(g)\,\lvert P\rangle = e^{i\phi(g)}\,\lvert P\rangle,\quad U(g)\,\lvert \bar{P}\rangle = e^{i\phi(g)}\,\lvert \bar{P}\rangle,
		\]
		where the phase $e^{i\phi(g)}$ is the same for both states because they lie in the same sector.
		
		Thus, the superposition state in Eq.~(\ref{EQ:ValidPackagedSuperpositionStatesOfSingleParticle}),
		\[
		\lvert \Psi\rangle = \alpha\,\lvert P\rangle + \beta\,\lvert \bar{P}\rangle,
		\]
		transforms as
		\[
		U(g)\,\lvert \Psi\rangle 
		= \alpha\,U(g)\,\lvert P\rangle + \beta\,U(g)\,\lvert \bar{P}\rangle
		= e^{i\phi(g)} \left(\alpha\,\lvert P\rangle + \beta\,\lvert \bar{P}\rangle\right)
		= e^{i\phi(g)}\,\lvert \Psi\rangle.
		\]
		This shows that $\lvert \Psi\rangle$ is gauge-invariant and remains in the same gauge sector (that is, $\hat{Q}=0$).	 
	\end{proof}

	\begin{remark}[Gauss law and single blocks]
		The irreducible block derived above transforms covariantly under $G$.
		Whether it is physical can be verified by the local gauge-invariance constraint:
		\begin{itemize}
			\item Abelian groups (QED):
			a dressed charged excitation satisfies local gauge-invariance constraint.
			
			\item Non-Abelian color $SU(3)$:
			an isolated $\mathbf 3$ or $\mathbf 8$ block is projected out by $\Pi_{\mathrm{phys}}$ and is therefore not gauge-invariant.
			Only color singlet combinations survive and is gauge-invariant.
		\end{itemize}
	\end{remark}

	\subsection{Superposition of Multi-particle Packaged States}
	\label{SEC:SuperpositionOfMultiParticlePackagedStates}

	In Sec. \ref{SEC:ConditionsForPackagedSuperpositionOfASingleParticle}, we discussed the necessary and sufficient conditions for a single‑particle packaged superposition.
	Here we do the same thing to multi‑particle.

	To obtain a better understanding on the superposition of multi-particle packaged states, we go back to the multi-particle packaged subspace in Sec. \ref{SEC:MultiParticlePackagedSubspaces}, where each
	state (vector) is a multi-particle packaged state.
	The only remaining question is when superpositions of different such n-particle packaged states are physical? We answer this question by building the next proposition.

	\begin{corollary}[Necessary and sufficient conditions for packaged superposition of multi-particle]
		\label{PROP:NecessarySufficientConditionsMultiParticle}
		Let	$G$ be a compact gauge group that acts on the physical Hilbert space $\mathcal H$ with a global conserved charge operator $\hat Q$ and a family of local Gauss-law generators $\{\hat G_x\}_{x}$.	
		Let $\ket{\Psi_{1}},\ket{\Psi_{2}}\in\mathcal H$ be two $n$-particle packaged states.
		Consider the normalized superposition
		\begin{equation}\label{EQ:ValidPackagedSuperpositionStatesOfMultieParticle}
			\ket{\Psi}\;=\;\alpha\,\ket{\Psi_{1}}\;+\;\beta\,\ket{\Psi_{2}},
			\quad |\alpha|^{2}+|\beta|^{2}=1 .
		\end{equation}
		The state $\ket{\Psi}$ is valid (physically admissible and non‑trivial) if and only if the following conditions hold:
		\begin{description}
			\item[C1.] Fixed total charge:	
			$$
			\hat Q\,\ket{\Psi_{k}} \;=\; Q\,\ket{\Psi_{k}}, 
			\quad k=1,2 ,
			$$
			for one and the same eigenvalue $Q$.
			
			\item[C2.] Gauss law at every site:		
			$$
			\hat G_{x}\,\ket{\Psi_{k}} \;=\; 0,
			\quad \forall x,\;k=1,2 .
			$$
			
			\item[C3.] Identical local gauge character:
			both states transform in the same irrep under each local gauge action $U(g)^{(i)}$:		
			$$
			U_{g}^{(i)}\ket{\Psi_{1}} = \chi(g)\,\ket{\Psi_{1}}, 
			\quad
			U_{g}^{(i)}\ket{\Psi_{2}} = \chi(g)\,\ket{\Psi_{2}},
			\quad \forall g\in G,\;\forall i.
			$$
			$\chi$ may be the trivial character.
			For groups with center only, this reduces to singlet behavior.

			\item[C4.] Linear independence:
			$\ket{\Psi_1}$ and $\ket{\Psi_2}$ are not proportional:
			$$
			\bigl|\langle\Psi_1 \mid \Psi_2\rangle\bigr|<1.
			$$
		\end{description}
	\end{corollary}

	\begin{proof}
		C1, C2, and C3 guarantee physical admissibility, that is, not forbidden by any superselection rule.
		C4 guarantees non‑triviality.
		
		\smallskip\noindent
		\emph{I. Necessity.}
		Assume $\ket{\Psi}$ is valid $\Longrightarrow$ C1, C2, C3, and C4.
		
		\begin{enumerate}
			\item Physically admissible
			
			Since $\ket{\Psi}$ is physical, it must lie in a single superselection sector.
			That forces each component $\ket{\Psi_k}$ to share the same eigenvalue of $\hat Q$ and to satisfy $\hat G_x=0$.
			Otherwise the superposition would lie across orthogonal sectors and be forbidden.
			Thus, we obtain C1 and C2.
			
			Moreover gauge invariance of the superposition requires that all local gauge transformations $U_{g}^{(i)}$ act by one and the same phase on $\ket{\Psi_1}$ and $\ket{\Psi_2}$.
			Otherwise, $U_{g}^{(i)}\ket{\Psi}$ would contain relative phases depending on $g$ and the superposition would not transform by a single overall phase.
			This will violate superselection.
			Thus, we must have C3.
						
			\item Non‑triviality.
			Linear independence isn’t needed for ``physical admissibility'', only for ``non-triviality''.
			Specifically, if $\ket{\Psi_{1}}\propto\ket{\Psi_{2}}$ the superposition collapses to a single vector and carries no logical degree of freedom, contradicting the premise.
		\end{enumerate}	
		
		\smallskip\noindent
		\emph{II. Sufficiency.}
		
		Conversely, assume C1, C2, C3, and C4 $\Longrightarrow$ $\ket{\Psi}$ is physically admissible and non‑trivial.
		\begin{enumerate}
			\item Physically admissible.
			
			By C1 and C2, any linear combination of two common eigenvectors of $\hat Q$ (with eigenvalue $Q$) and of all $\hat G_x$ (with eigenvalue $0$) is again an eigenvector with the same set of eigenvalues.
			Thus, $\ket{\Psi}$ lies in the same superselection block as $\ket{\Psi_{1,2}}$.
			
			By C3, every local $U_{g}^{(i)}$ acts as $\chi(g)\mathbf 1$ on each component, then it also acts as $\chi(g)\mathbf 1$ on their superposition:	
			$$
			U_{g}^{(i)}\ket{\Psi}
			= \chi(g)\,\ket{\Psi},
			\qquad \forall g,i .
			$$
			Thus $\ket{\Psi}$ is gauge‑invariant up to a global phase.
			
			Put these two together, we conclude that $\ket{\Psi}$ is physically admissible (not forbidden by any superselection rule).
			
			\item Non‑triviality.
			By C4, the two basis states are linearly independent, so $\ket{\Psi}$ genuinely spans a two‑dimensional logical subspace when $(\alpha,\beta)$ vary.
		\end{enumerate}
		
		Hence $\ket{\Psi}$ is a valid multi-particle packaged superposition state.
	\end{proof}

	Thus, any multi-particle superposition complying with C1-C4 remains inside a single gauge sector and supports non-trivial logic.
	We illustrate with two common cases.

	\begin{example}[GHZ-style entangled state]
		The tiny GHZ-style entangled state $$
		\frac{1}{\sqrt2}
		\bigl(\ket{e^+e^-}_1\ket{\mu^+\mu^-}_2 + \ket{\mu^+\mu^-}_1\ket{e^+e^-}_2\bigr)
		$$
		meets C1-C3, and is non-trivial by C4.
		Thus, it is physically allowed.
	\end{example}

	\begin{example}
		\leavevmode
		\begin{enumerate}
			\item \textbf{$K^0\bar K^0$-pair Bell states.}
			Each neutral kaon is individually a packaged single‑particle state.
			The two‑particle Bell state
			$
			\tfrac1{\sqrt2}\bigl(\ket{K^0}\ket{\bar K^0}
			\pm\ket{\bar K^0}\ket{K^0}\bigr)
			$
			satisfies total electric charge $0$ (C1), Gauss-law $\hat G_{x}\,\ket{\Psi_{k}} = 0$ (C2), identical gauge character (C3), and are linearly independent (orthonormal).
			Thus, the Bell superposition is valid.
			
			\item \textbf{Electron-positron pair vs.\ muon-antimuon pair.}
			The states
			$
			\ket{e^{+}e^{-}}
			$
			and
			$
			\ket{\mu^{+}\mu^{-}}
			$
			carry the same gauge charge profile (each is a QED singlet) but differ by global flavor number.
			Coherent superpositions are allowed in principle, although unstable in practice because of electroweak interactions \cite{Landau1957}.
			
			\item \textbf{Photon pairs.}
			Each photon is its own antiparticle, so condition C4 fails
			(no global observable distinguishes the two basis vectors).
			The corresponding superposition is therefore trivial and cannot offer
			logical degree of freedom for packaged encoding.
		\end{enumerate}
	\end{example}

	\begin{remark}[Relation to single‑particle case]
		For $n=1$ the total‑charge operator reduces to $\hat Q$ and the
		Gauss‑law constraints act trivially, so
		Proposition \ref{PROP:NecessarySufficientConditionsMultiParticle}
		reduces to
		Proposition \ref{COR:NecessarySufficientConditionsSingleParticle}.
	\end{remark}

	\begin{corollary}[Gauge‑invariance]
		\label{COR:MultiSuperpositionGaugeInvariant}
		Every state $\ket{\Phi}$ satisfying the conditions of
		Proposition \ref{PROP:NecessarySufficientConditionsMultiParticle} obeys
		$
		U(g)^{\otimes n}\ket{\Phi}=e^{i\phi(g)}\ket{\Phi},
		\; \forall g\in G,
		$
		and thus remains inside the packaged sector.
	\end{corollary}

	\subsection{Measurements in the Physical Packaged Space}
	\label{SEC:MeasurementsInHphys}
	
	Physical measurements correspond to gauge‑invariant observables acting on the physical Hilbert space
	$\mathcal H_{\rm phys} = \Pi_{\rm phys},\mathcal H_{\rm iso}$.
	Any measurement is described by a positive operator‑valued measure (POVM)
	$E_m$ with	
	$$
	E_m \succeq 0,
	\quad
	\sum_m E_m = \mathbb I_{\rm phys},
	\quad
	[E_m,\,U(g)]=0\quad\forall\,g\in G,
	$$	
	where each outcome $m$ occurs with probability	
	$$
	p_m = \langle\Psi|E_m|\Psi\rangle,
	$$	
	and leaves the system in the post‑measurement state
	$$
	|\Psi_m \rangle = \frac{E_m|\Psi\rangle}{\sqrt{p_m}} \in \mathcal H_{\rm phys}.
	$$
	Projective measurements are the special case $E_m=P_m$ with $P_mP_{m'}=\delta_{mm'}P_m$.
	The condition $[E_m,U(g)]=0$ is both necessary and sufficient for an operator to preserve $\mathcal H_{\mathrm{phys}}$.

	\subsubsection{External vs. Internal Measurements}
	
	Gauge invariance constrains any physical measurement to commute with the gauge action,
	that is
	$[M, U(g)]=0$.
	In particular:

	\begin{theorem}[Gauge‐commuting POVMs]
		Let
		$$
		\mathcal H_{\rm phys}
		=\Pi_{\rm phys}\,\mathcal H
		\;=\;\bigl\{\ket\Phi\in\mathcal H:\;U(g)\ket\Phi=\ket\Phi\;\;\forall\,g\in G\}
		$$
		be the gauge‐invariant (physical) subspace of a bipartite Hilbert space
		$\mathcal H=\mathcal H_{\rm int}\otimes\mathcal H_{\rm ext}$,
		where $G$ is a compact gauge group represented unitarily by
		$
		U(g)\;=\;U_{\rm int}(g)\;\otimes\;U_{\rm ext}(g)\,.
		$
		Write an arbitrary physical vector in its Schmidt form,
		$$
		\ket\Psi
		=\sum_{k=1}^r\lambda_k\,
		\underbrace{\ket{\chi_k}}_{\mathcal H_{\rm int}}
		\otimes
		\underbrace{\ket{e_k}}_{\mathcal H_{\rm ext}}
		\,,\qquad
		\lambda_k>0\,,\;\;\braket{\chi_k|\chi_\ell}=\braket{e_k|e_\ell}=\delta_{k\ell}.
		$$	
		Let $\{M_\alpha\}$ be a POVM on $\mathcal H$ subordinate to the factorization
		$\mathcal H_{\rm int}\otimes\mathcal H_{\rm ext}$.
		Suppose each outcome operator $M_\alpha$ commutes with the full gauge action,
		$
		[M_\alpha,\;U(g)]\;=\;0\,,\quad\forall\,g\in G.
		$
		Then each post‐measurement (unnormalized) state
		$\ket{\Psi_\alpha}=M_\alpha\ket\Psi$ also lies in $\mathcal H_{\rm phys}$, and its Schmidt decomposition is obtained simply by discarding those Schmidt‐terms which are annihilated by the local POVM factor.
		In particular:
		\begin{enumerate}
			\item External POVM.
			If
			$$
			M_m \;=\; \mathbb I_{\rm int}\;\otimes\;E_m
			\quad\text{with}\quad
			[E_m,\;U_{\rm ext}(g)]\;=\;0\quad\forall\,g,
			$$
			then
			$$
			\ket{\Psi_m}
			=M_m\ket\Psi
			=\sum_{k=1}^r \lambda_k\,
			\ket{\chi_k}\otimes\bigl(E_m\ket{e_k}\bigr).
			$$
			
			The probability of outcome $m$ is
			$\;p_m=\|\Psi_m\|^2=\sum_k\lambda_k^2\,\|E_m\ket{e_k}\|^2,$
			and only those $k$ with $E_m\ket{e_k}\neq0$ survive in the normalized post‐measurement state.  Moreover
			$\ket{\Psi_m}\in\mathcal H_{\rm phys}$.
			
			\item Internal POVM.
			If		
			$$
			M_n \;=\; E_n\;\otimes\;\mathbb I_{\rm ext}
			\quad\text{with}\quad
			[E_n,\;U_{\rm int}(g)]\;=\;0\quad\forall\,g,
			$$
			then
			$$
			\ket{\Psi_n}
			=M_n\ket\Psi
			=\sum_{k=1}^r \lambda_k\,
			\bigl(E_n\ket{\chi_k}\bigr)\otimes\ket{e_k},
			$$
			with $p_n=\sum_k\lambda_k^2\,\|E_n\ket{\chi_k}\|^2$, and again $\ket{\Psi_n}\in\mathcal H_{\rm phys}$.
		\end{enumerate}
	\end{theorem}

	\begin{proof}
		We break the argument into three Lemmas.
		
		\begin{enumerate}
			\item Gauge commutation $\Rightarrow$ preserves $\mathcal H_{\rm phys}$.

			Let $M$ be any bounded operator on $\mathcal H$ such that
			$$
			[M, U(g)] = 0,
			\quad
			\forall \, g \in G.
			$$			
			According to Definition \ref{EQ:CompactGaugeProjector}, we have
			$$
			M \Pi_{\rm phys} = \int_G dg M U(g) = \int_G dg U(g) M = \Pi_{\rm phys} M,
			$$
			so whenever $\ket\Phi\in\mathcal H_{\rm phys}$ (i.e.\ $\Pi_{\rm phys}\ket\Phi=\ket\Phi$), one has
			$$
			\Pi_{\rm phys}\,\bigl(M\ket\Phi\bigr)
			= M \left(\Pi_{\rm phys}\ket\Phi\right)
			= M \ket\Phi,
			$$
			and hence $M\ket\Phi\in\mathcal H_{\rm phys}$.

			\item Action of an external operator in Schmidt form:
			
			Starting from
			$$
			\ket\Psi
			=\sum_{k=1}^r\lambda_k\,
			\ket{\chi_k}\otimes\ket{e_k},
			$$
			apply $M_m=\mathbb I_{\rm int}\otimes E_m$.
			Linearity gives
			$$
			M_m\ket\Psi
			=\sum_k\lambda_k\,
			\bigl(\mathbb I_{\rm int}\ket{\chi_k}\bigr)
			\otimes
			\bigl(E_m\ket{e_k}\bigr)
			=\sum_{k=1}^r\lambda_k\,
			\ket{\chi_k}\otimes\bigl(E_m\ket{e_k}\bigr).
			$$		
			The norm squared is
			$$
			\bigl\|M_m\ket\Psi\bigr\|^2
			=\sum_{k,\ell}\lambda_k\lambda_\ell
			\underbrace{\braket{\chi_k|\chi_\ell}}_{\delta_{k\ell}}
			\underbrace{\bigl\langle e_k\bigm|E_m^\dagger E_m\bigm|e_\ell\bigr\rangle}_{\;=\delta_{k\ell}\,\|E_m e_k\|^2}
			=\sum_{k=1}^r\lambda_k^2\,\bigl\|E_m\ket{e_k}\bigr\|^2.
			$$		
			Define
			$$
			p_m=\|M_m\ket\Psi\|^2,
			\qquad
			\ket{\Psi_m}=\frac{1}{\sqrt{p_m}}\,M_m\ket\Psi.
			$$		
			Clearly any term with $E_m\ket{e_k}=0$ drops out of the sum, and the remainder is a perfectly good (normalized) Schmidt expansion of $\ket{\Psi_m}$.

			\item Gauge‐invariance of the collapsed state:
			
			Since $M_m$ commutes with the full $U(g)=U_{\rm int}(g)\otimes U_{\rm ext}(g)$, by Lemma 1 we know
			$\ket{\Psi_m}=M_m\ket\Psi\in\mathcal H_{\rm phys}$.
			Equivalently,
			$$
			U(g)\ket{\Psi_m}
			=\;U(g)\,M_m\ket\Psi
			=\;M_m\,U(g)\ket\Psi
			=\;M_m\ket\Psi
			=\ket{\Psi_m},
			$$
			so the post‐measurement state remains gauge‐invariant.
		\end{enumerate}	
	\end{proof}

	\subsubsection{Collapse of Hybrid‑Packaged Entanglement}

	We now prove a theorem on the measurement collapse in $\mathcal H_{\rm phys}$:
		
	\begin{theorem}[Collapse in $\mathcal H_{\rm phys}$]
		\label{THM:ExtMeasCollapse}
		Let		
		$
		\mathcal H_{\rm phys} = \mathcal H_{\rm int}\,\otimes\,\mathcal H_{\rm ext}
		$	
		be the physical (gauge-singlet) subspace and let		
		$$
		|\Psi\rangle
		=\sum_{k}\lambda_k\,|\chi_k\rangle_{\rm int}\,\otimes\,|e_k\rangle_{\rm ext}
		$$		
		be its Schmidt decomposition across the internal $\otimes$ external split, with $\langle\chi_k|\chi_{k'}\rangle=\langle e_k|e_{k'}\rangle=\delta_{kk'}$ and $\lambda_k 0$.
		Then:
		
		\begin{enumerate}
			\item [1.] \emph{External measurement.}
			Let
			$M_m^{\rm ext}=\mathbb I_{\rm int}\otimes E_m^{\rm ext}$ be any external POVM.
			For each outcome with $p_m 0$, the post‑measurement state			
			$$
			|\Psi_m\rangle=\frac{M_m^{\rm ext}|\Psi\rangle}{\sqrt{p_m}}
			$$
			
			satisfies:
			\begin{itemize}
				\item \emph{Gauge invariance preserved:}
				$|\Psi_m\rangle$ remains in the same net‑charge sector.
				
				\item \emph{Entanglement collapse:}
				Only those Schmidt components $|e_k\rangle$ with
				$E_m^{\rm ext}|e_k\rangle\neq0$ survive, but projects onto the $E_m$-support in the external factor so that off-diagonal coherences between different $|e_k\rangle$ vanish (killing all external coherence).
				$p_m>0$ iff at least one $k$ satisfies $\lVert E_m | e_k\rangle\rVert \neq 0$.
			\end{itemize}
			
			\item [2.] \emph{Internal measurement.}
			Let
			$M_n^{\rm int}=E_n^{\rm int}\otimes\mathbb I_{\rm ext}$ be any internal POVM.
			For each outcome with $p_n>0$, the post‑measurement state			
			$$
			|\Psi_n\rangle=\frac{M_n^{\rm int}|\Psi\rangle}{\sqrt{p_n}}
			$$			
			likewise preserves gauge invariance and collapses the Schmidt sum to those terms satisfying $E_n^{\rm int}|\chi_k\rangle\neq0$.
			In particular, no leakage into non-singlet sectors can occur.
		\end{enumerate}
	\end{theorem}

	\begin{proof}
		\
		\begin{enumerate}
			\item External measurement.
			
			For any POVM on the external subsystem		
			$$
			M_m^{\rm ext}
			= \mathbb I_{\rm int}\,\otimes\,E_m^{\rm ext},
			\quad
			\sum_m (M_m^{\rm ext})^\dagger M_m^{\rm ext} = \mathbb I_{\rm phys},
			\quad
			[\,M_m^{\rm ext},\,U(g)^{\rm int}\otimes I_{\rm ext}\,]=0,
			$$		
			outcome $m$ occurs with probability		
			$$
			p_m
			= \langle\Psi|
			(M_m^{\rm ext})^\dagger M_m^{\rm ext}
			|\Psi\rangle
			= \sum_k \lambda_k^2\,\bigl\lVert E_m^{\rm ext}\,|e_k\rangle\bigr\rVert^2,
			$$		
			and the post-measurement (normalized) state		
			$$
			|\Psi_m\rangle
			= \frac{M_m^{\rm ext}\,|\Psi\rangle}{\sqrt{p_m}}
			= \frac{1}{\sqrt{p_m}}
			\sum_k \lambda_k\;
			|\chi_k\rangle_{\rm int}\,\otimes\,
			E_m^{\rm ext}\,|e_k\rangle_{\rm ext}
			$$
			
			satisfies:
			\begin{itemize}
				\item Gauge invariance \& charge preservation.
				Since $[M_m^{\rm ext},\,\Pi_{\rm phys}]=0$ (they commute with every local gauge transformation), and $ |\Psi\rangle\in\mathcal H_{\rm phys}$ implies $\Pi_{\rm phys}|\Psi\rangle=|\Psi\rangle$, we have		
				$$
				\Pi_{\rm phys}\,|\Psi_m\rangle
				= \frac{\Pi_{\rm phys}\,M_m^{\rm ext}\,\Pi_{\rm phys}}
				{\sqrt{p_m}}
				\,|\Psi\rangle
				= \frac{M_m^{\rm ext}\,\Pi_{\rm phys}}{\sqrt{p_m}}
				\,|\Psi\rangle
				= |\Psi_m\rangle.
				$$		
				Hence $ |\Psi_m\rangle$ remains in the physical (same-charge) subspace.
				
				\item Entanglement collapse.
				In the unnormalized vector
				$\;M_m^{\rm ext}\,|\Psi\rangle = \sum_k \lambda_k\,|\chi_k\rangle\otimes E_m^{\rm ext}|e_k\rangle,$
				precisely those terms for which $E_m^{\rm ext}|e_k\rangle=0$ vanish.
				All surviving terms share the same external label (up to normalization), so any off-diagonal coherence between distinct $|e_k\rangle$ is gone.
			\end{itemize}
			
			\item Internal measurement.
			
			The argument is entirely parallel.
			For any internal POVM		
			$$
			M_n^{\rm int}
			= E_n^{\rm int}\,\otimes\,\mathbb I_{\rm ext},
			\quad
			\sum_n (M_n^{\rm int})^\dagger M_n^{\rm int} = \mathbb I_{\rm phys},
			\quad
			[\,E_n^{\rm int},\,U(g)^{\rm int}\,]=0,
			$$		
			we define		
			$$
			p_n
			= \langle\Psi| (M_n^{\rm int})^\dagger M_n^{\rm int} |\Psi\rangle
			= \sum_k \lambda_k^2\,\bigl\lVert E_n^{\rm int}\,|\chi_k\rangle\bigr\rVert^2,
			$$		
			and		
			$$
			|\Psi_n\rangle
			= \frac{M_n^{\rm int}\,|\Psi\rangle}{\sqrt{p_n}}
			= \frac{1}{\sqrt{p_n}}
			\sum_k \lambda_k\;
			E_n^{\rm int}\,|\chi_k\rangle_{\rm int}\,\otimes\,|e_k\rangle_{\rm ext}.
			$$		
			Then:
			\begin{itemize}
				\item Gauge invariance \& charge preservation follow by the same commutation argument with $\Pi_{\rm phys}$.
				
				\item Entanglement collapse now kills exactly those Schmidt terms for which $E_n^{\rm int}|\chi_k\rangle=0$, thus removing any coherence across the internal labels.
			\end{itemize}
		\end{enumerate}

		In both cases, because each $M$ commutes with the gauge projector, no component ever leaks into a non-singlet charge sector.
	\end{proof}

	This theorem shows that any physical measurement, whether probing external kinematics or internal packaging,
	acts entirely within the gauge‑singlet subspace, enforcing superselection and collapsing hybrid entanglement while preserving net charge.

	\begin{example}[Qubit-IQN Hybrid]
		A two‑particle state in $\mathcal H_{\rm phys}$,		
		$$
		|\Psi\rangle = \tfrac1{\sqrt2}(|\phi_+\rangle_{\rm int}\otimes|0\rangle_{\rm ext}
		+|\phi_-\rangle_{\rm int}\otimes|1\rangle_{\rm ext}),
		$$		
		with $|\phi_\pm\rangle_{\rm int}$ gauge singlets, is measured projectively on ${|0>,|1>}_{\rm ext}$:		
		$$
		M_i = I_{\rm int}\otimes|i\rangle\langle i|,
		\quad i=0,1.
		$$		
		The outcomes
		$M_0|\Psi\rangle=\tfrac1{\sqrt2}|\phi_+\rangle\otimes|0\rangle,
		\;M_1|\Psi\rangle=\tfrac1{\sqrt2}|\phi_-\rangle\otimes|1\rangle$
		preserve gauge‑singleticity and collapse the external-internal superposition.
	\end{example}

	\begin{example}[Colour Singlets in QCD]
		In the colour‑singlet meson		
		$$
		|\Psi\rangle=\tfrac1{\sqrt3}\sum_i|q_i\bar q^i\rangle\otimes|\sigma\rangle_{\rm spin},
		$$		
		a spin‑measurement
		$E_{\uparrow}=|\uparrow\rangle\langle \uparrow|$ on one quark yields
		$(E_{\uparrow}\otimes I_{\rm colour})|\Psi\rangle$,
		which remains in the $Q=0$ sector and collapses only the spin degree of freedom.
	\end{example}

	\begin{corollary}[Measurement-induced collapse of hybrid-packaged entangled state]
		Let
		$
		|\Psi\rangle\;\in\;\mathcal H_{\rm iso}
		$
		be a hybrid-packaged entangled state (Def. \ref{DEF:PackagedEntangledState}) which survives the gauge projection so that
		$
		|\Psi\rangle\;\in\;\mathcal H_{\rm phys}\subset\mathcal H_{\rm iso}.
		$
		Then:
		\begin{enumerate}
			\item Fully hybrid-packaged entangled states.
			If $|\Psi\rangle$ is fully hybrid-packaged entangled (that is, entangled both within the internal IQNs, within the external DOFs, and across internal $\leftrightarrow$ external), then any single-subsystem POVM measurement, whether acting on the IQNs or on the external label, completely collapses the joint wavefunction, projecting both domains into a definite post-measurement state in $\mathcal H_{\rm phys}$.
			
			\item Domain-separable hybrid-packaged entangled states.
			If instead $|\Psi\rangle$ factorizes into	
			$$
			|\Psi\rangle
			= 
			|\Psi_{\rm int}\rangle\;\otimes\;|\Psi_{\rm ext}\rangle
			\quad
			\in\;\mathcal H_{\rm int}\otimes\mathcal H_{\rm ext}
			=\mathcal H_{\rm phys},
			$$	
			then a POVM on the internal IQNs only collapses $|\Psi_{\rm int}\rangle$ (leaving $|\Psi_{\rm ext}\rangle$ untouched), and conversely a POVM on the external DOFs only collapses $|\Psi_{\rm ext}\rangle$.
		\end{enumerate}
	\end{corollary}

	\begin{proof}
		Let $\mathcal H_{\rm phys}=\mathcal H_{\rm int}\otimes\mathcal H_{\rm ext}$ denote the physical packaged subspace, and write the Schmidt decomposition of our pre-measurement state in that bipartition:
		$$
		|\Psi\rangle 
		\;=\; 
		\sum_{k}\lambda_k\;\bigl|\chi_k\bigr\rangle_{\rm int}\;\otimes\;\bigl|e_k\bigr\rangle_{\rm ext},
		\qquad
		\lambda_k>0,\quad
		\langle\chi_i|\chi_j\rangle
		=\langle e_i|e_j\rangle
		=\delta_{ij}.
		$$
		
		\begin{enumerate}
			\item Fully hybrid-packaged entangled case
			
			Let $\{E_m^{\rm ext}\}$ be a POVM on the external subsystem, so	
			$$
			M_m^{\rm ext}
			\;=\;
			\mathbb I_{\rm int}
			\;\otimes\;
			E_m^{\rm ext},
			\qquad
			\sum_m (M_m^{\rm ext})^\dagger M_m^{\rm ext}
			= \mathbb I_{\rm phys},
			$$	
			and $[E_m^{\rm ext},U(g)^{\rm ext}]=0$ so that $M_m^{\rm ext}$ commutes with the gauge projector.
			Then for each outcome $m$ with nonzero probability	
			$$
			p_m 
			\;=\; 
			\langle\Psi|\,(M_m^{\rm ext})^\dagger M_m^{\rm ext}\,|\Psi\rangle
			\;=\;
			\sum_k\lambda_k^2\,\bigl\lVert E_m^{\rm ext}|e_k\rangle\bigr\rVert^2
			\;>\;0,
			$$	
			the post-measurement (unnormalized) state is	
			$$
			M_m^{\rm ext}\,|\Psi\rangle
			\;=\;
			\sum_k\lambda_k\;\bigl|\chi_k\bigr\rangle_{\rm int}\;\otimes\;\bigl(E_m^{\rm ext}|e_k\rangle\bigr)_{\rm ext}.
			$$
			
			Any term with $E_m^{\rm ext}|e_k\rangle=0$ drops out.
			Renormalizing,	
			$$
			|\Psi_m\rangle
			\;=\;
			\frac{1}{\sqrt{p_m}}\;M_m^{\rm ext}\,|\Psi\rangle
			\;\in\;\mathcal H_{\rm phys}.
			$$
			\begin{itemize}
				\item Gauge-invariance \& charge preservation:
				Since $M_m^{\rm ext}$ commutes with the gauge projector, $ |\Psi_m\rangle$ remains in the same fixed-charge subspace of $\mathcal H_{\rm phys}$.
				
				\item Complete collapse of hybrid entanglement:
				Since a fully hybrid-packed state has all of its Schmidt coefficients nonzero, the selection $E_m^{\rm ext}$ kills every off-diagonal coherence in the external label \cite{Zurek1981}.
				The remaining state is a single term in the external basis (up to internal superposition among those $\chi_k$ whose $E_m^{\rm ext}|e_k\rangle\neq0$), that is, the hybrid entanglement is fully resolved into a definite external outcome and the associated internal projection.
			\end{itemize}
			
			An exactly analogous argument applies if one instead performs a POVM $\{E_n^{\rm int}\}$ on the internal subsystem:	
			$$
			M_n^{\rm int}
			=\;
			E_n^{\rm int}\,\otimes\,\mathbb I_{\rm ext},
			$$	
			which again commutes with the gauge projector and therefore preserves $\mathcal H_{\rm phys}$.
			Expanding in the same Schmidt basis,	
			$$
			M_n^{\rm int}\,|\Psi\rangle
			\;=\;
			\sum_k \lambda_k\,\bigl(E_n^{\rm int}|\chi_k\rangle\bigr)_{\rm int}\;\otimes\;|e_k\rangle_{\rm ext},
			$$	
			only those internal Schmidt terms survive, and normalization completes the collapse.  This simultaneously fixes an internal outcome and removes all coherence across external labels that were entangled with the eliminated internal terms.
			
			\item Domain-separable hybrid-packaged entangled case.
			
			If instead	
			$$
			|\Psi\rangle 
			= 
			|\Psi_{\rm int}\rangle 
			\;\otimes\;
			|\Psi_{\rm ext}\rangle,
			$$	
			then by definition there is no cross-Schmidt decomposition at all, each subsystem state is pure.
			Hence:
			\begin{itemize}
				\item A POVM on the internal subsystem, $\{E_n^{\rm int}\}$, acts as
				$\bigl(E_n^{\rm int}\otimes\mathbb I_{\rm ext}\bigr)$
				and collapses $|\Psi_{\rm int}\rangle$ alone, leaving $|\Psi_{\rm ext}\rangle$ unchanged.
				
				\item A POVM on the external subsystem,
				$\{E_m^{\rm ext}\}$, acts as
				$\bigl(\mathbb I_{\rm int}\otimes E_m^{\rm ext}\bigr)$
				and collapses $|\Psi_{\rm ext}\rangle$ alone, leaving $|\Psi_{\rm int}\rangle$ untouched.
			\end{itemize}
		\end{enumerate}
		In both cases, the measurement operators commute with the gauge projector.
		So the result remains in the same physical Hilbert space and net-charge sector.
	\end{proof}

	\begin{example}[POVM measurement on hybrid-packaged entangled state]
		Let us now reconsider Example~\ref{EXM:HybridPackagedEntangled} in physical packaged subspace where we can carry out a real measurement.
		\begin{enumerate}
			\item Fully hybrid-packaged entangled spin-momentum-charge pair:						
			\[
			\alpha \,\hat{a}_{e^-,\uparrow}^\dagger(\mathbf p)\,\hat{b}_{e^+,\downarrow}^\dagger(\mathbf{-p})\,\lvert0\rangle
			\;+\;
			\beta \,\hat{b}_{e^+,\downarrow}^\dagger(\mathbf{-p})\,\hat{a}_{e^-,\uparrow}^\dagger(\mathbf p)\,\lvert0\rangle.
			\]
			
			If we measure the first particle's charge and obtain $Q=-e$ or measure the first particle's spin and obtain spin $\uparrow$, then the state collapses to
			\[
			\hat{a}_{e^-,\uparrow}^\dagger(\mathbf p)\,\hat{b}_{e^+,\downarrow}^\dagger(\mathbf{-p})\,\lvert0\rangle.
			\]
			
			If we measure the first particle's charge and obtain $Q=e$ or measure the first particle's spin and obtain spin $\downarrow$, then the state collapses to
			\[
			\hat{b}_{e^+,\downarrow}^\dagger(\mathbf{-p})\,\hat{a}_{e^-,\uparrow}^\dagger(\mathbf p)\,\lvert0\rangle.
			\]
			
			In both cases, the net $Q=0$ sector is preserved, but the post-measurement outcome breaks the entanglement structure.
			
			\item Domain-separable hybrid-packaged entangled spin-momentum-charge pair:
			\[
			\left(\frac{1}{\sqrt{2}}\right)^3
			\Bigl(
			\lvert e^-\rangle_1\,\lvert e^+ \rangle_2 
			\;+\;
			\lvert e^+\rangle_1\,\lvert e^- \rangle_2
			\Bigr)
			\Bigl(
			\lvert \uparrow\rangle_1 \,\lvert \downarrow\rangle_2
			\;+\;
			\lvert \downarrow\rangle_1 \,\lvert \uparrow\rangle_2
			\Bigr)
			\Bigl(
			\lvert \mathbf p \rangle_1 \,\lvert \mathbf{-p} \rangle_2
			\;+\;
			\lvert \mathbf{-p} \rangle_1 \,\lvert \mathbf p \rangle_2
			\Bigr)
			\]
			If we measure the charge of first particle and obtain $Q=-e$, then the total wavefunction collapses to
			\[
			\frac{1}{2}
			\Bigl(
			\lvert e^-\rangle_1\,\lvert e^+ \rangle_2
			\Bigr)
			\Bigl(
			\lvert \uparrow\rangle_1 \,\lvert \downarrow\rangle_2
			\;+\;
			\lvert \downarrow\rangle_1 \,\lvert \uparrow\rangle_2
			\Bigr)
			\Bigl(
			\lvert \mathbf p \rangle_1 \,\lvert \mathbf{-p} \rangle_2
			\;+\;
			\lvert \mathbf{-p} \rangle_1 \,\lvert \mathbf p \rangle_2
			\Bigr)
			\]
			
			If we measure the spin of first particle and obtain spin $\uparrow$, then the total wavefunction collapses to
			\[
			\frac{1}{2}
			\Bigl(
			\lvert e^-\rangle_1\,\lvert e^+ \rangle_2 
			\;+\;
			\lvert e^+\rangle_1\,\lvert e^- \rangle_2
			\Bigr)
			\Bigl(
			\lvert \uparrow\rangle_1 \,\lvert \downarrow\rangle_2
			\Bigr)
			\Bigl(
			\lvert \mathbf p \rangle_1 \,\lvert \mathbf{-p} \rangle_2
			\;+\;
			\lvert \mathbf{-p} \rangle_1 \,\lvert \mathbf p \rangle_2
			\Bigr)
			\]
			
			If we measure the moment of first particle and obtain spin $\mathbf p$, then the total wavefunction collapses to
			\[
			\frac{1}{2}
			\Bigl(
			\lvert e^-\rangle_1\,\lvert e^+ \rangle_2 
			\;+\;
			\lvert e^+\rangle_1\,\lvert e^- \rangle_2
			\Bigr)
			\Bigl(
			\lvert \uparrow\rangle_1 \,\lvert \downarrow\rangle_2
			\;+\;
			\lvert \downarrow\rangle_1 \,\lvert \uparrow\rangle_2
			\Bigr)
			\Bigl(
			\lvert \mathbf p \rangle_1 \,\lvert \mathbf{-p} \rangle_2
			\Bigr)
			\]
			
			In all three cases, the net $Q=0$ sector is preserved, but the post-measurement outcome only breaks the entanglement structure of the measured part.
		\end{enumerate}
	\end{example}

	These results formalize how physical measurements act only in the gauge-singlet packaged subspace, collapsing hybrid wavefunctions while preserving the integrity of IQN packaging.

	\section{Discussion}

	We have examined the IQN packaging and find that:
	every local field operator carries all of its IQNs inside one irrep block of the gauge group $G$.
	Superselection rules \cite{Haag1996} and packaged entangled states are natural consequence of local gauge invariance.
	This bridges standard field-theoretic constraints with modern entanglement measures.

	Recent lattice-gauge and resource-theory studies \cite{Zohar2016,Sala2018} similarly show how gauge constraints shape the accessible state space. 
	The IQN packaging clarifies an important point:
	partial or fractional IQNs cannot appear as independent quantum DOFs in a single-particle state.  
	It further shows how gauge singlets (e.g., color singlet quark-antiquark states) can produce packaged entangled pairs within one net-charge sector.  
	Such states exemplify packaged entanglement:
	the wavefunction is packaged‑entangled while respecting local gauge‑invariance.

	The packaging principle can benefit:
	\begin{itemize}
		\item High‐energy collider phenomenology:
		symmetry packaging clarifies that partial local charges or partial color cannot exist, yet entangled color singlets \cite{Zweig1964,Hooft1974} are allowed \cite{Gross1973,Politzer1973}.
		Packaged entangled states may leave measurable signatures in multi‐boson or heavy‐flavor correlations.
		
		\item Foundations of superselection:
		IQN packaging clarifies how non‐local dressings and boundary conditions modify the spectrum of allowed superselection sectors.
		
		\item Gauge‐invariant quantum simulation:
		hybrid‐packaged qudits provide a robust embedding of internal plus external DOFs into near‐term quantum hardware, protected by gauge symmetry.
		The packaged quantum states may be used for gauge-invariant quantum simulations of lattice gauge theories \cite{Zohar2016,Sala2018}.
		
		\item Quantum computation and communication:	  
		The packaged quantum states could be applied to computation and quantum communication \cite{NielsenChuang}.
		Packaged entangled states may be used as robust quantum information carriers and in quantum error-correction protocols that inherently respect gauge symmetry \cite{Kitaev2003}.
		The conventional approaches usually treat entanglement independently of gauge considerations.
		Here we highlight that the very structure of a particle’s creation operator and thus the inseparability of its IQNs can serve as a powerful resource for quantum information processing.
		Building gauge-invariant error-correcting codes requires that each logical excitation be consistently packaged.
		This helps ensure that no unphysical states appear in the code space.
	\end{itemize}

	After all, we would like to emphasize that our discussion are mainly focused on finite and compact groups for simplicity and mathematical rigor.
	When coming to general groups that are neither finite nor compact, one must treat them more carefully with proper extensions.

	\section{Conclusion}

	We have established the Symmetry Packaging Principle for quantum field excitations and tracked its fate through every step of the excitation pipeline.
	Let us now summarize our findings as follows:

	\subsection*{(1) Packaging Principle}

	\begin{theorem}[Symmetry packaging principle]
		Every local creation or annihilation operator carries exactly one irrep $V_\lambda$ of the gauge group $G$.
		No physical process can split $V_\lambda$ into smaller pieces.
	\end{theorem}

	\begin{corollary}[Charge superselection]
		Since every local operator preserves the irrep label $\lambda$, coherent superpositions between different $\lambda$-sectors are forbidden.
		Physical states	decompose into disjoint charge sectors, reproducing the familiar rule that, e.g., one cannot observe superpositions of distinct electric charges.
	\end{corollary}

	\begin{corollary}[Packaged entanglement structure]
		Within any fixed $\lambda$ sector, the Schmidt decomposition must respect the irreducible block $V_\lambda$.
		This yields packaged entanglement in each of these sectors.
		Measuring the external DOFs of a fully hybrid-packaged entangled state collapses both its internal and external correlations, while preserving the net-charge sector.
	\end{corollary}

	\begin{corollary}[Local gauge-invariance confinement]
		The gauge projector $\Pi_{\rm phys} = \int_G dg U(g)$ annihilates every non-trivial irrep block, so only the trivial ($\lambda = 0$) sector	survives.
		In Abelian theories, it recovers the usual charge-neutrality condition.
		In non-Abelian gauge theories, this immediately enforces color singlet confinement.
	\end{corollary}

	These results together establish packaging as a fundamental organizing principle for gauge-invariant quantum systems.

	\subsection*{(2) Six Successive Stages (Persistence of packaging)}

	A six-stage pipeline through which every excitation passes:	
	
	\begin{table}[h]
		\centering
		\caption{Six stages of the excitation pipeline}
		\begin{tabular}[hbt!]{p{1cm}|p{8cm}|p{2.5cm}}
			\toprule
			Stage &Operation                                        & Map \\
			\midrule
			1     &Single-particle creation                         & $\emptyset \to \mathcal{H}_{\rm raw}$ \\
			2     &Coupling to gauge-blind (spectator) DOFs         & $\mathcal{H}_{\rm raw} \to \mathcal{H}_{\rm raw}$ \\
			3     &Many-body tensor-product assembly                & $\mathcal{H}_{\rm raw} \to \mathcal{H}_{\rm raw}$ \\
			4     &Isotypic projection into $\displaystyle \bigoplus_\lambda V_\lambda \otimes \mathbb{C}^{m \lambda}$ & $\mathcal{H}_{\rm raw} \to \mathcal{H}_{\rm iso}$ \\
			5     &Packaged superposition/entanglement              & $\mathcal{H}_{\rm iso} \to \mathcal{H}_{\rm iso}$ \\
			6     &Gauge projection onto singlets                   & $\mathcal{H}_{\rm iso} \to \mathcal{H}_{\rm phys}$ \\
			\bottomrule
		\end{tabular}
	\end{table}

	These six stages together form the complete quantum field excitation process and pipeline for constructing packaged states.

	\subsection*{(3) Three Layers (Layering of packaging)}

	We organize excitations into three packaging layers as follows (see Table \ref{tab:layer-structure}):

	\begin{table}[h]
		\centering
		\renewcommand{\arraystretch}{1.25}  
		\caption{Three packaging layers}
		\label{tab:layer-structure}
		\begin{tabular}{
				L{0.13\linewidth}|  
				L{0.21\linewidth}|  
				L{0.35\linewidth}|  
				L{0.28\linewidth}   
			}
			\toprule
			\textbf{Layer} &
			\textbf{Total Hilbert space} &
			\textbf{Subspaces} &
			\textbf{Defining principle} \\
			\midrule
			\textbf{Raw-Fock} &
			$\displaystyle
			\mathcal H_{\text{Fock}} = \bigoplus_{n \ge 0}
			\mathcal H^{(n)}_{\text{Fock}}
			$
			&\begin{itemize}\setlength\itemsep{2pt}
				\item single-particle package $V_{q} \otimes \mathcal H_{\text{ext}}$

				\item $n$-particle product $(V_{q_1} \otimes \cdots \otimes V_{q_n}) \otimes \mathcal H_{\text{ext}}^{(n)}$
			\end{itemize}
			&Single-particle packaging:	creation, annihilation,	multiplication \\
			\hline
			\textbf{Isotypic sector} &
			$\displaystyle
			\mathcal H_{\text{iso}}
			= \bigoplus_{\lambda \in \widehat G} 
			\mathcal H_\lambda
			$ &
			\begin{itemize}\setlength\itemsep{2pt}
				\item $n$-particle packaged $\bigl(V_\lambda \otimes M_\lambda^{(n)}\bigr) \otimes \mathcal H_{\text{ext}}^{(n)}$
			\end{itemize} &
			Multi-particle packaging: isotypic decomposition (sectorization), superselection, packaged entanglement \\
			\hline
			\textbf{Physical} &
			$\displaystyle
			\mathcal H_{\text{phys}}\subseteq\mathcal H_{\text{iso}}
			$ (often $\lambda=\mathbf 1$) &
			\begin{itemize}\setlength\itemsep{2pt}
				\item $n$-particle physical package $V_{\text{triv}} \otimes M_{0}^{(n)}$
			\end{itemize} &
			Gauge-invariant packaging: gauge projection
			$G(x) |\psi\rangle = 0$, enforcing gauge invariance $U_g |\psi\rangle = |\psi\rangle$ \\
			\bottomrule
		\end{tabular}
	\end{table}

	Single-particle packaging enforces that each excitation is a sharp $G$-irrep.
	Multi-particle packaging says that any multi-particle state lives in a single total-charge (isotypic)	block.
	Gauge-invariant packaging further restricts to those isotypic subspaces that are invariant under the local gauge group.
	These are the common kinematic root of the structure of physical Hilbert space in a gauge theory.
	It is independent of dynamics or specific interactions, and therefore universal across high-energy, condensed-matter, and quantum-information settings.

	\subsection*{(4) Applications \& Outlook}

	The symmetry packaging framework can be applied to many other areas:
	
	\begin{itemize}
		\item Model independence.
		The packaging relies only on the compactness of $G$ and locality of field operators.
		It applies equally to QED, QCD, lattice gauge models,
		topological phases, and quantum simulators.
		
		\item Information-theoretic utility.
		The packaged entanglement yields a symmetry-adapted Schmidt rank and entanglement entropy \cite{Casini2014}, and identifies each packaged irrep as a	noise-protected logical subsystem, ideal for encoding quantum information immune to gauge-variant errors.
		Thus, the packaged irreps behave as noise-protected logical qudits:
		gauge-invariant operations act only on spectator spaces, leaving the IQN block untouched.
	\end{itemize}

	These stages, layers, and physical consequences collectively establish packaging as a fundamental organizing principle.
	In Part II, we present a systematic classification of packaged states for concrete symmetry groups (including finite groups, compact groups, higher-form symmetry \cite{Gaiotto2015}
	and spacetime symmetries) and explore their quantum-information applications.

	\section*{Acknowledgments}

	\paragraph{Declaration of generative AI and AI-assisted technologies in the writing process.}
	During the preparation of this work the author used ChatGPT o3 in order to polish English language and refine technical syntax.
	After using this tool/service, the author reviewed and edited the content as needed and takes full responsibility for the content of the publication.

	\appendix

	\section{Conventions and Toolbox}

	\begin{center}
		\begin{minipage}{\linewidth}
			\hrule\smallskip
			\textbf{Conventions and Notation}
			
			\begin{itemize}
				\item \emph{Metric and indices}:
				We use flat-space metric $\eta_{\mu\nu}= \mathrm{diag}(+,-,-,-)$.
				Greek letters $\mu,\nu=0,1,2,3$ run over Lorentz indices and are summed when repeated.
				Four-vectors are plain ($p^\mu$).
				Spatial	three-vectors are bold ($\mathbf p$) with $|\mathbf p|=p$.
				
				\item \emph{Units}:
				We use natural units $\hbar=c=1$ throughout.
				
				\item \emph{Creation/annihilation operators}:
				All operators carry a hat ($\hat a,\hat a^\dagger,\hat b,\hat b^\dagger$).
				For a relativistic single-fermion field, we have
				\[
				\bigl\{\hat a_{s}(\mathbf p),
				\hat a^{\dagger}_{s'}(\mathbf p')\bigr\}
				=(2\pi)^{3}\,2E_{\mathbf p}\;
				\delta^{3}(\mathbf p-\mathbf p')\,
				\delta_{ss'} ,
				\qquad
				E_{\mathbf p}=\sqrt{\mathbf p^{2}+m^{2}} .
				\]
				(Replace $\{\,,\}$ with $[\, ,]$ for bosons.)
				
				\item \emph{Single-particle normalization}:
				$
				\langle \mathbf p',s'\lvert\mathbf p,s\rangle
				=(2\pi)^{3}\,2E_{\mathbf p}\;
				\delta^{3}(\mathbf p-\mathbf p')\,
				\delta_{ss'}$.
				
				\item \emph{Gauge indices}:
				Internal (color/flavor) indices are written $\alpha,\beta,\ldots$ and are suppressed when no confusion arises (e.g.\ 
				$\hat a^\dagger_{\alpha}\, D^{\alpha}{}_{\beta}(g)$).
				
				\item \emph{Commutator shorthand}:
				$[\mathcal O_{1},\mathcal O_{2}]_{\pm}
				=\mathcal O_{1}\mathcal O_{2}
				\mp\mathcal O_{2}\mathcal O_{1}$.
			\end{itemize}			
		\end{minipage}
	\end{center}

	\subsubsection*{Micro-causality}
	For any local (Heisenberg-picture) field operator $\psi(x)$, we assume
	\[
	[\psi(x),\psi^{\dagger}(y)]_{\pm} = 0
	\quad \text{whenever }(x-y)^{2} < 0
	\]
	to ensure relativistic locality in the sense of Doplicher-Haag-Roberts superselection theory.
	\hrule
	\end{document}